\documentclass[a4paper,onecolumn,11pt]{quantumarticle}
\pdfoutput=1

\usepackage[utf8]{inputenc}
\usepackage[T1]{fontenc}
\usepackage[english]{babel}

\usepackage{amsmath,amsfonts,amssymb,amsthm}
\usepackage{bbold,bbm,dsfont,stmaryrd}
\usepackage{braket}
\usepackage[mathscr]{euscript}
\usepackage{bm}
\usepackage{xfrac}
\usepackage{xspace}

\usepackage{graphicx}
\usepackage{xcolor}
\usepackage{makecell}
\usepackage{booktabs}
\usepackage{longtable}
\usepackage{array}
\usepackage{placeins}
\usepackage[labelfont=bf]{caption}
\usepackage{tikz}
\usepackage{pgfplots}
\pgfplotsset{compat=1.18}
\usepackage{quantikz}
\usetikzlibrary{shapes.geometric, positioning, decorations.pathreplacing, calc}
\usepackage{algorithm,float}
\usepackage{algpseudocode,eqparbox}

\usepackage{thmtools}
\usepackage{thm-restate}

\usepackage[numbers]{natbib}

\usepackage[colorlinks]{hyperref}

\tikzset{
  check/.style={draw,fill=blue!65,minimum size=9pt,rectangle,rounded corners=1.5pt},
  var/.style={draw,fill=yellow!85!orange,minimum size=9pt,circle},
  edge/.style={line width=.5pt,draw=gray},
  edge-normal/.style={edge},
  edge-strong/.style={edge,very thick,black},
}

\tikzset{
  check/.style={draw,fill=blue!65,minimum size=9pt,rectangle,rounded corners=1.5pt},
  var/.style={draw,fill=yellow!85!orange,minimum size=9pt,circle},
  edge/.style={line width=.5pt,draw=gray},
  edge-normal/.style={edge},
  edge-strong/.style={edge,very thick,black},
}

\tikzset{
  check/.style={draw,fill=blue!65,minimum size=9pt,rectangle,rounded corners=1.5pt},
  var/.style={draw,fill=yellow!85!orange,minimum size=9pt,circle},
  edge/.style={line width=.5pt,draw=gray},
  edge-normal/.style={edge},
  edge-strong/.style={edge,very thick,black},
}

\definecolor{DarkGreen}{rgb}{0.1,0.5,0.1}
\definecolor{DarkBlue}{rgb}{0.1,0.1,0.5}
\definecolor{NiceOrange}{rgb}{.9,0.55,0}
\definecolor{NicePurple}{rgb}{0.3,0.1,1}
\definecolor{DarkBlue}{rgb}{0.1,0.1,0.5}
\hypersetup{
    unicode=false,          
    pdftoolbar=true,        
    pdfmenubar=true,        
    pdffitwindow=false,      
    pdfnewwindow=true,      
    colorlinks=true,       
    linkcolor=DarkBlue,          
    citecolor=NicePurple,        
    filecolor=DarkGreen,      
    urlcolor=DarkBlue,          
    pdftitle={},
    pdfauthor={},
}

\DeclareMathOperator{\QFT}{QFT}

\newtheorem{theorem}{Theorem}

\newtheorem{lemma}[theorem]{Lemma}
\newtheorem{corollary}[theorem]{Corollary}
\newtheorem{definition}[theorem]{Definition}
\newtheorem{remark}{Remark}
\newtheorem*{theorem*}{Theorem}

\numberwithin{theorem}{section}
\numberwithin{conjecture}{section}
\numberwithin{lemma}{section}
\numberwithin{corollary}{section}
\numberwithin{definition}{section}
\numberwithin{remark}{section}

\numberwithin{question}{section}
\numberwithin{empirical}{section}

\newcommand{\COMMENT}[1]{}

\newif\ifnotes
\notestrue

\newcommand{\Tr}{\mathrm{Tr}}

\newcommand{\cC}{\mathcal{C}}

\newcommand{\cL}{\mathcal{L}}
\newcommand{\cM}{\mathcal{M}}

\newcommand{\cW}{\mathcal{W}}
\newcommand{\cX}{\mathcal{X}}

\newcommand{\ba}{\mathbf{a}}
\newcommand{\bb}{\mathbf{b}}
\newcommand{\bc}{\mathbf{c}}

\newcommand{\bs}{\mathbf{s}}

\newcommand{\bv}{\mathbf{v}}
\newcommand{\bw}{\mathbf{w}}
\newcommand{\bx}{\mathbf{x}}
\newcommand{\by}{\mathbf{y}}

\newcommand{\mI}{\mathbb{I}}
\newcommand{\mzero}{\mathbf{0}}
\newcommand{\mone}{\mathbf{1}}

\newcommand{\codesubspaces}{\mathcal{V}}
\newcommand{\subspace}{V}

\newcommand{\mbF}{\mathbb{F}}
\newcommand{\Image}{\text{Im}}
\newcommand{\reward}{\mathcal{R}}

\newcommand{\wC}{\widehat{\localcode}}

\newcommand{\fullpovm}{\{\Pi_{\subspace,\bs}\}_{\subspace\in \codesubspaces^{\localcode},\bs\in \Image(H_{\subspace})}\cup\{\Pi_{e}\}}

\newcommand{\subspaceperp}{\widehat{V}^{\perp}}
\newcommand{\fullpovmtilde}{\{\tilde{\Pi}_{\subspace,\bs}\}_{\subspace\in \codesubspaces^{\localcode},\bs\in \Image(H_{\subspace})}\cup\{\tilde{\Pi}_{e}\}}

\newcommand{\fullpovmstar}{\{\Pi_{\subspace,\bs}^*\}_{\subspace\in \codesubspaces^{\localcode},\bs\in \Image(H_{\subspace})}\cup\{\Pi_{e}^{*}\}}

\newcommand{\characterbasis}{\{\ket{\chi}\}_{\chi\in \wC}}
\newcommand{\dualcosetspace}{\wC/\subspaceperp}
\usepackage[most]{tcolorbox} 
\newcommand{\zetabasis}{\{\ket{\zeta_{\subspace,\bs}^{\subspaceperp_m}}\}_{\subspaceperp_m\in \dualcosetspace}}

\newcommand{\localcode}{\cC}
\newcommand{\chitriv}{\chi_{\text{triv}}}
\newcommand{\objsdp}{\mathrm{OPT}_{\mathrm{SDP}}}
\newcommand{\objlp}{\mathrm{OPT}_{\mathrm{LP}}}
\newcommand{\sdpfeasible}{\cM(\localcode,\{\ket{\psi_{\bc}}\}_{\bc\in \localcode})}
\newcommand{\dualvariable}{y}
\newcommand{\fsdp}{f_{\mathrm{SDP}}}
\newcommand{\flp}{f_{\mathrm{LP}}}

\newcommand{\lpfeasible}{\cL(\localcode,\{\ket{\psi_{\bc}}\}_{\bc\in \localcode})}
\newcommand{\fullbeta}{ \{\beta_{\subspace}(\subspaceperp_m)\}_{\subspace\in \codesubspaces^{\localcode},\subspaceperp_m\in\dualcosetspace }}
\newcommand{\fulltheta}{ \{\theta_{\subspace}(\subspaceperp_m)\}_{\subspace\in \codesubspaces^{\localcode},\subspaceperp_m\in\dualcosetspace }}

\newcommand{\fullbetastar}{ \{\beta_{\subspace}^*(\subspaceperp_m)\}_{\subspace\in \codesubspaces^{\localcode},\subspaceperp_m\in\dualcosetspace }}

\newcommand{\kappamax}{\kappa_{\text{max}}}
\newcommand{\kappamin}{\kappa_{\text{min}}}
\newtcolorbox{LPbox}[1][]{%
  breakable,
  colback=white,
  colframe=black,
  boxrule=0.8pt,
  arc=2pt,
  left=6pt,right=6pt,top=6pt,bottom=6pt,
  title=Final LP Form,
  fonttitle=\bfseries,
  #1
}

\newtcolorbox{LPdualbox}[1][]{%
  breakable,
  colback=white,
  colframe=black,
  boxrule=0.8pt,
  arc=2pt,
  left=6pt,right=6pt,top=6pt,bottom=6pt,
  title=Dual LP Form,
  fonttitle=\bfseries,
  #1
}

\title{Affine Filtering Measurements and Their Applications to Quantum Decoding}

\author[1,2,3]{Avijit Mandal}
\email{avijit.mandal@duke.edu}
\author[1]{Noah Shutty}
\author[2,3]{Henry D. Pfister}
\author[1]{Stephen P.~Jordan}

\affil[1]{Google Quantum AI, Venice, CA 90291}
\affil[2]{Department of Electrical and Computer Engineering, Duke University, Durham, NC 27708}
\affil[3]{Duke Quantum Center, Durham, NC 27708}
\date{}
\begin{document}
\maketitle
\begin{abstract}

Unambiguous state discrimination (USD) measurements are attractive because outcomes are either marked as conclusive (i.e., error free) or inconclusive (i.e., erased).
We study affine filtering measurements, a structured variant of USD for decoding classical linear codes over pure-state classical-quantum channels, where a conclusive outcome identifies an affine subspace containing the transmitted codeword and an inconclusive outcome is treated as an erasure.
For a group-covariant indexing of pure-state codewords, we show that the optimal design of affine filtering measurements is a semidefinite program that can be reduced to a linear program via character-based diagonalization.
We use the resulting measurement to build a quantum decoding framework for local codes, and we demonstrate (via simulations on regular LDPC codes from Gallager ensembles using single parity check local constraints) that affine filtering based decoding can outperform symbol-wise USD and symbol-wise pretty good measurement based decoding methods on i.i.d. pure-state channels.
In an independent and concurrent work, Buzet and Chailloux study similar fine-grained USD measurements for symmetric families of states. Their focus is on the code-agnostic setting whereas our focus is on code-aware constructions and decoding.
  
\end{abstract}

\section{Introduction}
Coding theory for classical–quantum (CQ) channels traces back to the foundational work by Holevo~\cite{holevo1998capacity} and by Schumacher and Westmoreland~\cite{schumacher1997sending}, who established that one can communicate reliably at any rate up to the channel capacity by performing joint (collective) measurements across long codeword blocks. For many practically important structured CQ channels such as pure-state signal constellations that arise in optical communications~\cite{krovi2015optimal,da2013achieving}, the main hurdle in practice is
the implementation of these large
collective measurements. This challenge has spurred the development of structured receiver designs whose complexity grows slowly
with blocklength while still offering provable performance. 

Recent progress on quantum algorithms based on Regev's reduction has renewed interest in designing efficient algorithms for the quantum decoding problem~\cite{chailloux2023quantum, chailloux2024quantum,chen2022quantum,yamakawa2024verifiable}.
 In the same spirit, a recently introduced method called decoded quantum interferometry (DQI)~\cite{jordan2025optimization} uses reversible decoding together with quantum Fourier sampling to transform certain structured optimization problems into decoding problems.
In \cite{shutty2026lqd}, the authors develop locally quantum decoders for LDPC codes and use them, through the Regev and DQI framework, to obtain improved approximate solutions for average-case \(D\)-regular max-\(k\)-XORSAT problems, 
further motivating the study of efficient quantum decoding algorithms.

Several works have investigated the performance of classical codes, such as polar and Reed-Muller codes~\cite{wilde2012polar, wilde2013towards,mandal2025reed}, on classical quantum channels in the context of achieving the Holevo bound. Although these works studied the capacity-achieving potential of such codes, they did not provide efficient constructions to design practical quantum decoding methods.

In pursuit of designing low complexity quantum decoders, Belief propagation with quantum messages (BPQM) was introduced by Renes~\cite{Renes-njp17} as a quantum analogue of classical BP for decoding binary linear codes on pure-state channels (PSCs).
Subsequent works developed BPQM-based density-evolution (DE) for symmetric binary CQ channels and used them to study LDPC and polar constructions~\cite{brandsen2022belief,mandal2023belief,mandal2024polar}, with further extensions to turbo-style constructions~\cite{piveteau2025efficient}.  Recently, BPQM has been extended to codes with prime alphabet sizes and abelian groups with more general factor graphs in \cite{mandal2026belief, mandal2026qmp}. Although efficient constructions of BPQM-based decoders have been provided for polar codes \cite{mandal2024polar}, it is still not clear how to implement an efficient BPQM-based decoder for LDPC codes due to the sequential nature of message update rules.  

For a collection of geometrically uniform states with equal priors, the pretty good measurement (PGM) \cite{hausladen1994pretty,holevo1978asymptotically} has been proven to minimize the error probability of state discrimination
\cite{eldar2002quantum,eldar2003geometrically}.
When decoding codes on PSCs, one approach is to apply the
PGM to discriminate between the quantum states associated with individual code symbols
and then model the measurement outcome
as a
classical q-ary channel. 
For classical post-processing, one can
use a belief-propagation(BP)-based decoder to estimate the transmitted codeword.
While implementing PGM to decode codeword $\bc$ of sufficiently large code $\localcode$ is extremely complex, it is possible to apply PGM
to decode smaller local projections
of a larger code. However, this introduces correlated errors, making the analysis of classical post-processing difficult. 

For a collection of quantum states, optimal USD
measurements
minimize the probability of an inconclusive outcome~\cite{ivanovic1987differentiate,peres1988differentiate,dieks1988overlap, chefles1998optimum}. In the context of decoding codes on a PSC, when we apply the optimal USD measurement to quantum states associated with individual code symbols, the resultant measurement outcome can be modelled as erasure errors.
For linear codes, we then
use a Gaussian-elimination (GE)-based decoder to estimate the transmitted codeword.

 In this work, we study affine filtering measurements, a code-aware filtering
form of USD. 
The terminology filtering is motivated by the filtering framework of
\cite{chen2022quantum}, where conclusive quantum measurement outcomes provide
partial information about a hidden variable and inconclusive outcomes are
discarded. Our formulation is specific to linear codes where the conclusive information takes the form of affine subspaces containing the transmitted codeword.
A similar class of measurements was studied recently in
\cite{buzet2025fine}, where the authors describe them as fine-grained USD
measurements (FGUM). 
While their work does not leverage code structure, our
affine filtering measurements are designed to do so. The affine-subspace
outcomes are converted into linear constraints, which we then use to build a
quantum decoding framework for linear codes on pure-state channels.
More specifically, given a collection of
symmetric quantum states indexed by codewords transmitted over a PSC, an affine
filtering measurement outputs an affine subspace of the code $\localcode$ that
contains the transmitted codeword, or outputs an inconclusive outcome. The goal
is to maximize a chosen objective, such as the expected number of recovered
linear equations. 
Constructing the optimal affine filtering measurement is a
semidefinite program (SDP).
 However, even for a code with only a few symbols, the number of affine subspaces can be exponentially large.
 This makes finding the optimal measurement extremely challenging.

Our primary contributions are the simplification of the SDP and the analysis of the implied decoder. We
 prove that, for a collection of symmetric quantum states, the aforementioned SDP
 can be simplified to
 a linear program (LP). 
 This makes the optimization
 much more tractable
and, using the solution of the LP,
 we provide an
 explicit form of the optimal measurement.
 In \cite{shutty2026lqd}, the authors use FGUM as a locally quantum decoding method for the quantum decoding problem arising from average-case \(D\)-regular max-\(k\)-XORSAT.  These measurements are applied to a specific family of local quantum states obtained from a special blockwise choice of the state-preparation function.  For each degree-\(D\) single parity check (SPC) code \(C_I\) supported on an index set \(I\), this gives a family of states indexed by the local codewords \(\bc\in C_I\).  The FGUM either identifies \(\bc\) exactly or returns a failure outcome.  The failed local measurements are then treated as erasures and combined with the global linear constraints by GE.  In this paper, we show that, for the same SPC-associated choice of quantum states, the FGUM is an optimal 
 affine filtering measurement for maximizing the expected number of recovered linear equations.

 
 Next, we
 design a new class of quantum decoding methods
 for a global
 code $\localcode$
 whose code symbols can be partitioned into a collection of subsets $\{I_{j}\}$ such that each subset is associated with a local code $\localcode_{I_j}$. The idea is to apply the optimal 
 affine filtering measurement 
to the states $\{\ket{\psi_{\bc}}\}_{\bc\in \localcode_{I_j}}$.
For each index set $I_j$, the measurement outcome is either inconclusive or an affine subspace containing the correct local codeword.
The objective function is the expected number of equations generated by the 
affine filtering decoding of each local code.
After running the above procedure for each index set, we try to recover the codeword using GE. Based on this procedure, we simulate an 
affine filtering measurement decoder
that applies the measurement to the qudit states associated with SPC local codes in
a $(k, D)$-regular LDPC code. We observe that, for most $(k,D)$ pairs, when we fix the channel to be an i.i.d.~PSC, the 
affine-filtering+GE-based decoder outperforms single qudit-wise measurement based decoders or more specifically, qudit-PGM+BP decoder and the qudit-USD+GE decoder.
 

 In Section~\ref{sec:affineUSD}, we give a complete description of the 
 affine filtering measurement and derive its SDP to LP reduction.  We also describe the limiting LP formulation obtained as the limit of full-rank perturbations when the states are linearly dependent.  In Section~\ref{sec:affine spc decoder}, we specialize the local code to a SPC code and describe the 
 affine filtering+GE decoding algorithm for the $(k,D)$ Gallager ensemble of regular LDPC codes~\cite{gallager1962low,richardson2008modern,RU01} over i.i.d.\ PSCs.  In Section~\ref{sec:numerical results}, we present numerical results for the resulting decoder performance.
\section{Background}\label{sec: quantum  decoding intro}
\subsection{Notation}
The natural numbers are denoted by $\mathbb{N}=\left\{ 1,2,\ldots\right\} $ and  $\mathbb{N}_0=\mathbb{N}\cup \{0\}$. For $m\in \mathbb{N}_0$,  we use the shorthand $[m]\coloneqq\left\{ 0,\ldots,m-1\right\} $. Throughout the rest of the paper, we assume alphabet size $q$ to be prime. The Galois field with $q$ elements (i.e., the
integers $[q]$ with addition and multiplication modulo $q$) is denoted by $\mathbb{F}_{q}$.
We use boldface notation to denote vectors (e.g. $\by,\bw$).
We use $\mathbf{0}$ to denote all-zero vector or matrix if it is clear from the context. The identity matrix is denoted by $\mI$. 
The canonical basis is represented using $\{\ket{j}\}_{j\in \mathbb{F}_q}$.
For a vector $\bc$ and a $D$-element index set $I=\{I_{0},\dots, I_{D-1}\}\subseteq [N]$, we define the subvector $\bc_{I}=\{c_{I_0},\dots ,c_{I_{D-1}}\}$.
For a code $\localcode$ of length $N$ and a subset $I\subseteq [N]$, we denote by $\localcode_{I}$ the punctured code formed by keeping only the symbol positions in $I$.
For example, when $\localcode_{I}$ is a  single-parity check (SPC) on index set $I$, it satisfies 
\[
\localcode_{I}=\{\bc_{I}\in\mathbb{F}_q^{|I|} : \mathbf{1}_I^T \bc_{I} = 0 \}.
\]
Define $\omega\coloneqq e^{2\pi\mathrm{i}/q}$ to be a primitive $q$\textsuperscript{th} root of unity.
We denote the generalized Pauli shift and phase operators by $X$ and $Z$, respectively. They satisfy
\begin{align*}
    X^a\ket{j} & =\ket{j+a},\\
    Z^a\ket{j} & =\omega^{aj}\ket{j},
\end{align*}
for all $a,j \in \mathbb{F}_q$. For a vector $\ba\in \mathbb{F}_{q}^m$, we denote the operators as $X^{\ba}$ and $Z^{\ba}$ such that 
\begin{align*}
    X^{\ba}\coloneqq \bigotimes_{i=0}^{m-1}X^{a_i}, \quad Z^{\ba}\coloneqq \bigotimes_{i=0}^{m-1}Z^{a_i}.
\end{align*}
The unitary $\QFT$ corresponds to the quantum Fourier transform which satisfies 
\begin{align*}
    \QFT \ket{j}=\frac{1}{\sqrt{q}}\sum_{k\in \mathbb{F}_q}\omega^{jk}\ket{k},
\end{align*}
for all $j\in \mathbb{F}_q$. The unitary $\QFT$ satisfies the following relations 
\begin{align*}
    \QFT X^a\QFT^{\dagger}= Z^a, \quad \QFT^{\dagger} X^a\QFT= Z^{-a}.
\end{align*}
\subsection{Preliminaries}
\begin{definition}
    A classical-quantum (CQ) channel $W\colon x\rightarrow \rho_{x}$ takes a classical input $x$ from a finite size alphabet $\cX$ and outputs a density matrix $\rho_{x}$. 
\end{definition}
\begin{definition}
      A CQ channel $W\colon x\rightarrow \rho_{x}$ is called a pure-state channel (PSC), if the output state $\rho_{x}=\ket{\psi_x}\bra{\psi_x}$ is a pure state, i.e. a rank-1 matrix $\forall x\in \cX$.
\end{definition}

\begin{definition}
    The pretty good measurement (PGM), which is also known as the square root measurement, is defined as follows: consider the CQ channel $W\colon x\rightarrow \rho_x$ with a uniform input distribution on $\cX$. For the output state $\rho_x$ for $x\in \cX$, the measurement operator is $M_{x}=\frac{1}{|\cX|}\bar{\rho}^{-\frac{1}{2}}\rho_x\bar{\rho}^{-\frac{1}{2}}$ where $\bar{\rho}=\sum_{x\in \cX}\frac{1}{|\cX|}\rho_x$ and the inverse is taken on the support of $\bar{\rho}$.
\end{definition}
\begin{definition}
    Consider the PSC $W\colon x\rightarrow \ket{\psi_x}\bra{\psi_x}$, with $x\in \cX$. The unambiguous state discrimination (USD) measurement  constructs the measurement $\{\Pi_{x}\}_{x\in \cX}\cup\{\Pi_e\}$ such that 
    the measurement operators satisfy the unambiguity criterion $\Tr(\Pi_x\ket{\psi_{x'}}\bra{\psi_{x'}})=0$ $\forall x\neq x'$ and $\sum_{x\in \cX}\Pi_{x}+\Pi_{e}=\mI$. The optimal USD measurement minimizes the average probability of obtaining the inconclusive outcome of the measurement, i.e., $\Pi_e$. For the uniform input distribution on $\cX$, the problem is defined by an SDP below:
    \begin{align}
 \text{Minimize}  \quad &\   \frac{1}{|\cX|}\sum_{x\in \cX}\Tr(\Pi_e\ket{\psi_x}\bra{\psi_x})\\
 \text{Subject to:} \quad &\ \Tr(\Pi_x\ket{\psi_{x'}}\bra{\psi_{x'}})=0, \forall x\neq x'\nonumber \\
    &\ \sum_{x\in \cX}\Pi_{x}+\Pi_{e}=\mI \nonumber\\
    &\ \Pi_{x}\succeq 0, \forall x\in \cX, \Pi_{e}\succeq 0.        \end{align}
\end{definition}
\noindent We will also need the following specialization of Definition~2.4 to the collection of codeword states associated with a local code.

\begin{definition}\label{def:global USD}
Consider the collection of quantum states $\{\ket{\psi_{\bc}}\}_{\bc\in \localcode_I}$ associated with the local code $\localcode_I$. For code $\localcode_{I}$, the 
{\it codeword filtering measurement } is the POVM
\begin{align*}
    \left\{\Pi_{\bc}^{\mathrm{Cw\mbox{-}Filt}}\right\}_{\bc\in \localcode_I}\cup\left\{\Pi_{e}^{\mathrm{Cw\mbox{-}Filt}}\right\}
\end{align*}
such that the measurement operators satisfy
\begin{align*}
    \Tr\!\left(\Pi_{\bc}^{\mathrm{Cw\mbox{-}Filt}}
    \ket{\psi_{\bc'}}\bra{\psi_{\bc'}}\right)=0,
    \qquad \forall \bc\neq \bc',\ \bc,\bc'\in \localcode_I,
\end{align*}
and
\begin{align*}
    \sum_{\bc\in \localcode_I}\Pi_{\bc}^{\mathrm{Cw\mbox{-}Filt}}
    +\Pi_{e}^{\mathrm{Cw\mbox{-}Filt}}=\mI.
\end{align*}
Obtaining the outcome $\Pi_{\bc}^{\mathrm{Cw\mbox{-}Filt}}$ corresponds to identifying or filtering the transmitted codeword $\bc$ exactly, while $\Pi_{e}^{\mathrm{Cw\mbox{-}Filt}}$ corresponds to an inconclusive outcome.
\end{definition}

To characterize the CQ channel $W$, we use the Holevo information as an information measure. We use the symmetric Holevo information to parameterize and compare the
PSCs used in the numerical experiments.
Symmetric Holevo information measures the maximum amount of classical information that can be sent through the channel $W$ when the input distribution is uniform.
\begin{definition}
    For a CQ channel $W\colon x\rightarrow \rho_x$ with $x\in \cX$, the symmetric Holevo information is computed as 
\begin{align*}
    I(W)=S\Bigg(\frac{1}{|\cX|}\sum_{x\in \cX}\rho_x\Bigg)-\sum_{x\in \cX}\frac{1}{|\cX|}S\Bigg(\rho_x\Bigg)
\end{align*}
where $S(\rho)=-\Tr(\rho\log\rho)$ is the von Neumann entropy for the density matrix $\rho$.
\end{definition}

The 
affine filtering measurements studied in the next section are defined for
codeword-indexed families of symmetric pure states.  We now introduce the
particular symmetric form used throughout the paper.
Although the analysis in this paper extends to arbitrary abelian groups, and more generally to abelian group codes, we restrict the presentation in this paper to 
affine filtering measurements for codes over $\mathbb{F}_q$, where $q$ is prime. To analyze these measurements, we define the canonical quantum state of the $|I|$ qudits $\ket{\psi}$ as follows. 
    \begin{align*}
        \ket{\psi}\coloneqq \sum_{\bx\in \mathbb{F}_{q}^{|I|}}\alpha_{\bx}\ket{\bx}.
    \end{align*}
where $\{\alpha_{\bx}\}_{\bx\in \mathbb{F}_{q}^{|I|}}$ are arbitrary complex coefficients that satisfy $\sum_{\bx\in \mathbb{F}_{q}^{|I|}}|\alpha_{\bx}|^2=1$.
Next, consider the classical-quantum channel $W:\bc\rightarrow \ket{\psi_{\bc}}$, that takes a $|I|$ symbol codeword as input and maps to the quantum state $\ket{\psi_{\bc}}$  
  \begin{align}\label{def: quantum state}
      \ket{\psi_{\bc}} =X^{\bc}\QFT^{\otimes |I|}\ket{\psi}.
  \end{align}
  This provides a canonical choice of quantum states
$\{\ket{\psi_{\bc}}\}_{\bc\in \localcode}$ for the analysis. Since the
coefficients $\{\alpha_{\bx}\}_{\bx\in \mathbb{F}_{q}^{|I|}}$ are arbitrary,
the arguments below apply to any codeword-indexed symmetric family of states
with a circulant Gram matrix.

Before describing the 
affine filtering measurement, we must define a few notations necessary for affine subspaces. 
Let $\codesubspaces^{\localcode}\coloneqq\{\subspace\subsetneq \localcode: \subspace \text{ is a linear subspace}\}$ denote the collection of proper linear subspaces for the code $\localcode$. 
For each subspace $\subspace\in\codesubspaces^{\localcode}$,
define a fixed linear map 
$H_{\subspace}:\localcode\rightarrow \mathbb{F}_{q}^{\dim(\localcode)-\dim(\subspace)}$ such that $\ker(H_{\subspace})=\subspace$. 
We denote $H_{\subspace}$ as  the parity check matrix  for the subspace $\subspace$ restricted to the local code $\localcode$ with rank$(H_{\subspace})=\dim(\localcode)-\dim(\subspace)$. 
For the parity check matrix $H_{\subspace}$, we define the image 
$\Image(H_{\subspace})\coloneqq \{H_{\subspace}\bc:\bc\in\localcode\}$. 
Therefore, 
$|\Image(H_{\subspace})|=q^{\dim(\localcode)-\dim(\subspace)}$.
For a syndrome $\bs\in \Image(H_{\subspace})$ and the subspace $\subspace$, define $\subspace_{\bs}\coloneqq\{\bc\in \localcode:H_{\subspace}\bc=\bs\}$, the affine subspace of $\localcode$ consisting of all codewords whose syndrome
with respect to $H_{\subspace}$ is $\bs$. In this notation,
$\subspace_{\mzero}=\subspace$ for the all-zero syndrome vector $\mzero\in \Image(H_{\subspace})$.

Our analysis is heavily based on group characters. For a local code $\localcode$ , 
define the dual group $\wC$ to be the set of characters $\chi$, which are group homomorphisms
$\chi :\localcode \to \mathbb{T}$ where $\mathbb{T}=\{z\in \mathbb{C}\colon |z|=1\}$ is the unit circle under multiplication~\cite{luong2009fourier}. It follows that 
\begin{align*}
    \chi(\bc_{1}+\bc_{2})=\chi(\bc_1)\chi(\bc_2), \forall \bc_{1},\bc_{2}\in \localcode,\forall\chi\in\wC.
\end{align*}
Since $\localcode$ is finite, each character $\chi\in\wC$ satisfies $|\chi(\bc)|=1$ for all  $\bc\in \localcode$. The group $\wC$ is equipped with trivial character $\chitriv$ such that $\chitriv(\bc)=1$ $\forall\bc\in \localcode$.
For two characters $\chi_1, \chi_2$, we define $\chi_1\chi_2$ as the character such that
\begin{align*}
   (\chi_{1}\chi_{2})(\bc) := \chi_{1}(\bc)\chi_{2}(\bc).
\end{align*}
One can readily check that $\chi_1\chi_2\in \wC$. For character $\chi$, let $\chi^{-1}$ denote its inverse which satisfies
\begin{align*}
    \chi\chi^{-1}=\chitriv.
\end{align*}
For a character $\chi\in\wC$, the inverse character $\chi^{-1}\in \wC$ is unique and satisfies the following relations $\forall\bc\in \localcode$
\begin{align*}
    \chi^{-1}(\bc)=\chi(-\bc)=\frac{1}{\chi(\bc)}=\overline{\chi(\bc)}
\end{align*}

\noindent The character group $\wC$ satisfies $|\wC|=|\localcode|$ since $\localcode$ is a finite abelian group. For a subspace $\subspace\in \codesubspaces^{\localcode}$, define its annihilator as
\begin{align*}
    \subspaceperp\coloneqq\{\chi\in\wC:\chi(\bv)=1, \quad\forall\bv\in \subspace\}.
\end{align*}
Then $\subspaceperp$ is a subgroup of $\wC$, with
$|\subspaceperp|=q^{\dim(\localcode)-\dim(\subspace)}$. Hence $\wC/\subspaceperp$ has
$q^{\dim(\subspace)}$ cosets, which we denote by $\subspaceperp_m$ for
$m\in [q^{\dim(\subspace)}]$.
The following standard character identities will be used repeatedly throughout the rest of the paper.

\begin{lemma}[Character identities for finite abelian groups]
\label{lem:character-identities}
Let $\localcode$ be a finite abelian group with character group $\wC$, and let a subspace $\subspace\in \codesubspaces^{\localcode}$
have annihilator $\subspaceperp$. Then the following hold.
\begin{enumerate}
    \item For any $\chi_1,\chi_2\in\wC$,
    \begin{align*}
        \sum_{\bc\in\localcode}\chi_{1}^{-1}(\bc)\chi_{2}(\bc)
        =
        \sum_{\bc\in\localcode}\overline{\chi_{1}(\bc)}\chi_{2}(\bc)
        =
        |\localcode|\delta_{\chi_{1},\chi_{2}}.
    \end{align*}

    \item For any $\bv\in\localcode$,
    \begin{align*}
        \sum_{\chi\in\subspaceperp}\chi(\bv)=
        \begin{cases}
            |\subspaceperp|, &\text{if $\bv\in\subspace$,}\\
            0, &\text{otherwise.}
        \end{cases}
    \end{align*}

    \item If $\chi_1\in\subspaceperp_{m_1}$ and
    $\chi_2\in\subspaceperp_{m_2}$, where
    $\subspaceperp_{m_1},\subspaceperp_{m_2}\in\wC/\subspaceperp$, then
    \begin{align*}
        \chi_2\chi_1^{-1}\in\subspaceperp
        \iff
        m_1=m_2.
    \end{align*}
\end{enumerate}
\end{lemma}
These identities follow from the standard orthogonality relations for characters of finite abelian groups~\cite{luong2009fourier}.

\section{Affine Filtering Measurements}\label{sec:affineUSD}

In this section, we study a structured class of unambiguous measurements for codeword-indexed symmetric state families.
Instead of attempting to identify the transmitted codeword exactly, these measurements return an affine subspace of the local code $\localcode_I$ that is guaranteed to contain the transmitted codeword, or else return an erasure outcome. This leads naturally to an optimization problem in which different affine subspaces may be assigned different rewards according to how much linear information they reveal.

Our goal is to characterize the optimal affine filtering measurement for a local code $\localcode$ and its associated symmetric state family $\{\ket{\psi_{\bc}}\}_{\bc\in \localcode_I}$. 
In the remainder of section, we write
$\localcode$ for the local code $\localcode_I$ to simplify notation.

\begin{definition}\label{def:affineUSD}
    An {\it affine filtering measurement} for a local code $\localcode$ is a joint generalized measurement (i.e. POVM) on the quantum states $\{\ket{\psi_{\bc}}\}_{\bc\in \localcode}$, 
    with one measurement outcome $\Pi_{\subspace,\bs}$ per affine subspace
$\subspace_{\bs}$ where
$\subspace\in \codesubspaces^{\localcode}$ and
$\bs\in \Image(H_{\subspace})$.
    We denote this POVM by $\{\Pi_{\subspace,\bs}\}_{\subspace\in \codesubspaces^{\localcode},\bs\in \Image(H_{\subspace})}\cup\{\Pi_{e}\}$. 
 For the unambiguity or filtering validity criterion, an 
 affine filtering measurement must satisfy
    \begin{align*}
            \Tr\left(\Pi_{\subspace,\bs}\ket{\psi_{\mathbf{c}}}\bra{\psi_{\mathbf{c}}}\right)=0 \quad \text{if $\bc\notin \subspace_{\bs}$}
    \end{align*}
    for all affine subspaces $\subspace_{\bs}\subset \localcode$ with $\subspace\in \codesubspaces^{\localcode}$ and $\bs\in \Image(H_{\subspace})$.
    Since each $\subspace_{\bs}$ is a coset of $\subspace$, the dimensions satisfy $\dim (\localcode) - \dim(\subspace_{\bs})=\dim (\localcode) - \dim(\subspace)$. Therefore, obtaining a measurement outcome $\Pi_{\subspace,\bs}$ corresponds to learning $\dim \localcode - \dim(\subspace)$ linear equations (i.e. parities).  The measurement operator $\Pi_{e}$ corresponds to the erasure outcome which provides no information about $\mathbf{c}$. 
\end{definition}
A conclusive affine filtering outcome need not identify the transmitted codeword uniquely; it only needs to identify an affine subspace of $\localcode$ that contains it. This makes it natural to optimize over reward functions that depend on how much linear information the outcome reveals.
For each subspace $\subspace\in\codesubspaces^{\localcode}$, let $\reward(\subspace)\geq 0$ denote the reward function representing the contribution of $\subspace$ to the maximization objective. We assume the reward depends only on the linear subspace $\subspace$ itself, such that every affine subspace $\subspace_{\bs}$ with $\bs\in\Image(H_{\subspace})$
 yields the same reward $\reward(\subspace)$. When we set $\reward(\subspace)=\dim(\localcode)-\dim(\subspace)$, the problem becomes maximizing the expected number of linear equations.
Also, observe that 
codeword filtering is a special case of affine filtering in which every conclusive outcome corresponds to a singleton affine subspace, i.e., $\subspace=\{\mzero\}$, and hence $\subspace_{\bs}=\{\bc\}$ for some $\bc\in \localcode$.
Suppose we wish to find an 
affine filtering measurement maximizing the expected reward based on the reward function $\reward(.)$. This can be formulated as an SDP:
     \begin{align}{\label{SDP main}}
        \text{Maximize}  \quad &\   \frac{1}{|\localcode|}\sum_{\bc\in \localcode}\sum_{\subspace\in \codesubspaces^{\localcode},\bs\in\Image(H_{\subspace})}\Tr(\Pi_{\subspace,\bs}\ket{\psi_{\bc}}\bra{\psi_{\bc}})\reward(\subspace)\\
 \text{Subject to:} \quad &\ \Tr\left(\Pi_{\subspace,\bs}\ket{\psi_{\bc}}\bra{\psi_{\bc}}\right)=0 \quad \text{if $\mathbf{c}\notin \subspace_{\bs}$} \nonumber \\
    &\ \sum_{\subspace\in \codesubspaces^{\localcode},\bs\in\Image(H_{\subspace})}\Pi_{\subspace,\bs}+\Pi_e=\mI\\
    &\ \Pi_{\subspace,\bs}\succeq 0, \forall \subspace\in \codesubspaces^{\localcode},\bs\in\Image(H_{\subspace})\nonumber\\&\ \Pi_{e}\succeq 0.\nonumber
    \end{align}
\noindent
For later use, and to state the symmetry reductions without repeatedly
rewriting the full SDP, we separate the objective from the feasible region.
For a POVM $\fullpovm$, define
\begin{align}
    \fsdp(\fullpovm)
    \coloneqq
    \frac{1}{|\localcode|}
    \sum_{\bc\in \localcode}
    \sum_{\subspace\in \codesubspaces^{\localcode},\,\bs\in\Image(H_{\subspace})}
    \Tr(\Pi_{\subspace,\bs}\ket{\psi_{\bc}}\bra{\psi_{\bc}})
    \reward(\subspace).
\end{align}
Let \(\sdpfeasible\) denote the feasible region of \eqref{SDP main}, i.e.,
the set of 
affine filtering POVMs satisfying the unambiguity, completeness, and
positivity constraints displayed above. With this notation, the optimal SDP
value is
\begin{align}
   & \objsdp(\sdpfeasible)\\
 &   \coloneqq
    \text{Max}_{\fullpovm\in\sdpfeasible} \fsdp(\fullpovm).\nonumber
\end{align}
The SDP above is still high-dimensional, since it involves one POVM element
for each affine subspace of \(\localcode\). The rest of this section reduces
this SDP to an equivalent LP by exploiting the symmetries of the state family.
First, we show that any feasible 
affine filtering POVM can be symmetrized over
translations in \(\localcode\), without changing the value of
\(\fsdp\). This allows us to restrict attention to POVMs for which affine
cosets of the same subspace are related by certain translation operators
\(X^{\bw}\) with $\bw\in \localcode$.

After this symmetrization, we aggregate the conclusive outcomes associated
with a fixed subspace \(\subspace\) into the operator
\(\Pi_{\subspace}\). The symmetry implies that \(\Pi_{\subspace}\) commutes
with all translations by codewords in \(\localcode\), and hence it is diagonal
in the character basis \(\{\ket{\chi}\}_{\chi\in\wC}\). This converts the
positivity constraint for the erasure operator into scalar inequalities in
the corresponding diagonal coefficients.

It remains to handle the individual affine POVM elements
\(\Pi_{\subspace,\bs}\), since the unambiguity constraints still distinguish
the different affine cosets. We show that these operators can be restricted to
the codeword-state span \(\mathcal W\), and then further to the subspace
compatible with the 
affine filtering constraint. In the character basis, this forces
\(\Pi_{\subspace,\bs}\) to decompose according to cosets of
\(\subspaceperp\). A final diagonalization in the corresponding
\(\ket{\zeta}\)-basis leaves only scalar variables
\(\beta_{\subspace}(\subspaceperp_m)\).

Substituting this diagonal form into the objective and the erasure-positivity
constraints gives the LP. The proof therefore proceeds in two directions. First, any
optimal SDP solution can be symmetrized and diagonalized to produce a feasible
LP solution with the same value. Conversely, any feasible LP solution
constructs a valid 
affine filtering POVM with the same objective value.
We begin by exploiting coset symmetry of the state family in the following section.
\subsection{Reduced SDP using Coset Symmetry}

In this section we obtain an SDP reduction using the translation symmetry of the state family $\{\ket{\psi_{\bc}}\}_{\bc\in \localcode}$. We show that any feasible POVM can be symmetrized, without changing its objective value, so that POVM elements corresponding to affine cosets of the same subspace are related by conjugation under translation operators.
  

\begin{lemma}[Coset Symmetry]\label{lem :coset symmetry}
For any feasible POVM $ \fullpovmtilde\in \sdpfeasible$ there exists a feasible POVM $\fullpovm\in \sdpfeasible$
such that
\begin{align*}
    \fsdp(\fullpovm)=\fsdp(\fullpovmtilde),
\end{align*}
and for every $\subspace\in \codesubspaces^{\localcode}$, every $\bs_1,\bs_2\in \Image(H_{\subspace})$, every $\bw_1\in \subspace_{\bs_1}$, and every $\bw_2\in \subspace_{\bs_2}$,
\begin{align*}
    \Pi_{\subspace,\bs_2}=X^{\bw_2-\bw_1}\Pi_{\subspace,\bs_1}X^{-(\bw_2-\bw_1)}.
\end{align*}
In particular, if $\fullpovmtilde$ is optimal, then $\fullpovm$ is also optimal.
\end{lemma}
\begin{proof}


    Let $\fullpovmtilde\in \sdpfeasible$ be any feasible POVM, where
\begin{align*}
    \tilde{\Pi}_{e}
    = \mI-\sum_{\subspace\in \codesubspaces^{\localcode},\bs\in\Image(H_{\subspace})}\tilde{\Pi}_{\subspace,\bs}.
\end{align*}
    Consider the following transformation for a codeword $\bw\in \localcode$
    \begin{align*}
        \Pi_{\subspace,\bs}^{(\bw)}=X^{\bw}\tilde{\Pi}_{\subspace,\bs-H_{\subspace}\bw}X^{-\bw}.
    \end{align*}
That is, $\{\Pi_{\subspace,\bs}^{(\bw)}\}_{\subspace,\bs}$ is obtained from $\fullpovmtilde$ by translating the codeword labels by $\bw$ and conjugating by the corresponding shift operator.
  For this transformed POVM,
    \begin{align*}
        \Tr(\Pi_{\subspace,\bs}^{(\bw)}\ket{\psi_{\bc}}\bra{\psi_{\bc}})=\Tr(\tilde{\Pi}_{\subspace,\bs-H_{\subspace}\bw}\ket{\psi_{\bc-\bw}}\bra{\psi_{\bc-\bw}}).
    \end{align*}
If $\bc\notin \subspace_{\bs}$, then $\bc-\bw \notin \subspace_{\bs-H_{\subspace}\bw}$. Therefore, when $\bc\notin \subspace_{\bs}$,
\begin{align*}
    \Tr(\Pi_{\subspace,\bs}^{(\bw)}\ket{\psi_{\bc}}\bra{\psi_{\bc}})=0.
\end{align*}
Next, consider the operator  $\Pi_{e}^{(\bw)}$ associated with the  erasure outcome after this transformation defined by
\begin{align*}
    \Pi_{e}^{(\bw)}=\mI-\sum_{\subspace\in \codesubspaces^{\localcode},\bs\in\Image(H_{\subspace})}\Pi_{\subspace,\bs}^{(\bw)}.
\end{align*}
Since the matrix $\tilde{\Pi}_{\subspace,\bs}$ is positive semidefinite (PSD), it follows immediately that $\Pi_{\subspace,\bs}^{(\bw)}$ is also
PSD $\forall\subspace\in \codesubspaces^{\localcode},\bs\in\Image(H_{\subspace})$. It also follows that 
\begin{align*}
    \sum_{\subspace\in \codesubspaces^{\localcode},\bs\in\Image(H_{\subspace})}\Pi_{\subspace,\bs}^{(\bw)}
&\ =X^{\bw}\left(\sum_{\subspace\in \codesubspaces^{\localcode},\bs\in\Image(H_{\subspace})}\tilde{\Pi}_{\subspace,\bs-H_{\subspace}\bw}\right)X^{-\bw}\\
&\ \preceq \mI.
\end{align*}
Hence, the operator $\Pi_{e}^{(\bw)}$ is also PSD and $\{\Pi_{\subspace,\bs}^{(\bw)}\}_{\subspace\in \codesubspaces^{\localcode},\bs\in \Image(H_{\subspace})}\cup\{\Pi_{e}^{(\bw)}\}$ is a valid POVM.
On the other hand, $(\subspace_{\bs},\bc)\rightarrow (\subspace_{\bs-H_{\subspace}\bw},\bc-\bw)$ is a bijection with $\dim(\subspace_{\bs-H_{\subspace}\bw})=\dim(\subspace_{\bs})$ $\forall\subspace\in \codesubspaces^{\localcode},\bs\in\Image(H_{\subspace})$. Thus the overall objective for an arbitrary codeword $\bw\in \localcode$ stays the same, i.e.,
\begin{align*}
   &\  \fsdp(\{\Pi_{\subspace,\bs}^{(\bw)}\}_{\subspace\in \codesubspaces^{\localcode},\bs\in \Image(H_{\subspace})}\cup\{\Pi_{e}^{(\bw)}\}) \\
   &\ =\frac{1}{|\localcode|}\sum_{\bc\in \localcode}\sum_{\subspace\in \codesubspaces^{\localcode},\bs\in\Image(H_{\subspace})}\Tr(\Pi_{\subspace,\bs}^{(\bw)}\ket{\psi_{\bc}}\bra{\psi_{\bc}})\reward(\subspace)\\
 &\ =   \frac{1}{|\localcode|}\sum_{\bc\in \localcode}\sum_{\subspace\in \codesubspaces^{\localcode},\bs\in\Image(H_{\subspace})}\Tr(\tilde{\Pi}_{\subspace,\bs-H_{\subspace}\bw}\ket{\psi_{\bc-\bw}}\bra{\psi_{\bc-\bw}})\reward(\subspace)\\
     &\ =   \frac{1}{|\localcode|}\sum_{\bc'\in \localcode}\sum_{\subspace\in \codesubspaces^{\localcode},\bs'\in\Image(H_{\subspace})}\Tr(\tilde{\Pi}_{\subspace,\bs'}\ket{\psi_{\bc'}}\bra{\psi_{\bc'}})\reward(\subspace)\\
    &\ = \fsdp(\fullpovmtilde).
\end{align*}
Consequently $\{\Pi_{\subspace,\bs}^{(\bw)}\}_{\subspace,\bs}\cup\{\Pi_e^{(\bw)}\}$ is feasible for the SDP in \eqref{SDP main} and has the same objective value as $\fullpovmtilde$.
   Now define the following symmetry transformation, where 
    \begin{align*}
        \Pi_{\subspace,\bs}=\frac{1}{|\localcode|}\sum_{\bc\in \localcode}\Pi_{\subspace,\bs}^{(\bc)}.
    \end{align*}
    For affine subspaces $\subspace_{\bs_1},\subspace_{\bs_2}\subset \localcode$ and codewords $\bw_{1}\in \subspace_{\bs_1}$, $\bw_{2}\in \subspace_{\bs_2}$ this gives 
    \begin{align*}
    X^{\bw_{2}-\bw_{1}}    \Pi_{\subspace,\bs_1} X^{-(\bw_{2}-\bw_{1})} &\ = \frac{1}{|\localcode|}\sum_{\bc \in \localcode}X^{\bc+\bw_{2}-\bw_{1}}\tilde{\Pi}_{\subspace,\bs_1-H_{\subspace}\bc}X^{-(\bc+\bw_{2}-\bw_{1})}\\
        &\ = \frac{1}{|\localcode|}\sum_{\bc'\in \localcode}X^{\bc'}\tilde{\Pi}_{\subspace,\bs_1-H_{\subspace}(\bc'-\bw_{2}+\bw_{1})}X^{-\bc'}\\
        &\ =\frac{1}{|\localcode|}\sum_{\bc'\in \localcode}X^{\bc'}\tilde{\Pi}_{\subspace,\bs_1+\bs_{2}-\bs_{1}-H_{\subspace}\bc'}X^{-\bc'}\\
         &\ = \frac{1}{|\localcode|}\sum_{\bc'\in \localcode}X^{\bc'}\tilde{\Pi}_{\subspace,\bs_2-H_{\subspace}\bc'}X^{-\bc'}\\
        &\ =  \Pi_{\subspace,\bs_2}.
    \end{align*}
Since each $\Pi_{\subspace,\bs}^{(\bw)}$ is PSD $\forall \bw\in \localcode$, the operator $\Pi_{\subspace,\bs}$ is also a PSD $\forall\subspace\in \codesubspaces^{\localcode},\bs\in\Image(H_{\subspace})$.
Defining the operator $\Pi_{e}=\mI-\sum_{\subspace\in \codesubspaces^{\localcode},\bs\in\Image(H_{\subspace})}\Pi_{\subspace,\bs}$, we obtain $\Pi_{e}\succeq\mzero$, so that $\fullpovm$ is a valid POVM.
The overall objective becomes 
\begin{align*}
  &\  \fsdp(\fullpovm)\\
  &\ = \frac{1}{|\localcode|}\sum_{\bc\in \localcode}\sum_{\subspace\in \codesubspaces^{\localcode},\bs\in\Image(H_{\subspace})}\Tr(\Pi_{\subspace,\bs}\ket{\psi_{\bc}}\bra{\psi_{\bc}})\reward(\subspace)\\
      &\ = \frac{1}{|\localcode|}\sum_{\bw\in \localcode}\left(\frac{1}{|\localcode|}\sum_{\bc\in \localcode}\sum_{\subspace\in \codesubspaces^{\localcode},\bs\in\Image(H_{\subspace})}\Tr(\Pi_{\subspace,\bs}^{(\bw)}\ket{\psi_{\bc}}\bra{\psi_{\bc}})\reward(\subspace)\right)\\
&\ =\frac{1}{|\localcode|}\sum_{\bw\in \localcode}\left(\fsdp(\{\Pi_{\subspace,\bs}^{(\bw)}\}_{\subspace\in \codesubspaces^{\localcode},\bs\in \Image(H_{\subspace})}\cup\{\Pi_{e}^{(\bw)}\})\right)\\
    &\ =\fsdp(\{\Pi_{\subspace,\bs}^{(\bw)}\}_{\subspace\in \codesubspaces^{\localcode},\bs\in \Image(H_{\subspace})}\cup\{\Pi_{e}^{(\bw)}\})\\
    &\ = \fsdp(\fullpovmtilde).
\end{align*}
Thus, starting from any feasible POVM $\fullpovmtilde\in \sdpfeasible$, we have constructed a feasible POVM $\fullpovm\in \sdpfeasible$ with
\begin{align*}
    \fsdp(\fullpovm)=\fsdp(\fullpovmtilde),
\end{align*}
that satisfies the stated coset symmetry relation. In particular, if $\fullpovmtilde$ is optimal, then $\fullpovm$ is also optimal.
\end{proof}

\noindent For a fixed POVM $\fullpovm$, define the operator $\Pi_{\subspace}$ associated with the subspace $\subspace\in \codesubspaces^{\localcode}$ by
\begin{align*}
    \Pi_{\subspace}=\sum_{\bs\in \Image({H_{\subspace}})}\Pi_{\subspace,\bs}.
\end{align*}
Next, we use the coset symmetry lemma, to prove that the operators $\Pi_{\subspace}$ are invariant under all transformations $X^{\bc}$ for all $\bc\in \localcode$ and $\subspace\in \codesubspaces^{\localcode}$.
\begin{lemma}\label{lem: povm X symmetry relations}
    Consider a POVM $\fullpovm\in\sdpfeasible$, which achieves the optimal objective value for SDP problem in \eqref{SDP main} such that the POVM elements $\Pi_{\subspace,\bs_1}$, $\Pi_{\subspace,\bs_2}$ corresponding to two arbitrary affine subspaces $\subspace_{\bs_1}, \subspace_{\bs_2}\subset \localcode$ for $\subspace\in \codesubspaces^{\localcode}$ and $\bs_1,\bs_2\in \Image(H_{\subspace})$, satisfy 
    \begin{align}\label{eq: povm coset symmetry relation}
    \Pi_{\subspace,\bs_2}=X^{\bw_2-\bw_1}\Pi_{\subspace,\bs_1}X^{-(\bw_{2}-\bw_1)}.
\end{align}
$\forall \bw_1\in \subspace_{\bs_1}$ and $\bw_{2}\in \subspace_{\bs_2}$. Then the following equalities hold
    \begin{align*}
\Pi_{\subspace,\bs} &\ =X^{\bv}\Pi_{\subspace,\bs}X^{-\bv},  \quad \forall\subspace\in \codesubspaces^{\localcode},\bs\in\Image(H_{\subspace}),\bv\in \subspace\\
\Pi_{\subspace} &\ = X^{\bc}\Pi_{\subspace}X^{-\bc} ,  \quad \forall\subspace\in \codesubspaces^{\localcode},\bc\in \localcode.
        \end{align*}
\end{lemma}
\begin{proof}
For $\bw_{1},\bw_{2}\in \subspace_{\bs}$ and a subspace $\subspace\in \codesubspaces^{\localcode}$ and $\bs\in \Image(H_{\subspace})$, we obtain 
\begin{align}\label{eq: povm commuting in subspace}
    X^{\bw_{2}-\bw_{1}}\Pi_{\subspace,\bs}X^{-(\bw_{2}-\bw_{1})}=\Pi_{\subspace,\bs}.
\end{align}
Since $\bw_{1},\bw_{2}\in \subspace_{\bs}\implies \bw_{2}-\bw_{1}\in \subspace$ and the equation \eqref{eq: povm commuting in subspace} holds for all pairs $\bw_{1},\bw_{2}\in \subspace_{\bs}$, it follows that $\Pi_{\subspace,\bs}  =X^{\bv}\Pi_{\subspace,\bs}X^{-\bv}$,  $\forall\subspace\in \codesubspaces^{\localcode},\bs\in\Image(H_{\subspace}),\bv\in \subspace$.
For any codeword $\bc\in \localcode$, using the relation in \eqref{eq: povm coset symmetry relation} and \eqref{eq: povm commuting in subspace}, for some $\subspace\in \codesubspaces^{\localcode}$ and $\bs\in \Image(H_{\subspace})$ this yields
\begin{align}
    X^{\bc}\Pi_{\subspace,\bs}X^{-\bc}=\Pi_{\subspace,\bs+H_{\subspace}\bc}.
\end{align}
Therefore, 
\begin{align*}
    X^{\bc}\Pi_{\subspace}X^{-\bc} &\ =\sum_{\bs\in \Image(H_{\subspace})}\Pi_{\subspace,\bs+H_{\subspace}\bc}\\
    &\ = \sum_{\bs'\in \Image(H_{\subspace})}\Pi_{\subspace,\bs'}\\
    &\ =\Pi_{\subspace}.
\end{align*}
Since it holds for any $\subspace\in \codesubspaces^{\localcode}$ and $\bc\in \localcode$, it completes the proof.
\end{proof}
By Lemma~\ref{lem :coset symmetry}, we may fix an optimal POVM $\fullpovm\in \sdpfeasible$ that satisfies the coset symmetry relation. Then Lemma~\ref{lem: povm X symmetry relations} implies that the objective in Eq.~\eqref{SDP main} can be rewritten as
\begin{align}
    \fsdp(\fullpovm) &\ =\frac{1}{|\localcode|}\sum_{\bc\in \localcode}\sum_{\subspace\in \codesubspaces^{\localcode},\bs\in\Image(H_{\subspace})}\Tr(\Pi_{\subspace,\bs}\ket{\psi_{\bc}}\bra{\psi_{\bc}})\reward(\subspace)\\
    &\ = \frac{1}{|\localcode|}\sum_{\bc\in \localcode}\sum_{\subspace\in \codesubspaces^{\localcode}}\Tr(\Pi_{\subspace}\ket{\psi_{\bc}}\bra{\psi_{\bc}})\reward(\subspace)\\
    &\ =\sum_{\subspace\in \codesubspaces^{\localcode}}\Tr(\Pi_{\subspace}\ket{\psi_{\mzero}}\bra{\psi_{\mzero}})\reward(\subspace).
\end{align}

\noindent After symmetrizing the optimal POVM as in Lemma~\ref{lem :coset symmetry},
we may rewrite the objective in \eqref{SDP main} as
\[
\sum_{\subspace\in \codesubspaces^{\localcode}}
\Tr(\Pi_{\subspace}\ket{\psi_{\mzero}}\bra{\psi_{\mzero}})
\reward(\subspace),
\]
while retaining the 
affine filtering constraints on the operators
$\Pi_{\subspace,\bs}$. Equivalently, we may work with
 \begin{align}{\label{SDP state independence}}
        \text{Maximize}  \quad &\   \sum_{\subspace\in \codesubspaces^{\localcode}}\Tr(\Pi_{\subspace}\ket{\psi_{\mzero}}\bra{\psi_{\mzero}})\reward(\subspace)\\
 \text{Subject to:} \quad &\ \Tr\left(\Pi_{\subspace,\bs}\ket{\psi_{\bc}}\bra{\psi_{\bc}}\right)=0 \quad \text{if $\mathbf{c}\notin \subspace_{\bs}$} \nonumber \\
    &\ \sum_{\subspace\in \codesubspaces^{\localcode}}\Pi_{\subspace}+\Pi_e=\mI \nonumber\\
    &\ \Pi_{\subspace,\bs}\succeq 0, \forall \subspace\in \codesubspaces^{\localcode},\bs\in\Image(H_{\subspace})\nonumber\\&\ \Pi_{e}\succeq 0.
    \end{align}

\noindent    
This removes the explicit average over transmitted codewords from the objective, but the individual affine unambiguity constraints remain to be handled.
  

 \subsection{Restricting POVM to Subspace Spanned by Codeword Quantum States}
 In this section, we show that the informative POVM elements can be restricted
to the subspace spanned by the codeword quantum states without losing optimality.
Let $\mathcal{W}\subset \mathcal{H}^{\otimes |I|}$ be the vector space associated with the codeword states for the local code $\localcode$ which satisfies
\begin{align*}
  \mathcal{W}  =\text{span}\left(\{\ket{\psi_{\mathbf{c}}}\}_{\mathbf{c}\in \localcode}\right).
\end{align*} 
Let $\mathcal{L}(\mathcal{W})=\{T:\mathcal{W}\rightarrow \mathcal{W}\}$ denote the space of linear maps in the complex vector space $\mathcal{W}$.

\begin{lemma}[POVM Space]\label{lem: POVM space} There exists a POVM  $\fullpovm$ which achieves the optimal objective for the SDP problem in \eqref{SDP main} such that  $\Pi_{\subspace,\bs}\in \mathcal{L}(\mathcal{W})$.
    
\end{lemma}

\begin{proof}

    Let $\fullpovmtilde$ be the POVM that achieves the objective described in \eqref{SDP main}.  Let $P_{\mathcal{W}}$  be the projection operator onto $\mathcal{W}$. Then
    \begin{align*}
        P_{\mathcal{W}}\ket{\psi_{\bc}}=\ket{\psi_{\bc}}.
    \end{align*}
    Let us define a POVM $\fullpovm$ such that 
    \begin{align*}
        \Pi_{\subspace,\bs}=P_{\mathcal{W}}\tilde{\Pi}_{\subspace,\bs}P_{\mathcal{W}},
    \end{align*}
    and
\begin{align*}
    \Pi_e \coloneqq \mI-\sum_{\subspace\in \codesubspaces^{\localcode},
    \bs\in\Image(H_{\subspace})}\Pi_{\subspace,\bs}.
\end{align*}
    Therefore $\Pi_{\subspace,\bs}\succeq 0$ and
    \begin{align*}
       \sum_{\subspace\in \codesubspaces^{\localcode},\bs\in\Image(H_{\subspace})}\Pi_{\subspace,\bs}=P_{\mathcal{W}}(\sum_{\subspace\in \codesubspaces^{\localcode},\bs\in\Image(H_{\subspace})}\tilde{\Pi}_{\subspace,\bs}) P_{\mathcal{W}}\preceq \mI,
    \end{align*}
which implies $\Pi_e\succeq \mzero$.
The trace constraints are also preserved:
\begin{align*}
    \Tr(\Pi_{\subspace,\bs}\ket{\psi_{\bc}}\bra{\psi_{\bc}}) &\ =\Tr(\tilde{\Pi}_{\subspace,\bs}P_{\mathcal{W}}\ket{\psi_{\bc}}\bra{\psi_{\bc}}P_{\mathcal{W}})\\
    &\ =\Tr(\tilde{\Pi}_{\subspace,\bs}\ket{\psi_{\bc}}\bra{\psi_{\bc}}).
\end{align*}
Therefore, if $\bc\notin \subspace_{\bs}$, $\Tr(\Pi_{\subspace,\bs}\ket{\psi_{\bc}}\bra{\psi_{\bc}})=0$ and the overall objective remains unchanged:
\begin{align*}
    \fsdp(\fullpovm)=\fsdp(\fullpovmtilde).
\end{align*}
Consequently, the POVM $\fullpovm$ achieves the optimal objective while satisfying $\Pi_{\subspace,\bs}\in \mathcal{L}(\mathcal{W})$ $\forall \subspace\in \codesubspaces^{\localcode},\bs\in\Image(H_{\subspace})$.
\end{proof}

\noindent For an arbitrary affine subspace $\subspace_{\bs}\subset \localcode$, let $\sim \subspace_{\bs}=\localcode\setminus \subspace_{\bs}$ and define the vector space 
\begin{align*}
  \mathcal{W}_{\sim \subspace_{\bs}}= \mathcal{W}_{\localcode\setminus \subspace_{\bs}}=\text{span}\left(\{\ket{\psi_{\mathbf{c}}}\}_{\mathbf{c}\in \localcode\setminus \subspace_{\bs}}\right).
\end{align*}

\noindent Then we can write the vector space 
\begin{align*}
    \mathcal{W}_{\sim \subspace_{\bs}}^{\perp}\cap \mathcal{W}=\text{span}\left(\{\ket{\psi_{\bc}}\}_{\bc\in \sim \subspace_{\bs}}\right)^{\perp}\cap \text{span}\left(\{\ket{\psi_{\bc}}\}_{\bc\in \localcode}\right).
\end{align*}
\begin{corollary}\label{corol:povm space}
There exists a POVM $\fullpovm$ that achieves the optimal objective value for the SDP problem in \eqref{SDP main} such that  $\Pi_{\subspace,\bs}\in \mathcal{L}(\mathcal{W}_{\sim \subspace_{\bs}}^{\perp}\cap \mathcal{W})$ $\forall \subspace\in \codesubspaces^{\localcode},\bs\in\Image(H_{\subspace})$.
\end{corollary}

\begin{proof}
The USD constraints imply that
\[
\bra{\psi_{\bc}}\Pi_{\subspace,\bs}\ket{\psi_{\bc}}=0,
\qquad \forall \bc\notin \subspace_{\bs}.
\]
Since $\Pi_{\subspace,\bs}\succeq \mzero$, this implies
$\Pi_{\subspace,\bs}\ket{\psi_{\bc}}=\mzero$ for all
$\bc\notin \subspace_{\bs}$. It follows that
\[
\text{Im}(\Pi_{\subspace,\bs})
\subseteq \mathcal{W}^{\perp}_{\sim \subspace_{\bs}}.
\]
By Lemma~\ref{lem: POVM space}, we may furthermore take
$\Pi_{\subspace,\bs}\in\mathcal{L}(\mathcal{W})$. Combining the two facts gives
\[
\Pi_{\subspace,\bs}\in
\mathcal{L}\bigl(\mathcal{W}^{\perp}_{\sim \subspace_{\bs}}\cap \mathcal{W}\bigr),
\]
while preserving the same objective value.
\end{proof}
\subsection[]%
{Diagonalization of the Operator $\Pi_{\subspace}$}
In this section, we show that the operators $\{\Pi_{\subspace}\}_{\subspace\in \codesubspaces^{\localcode}}$ are diagonal in a carefully constructed basis due to
invariance under the $X^{\bc}$ transformations $\forall\bc\in \localcode$. For this, we use characters $\chi\in \wC$.
More specifically, for the collection of quantum states $\{\ket{\psi_{\bc}}\}_{\bc\in \localcode}$ and the character $\chi\in \wC$, define the
unnormalized character vector by
\begin{align*}
    \ket{\tilde{\chi}}\coloneqq\frac{1}{\sqrt{|\localcode|}}\sum_{\bc\in \localcode}\chi(\bc)\ket{\psi_{\bc}}.
\end{align*}
\begin{lemma}
   For the character vectors $\{\ket{\tilde{\chi}}\}_{\chi\in \wC}$ associated with the collection of quantum states $\{\ket{\psi_{\bc}}\}_{\bc\in \localcode}$, the following holds
    \begin{align}
        \braket{\tilde{\chi}_{1}|\tilde{\chi}_{2}}=\kappa_{\chi_1}\delta_{\chi_1,\chi_2}\label{eq:innerproductCharacterStates},
    \end{align}
    where 
\[
\kappa_{\chi}\coloneqq \sum_{\bc\in \localcode} g(\bc)\chi(\bc),
\]
and $g(\bc)\coloneqq \braket{\psi_{\mzero}|\psi_{\bc}}$.
\end{lemma}
\begin{proof}
   
For the Gram matrix $G$ for the collection of quantum states $\{\ket{\psi_{\bc}}\}_{\bc\in \localcode}$, the following holds
\begin{align*}
    G(\bc,\bc') &\ =\braket{\psi_{\bc}|\psi_{\bc'}}\\
    &\ = \bra{\psi}(\QFT^{\dagger})^{\otimes |I|}X^{-\bc}X^{\bc'}\QFT^{\otimes |I|}\ket{\psi}\\
    &\ = \bra{\psi}(\QFT^{\dagger})^{\otimes |I|}X^{\bc'-\bc} \QFT^{\otimes |I|}\ket{\psi}\\
    & = \braket{\psi_{\mzero}|\psi_{\bc'-\bc}}\\
    &\ = g( \bc'-\bc).
\end{align*}
Therefore the Gram matrix element $G(\bc,\bc')$ depends only on $\bc'-\bc$.
In addition
\begin{align*}
    \braket{\tilde{\chi}_{1}|\tilde{\chi}_{2}} &\ =\frac{1}{|\localcode|}\sum_{\bc,\bc'\in \localcode}\chi_{1}^{-1}(\bc)\chi_{2}(\bc')\braket{\psi_{\bc}|\psi_{\bc'}}\\
    &\ =\frac{1}{|\localcode|}\sum_{\bc,\bc'\in \localcode}\chi_{1}^{-1}(\bc)\chi_{2}(\bc')g( \bc'-\bc)\\
    &\ =\frac{1}{|\localcode|}\sum_{\bc,\ba\in \localcode}\chi_{1}^{-1}(\bc)\chi_{2}(\bc+\ba)g(\ba)\\
    &\ = \frac{1}{|\localcode|}\sum_{\ba\in \localcode}g(\ba)\chi_{2}(\ba)\left(\sum_{\bc\in \localcode}\chi_{1}^{-1}(\bc)\chi_{2}(\bc)\right).
\end{align*}
The characters also satisfy 
\begin{align*}
    \sum_{\bc\in \localcode}\chi_{1}^{-1}(\bc)\chi_{2}(\bc)=|\localcode|\delta_{\chi_1,\chi_2}.
\end{align*}
Therefore, this yields
\begin{align*}
    \braket{\tilde{\chi}_{1}|\tilde{\chi}_{2}}=\kappa_{\chi_1}\delta_{\chi_1,\chi_2}, 
\end{align*}
where we use the fact that $\kappa_{\chi_1}=\sum_{\bc\in \localcode}g(\bc)\chi_{1}(\bc)$.
  \end{proof}
\noindent We define the normalized character vectors 
$\ket{\chi}$ as $\ket{\chi}=\frac{\ket{\tilde{\chi}}}{\sqrt{\kappa_{\chi}}}$ whenever $\kappa_{\chi}>0$. When $G$ is full rank, i.e., the states $\{\ket{\psi_{\bc}}\}_{\bc\in \localcode}$ are linearly independent, this implies  $\kappa_{\chi}>0$ for all $\chi\in \wC$ and the normalized character vectors are well defined.

\begin{lemma}\label{lem:Gram eigenvalues}
Consider the collection of symmetric quantum states $\{\ket{\psi_{\bc}}\}_{\bc\in \localcode}$, whose Gram matrix $G$ is circulant, i.e. $G(\bc,\bc')$ depends only on the difference $\bc'-\bc$. Then $\{\kappa_{\chi}\}_{\chi\in \wC}$ are the eigenvalues of $G$.

\end{lemma}
\begin{proof}
    Using the canonical states $\{\ket{\bc}\}_{\bc\in \localcode}$, define the vector $\{\ket{v_{\chi}}\}_{\chi\in \wC}$ by
    \begin{align*}
        \ket{v_{\chi}}\coloneqq \frac{1}{\sqrt{|\localcode|}}\sum_{\bc'\in \localcode}\chi(\bc')\ket{\bc'}.
    \end{align*}
    Let $G(\bc,:)$ denote the row of the Gram matrix $G$ indexed by codeword $\bc$. Then we obtain 
    \begin{align*}
        G(\bc,:)\ket{v_{\chi}} & =\frac{1}{\sqrt{|\localcode|}}\sum_{\bc'\in \localcode}g(\bc'-\bc)\chi(\bc')\\
        & = \frac{\chi(\bc)}{\sqrt{|\localcode|}}\sum_{\bc'\in \localcode}g(\bc'-\bc)\chi(\bc'-\bc)\\
        & = \frac{\chi(\bc)}{\sqrt{|\localcode|}}\sum_{\ba\in \localcode}g(\ba)\chi(\ba)\\
        & = \kappa_{\chi}\frac{\chi(\bc)}{\sqrt{|\localcode|}}.
    \end{align*}
    Therefore, the vector $\ket{v_{\chi}}$ satisfies 
    \begin{align*}
        G\ket{v_{\chi}}=\frac{\kappa_{\chi}}{\sqrt{|\localcode|}}\sum_{\bc\in\localcode}\chi(\bc)\ket{\bc}= \kappa_{\chi}\ket{v_{\chi}}, \quad \forall \chi\in \wC.
    \end{align*}
For two characters $\chi_1,\chi_2\in \wC$, from the character orthogonality relation it follows that
\begin{align*}
    \braket{v_{\chi_1}|v_{\chi_2}}=\frac{1}{|\localcode|}\sum_{\bc\in \localcode}\chi_{1}^{-1}(\bc)\chi_{2}(\bc)=\delta_{\chi_1,\chi_2}.
\end{align*}
Hence, $\{v_{\chi}\}_{\chi\in\wC}$ are the eigenvectors of the Gram matrix $G$ with eigenvalues $\{\kappa_{\chi}\}_{\chi\in \wC}$. 
\end{proof}
\noindent Since the Gram matrix is PSD, the eigenvalues satisfy $\kappa_{\chi}\geq 0$ $\forall \chi \in \wC$. For the remainder of this subsection, we assume the codeword quantum states $\{\ket{\psi_{\bc}}\}_{\bc\in \localcode}$ are linearly independent, which makes the Gram matrix positive definite and $\kappa_{\chi}> 0$ $\forall \chi\in \wC$.  Next lemma shows that these eigenvalues are always bounded. 
\begin{lemma}
    For the collection of quantum states $\{\ket{\psi_{\bc}}\}_{\bc\in \localcode}$, consider the Gram matrix $G$ with eigenvalues $\{\kappa_{\chi}\}_{\chi\in \wC}$. Then the following relation holds
    \begin{align*}
        \sum_{\chi\in \wC}\kappa_{\chi}=|\localcode|.
    \end{align*}
\end{lemma}
\begin{proof}
    This follows directly from the trace relation 
    \begin{align*}
        \sum_{\chi\in \wC}\kappa_{\chi} =\Tr(G)=\sum_{\bc}G(\bc,\bc)=\sum_{\bc\in \localcode}\braket{\psi_{\bc}|\psi_{\bc}}=|\localcode|.
    \end{align*}
\end{proof}
\noindent Next, we prove that the character states $\{\ket{\chi}\}_{\chi\in \wC}$ span the space spanned by codeword states $\{\ket{\psi_{\bc}}\}_{\bc\in \localcode}$.
\begin{lemma}
    For character vectors $\characterbasis$, 
    \begin{align*}
        \sum_{\chi\in \wC}\ket{\chi}\bra{\chi}=P_{\mathcal{W}},
    \end{align*}
where $P_{\mathcal{W}}$ is the projection operator onto the codeword state space $\mathcal{W}=\text{span}\{\ket{\psi_{\bc}}\}_{\bc\in \localcode}\}$.
\end{lemma}
\begin{proof}
For $g(\bv-\bc)=\braket{\psi_{\bc}|\psi_{\bv}}$ and for the codeword $\bv\in \localcode$, we obtain 

\begin{align*}
    \braket{\chi|\psi_{\bv}}&\ =\frac{1}{\sqrt{|\localcode|\kappa_{\chi}}}\sum_{\bc\in \localcode}\chi^{-1}(\bc)\braket{\psi_{\bc}|\psi_{\bv}}\\
    &\ =\frac{1}{\sqrt{|\localcode|\kappa_{\chi}}}\sum_{\ba\in \localcode} g(\ba)\chi^{-1}(\bv-\ba)\\
    &\ =\frac{\sqrt{\kappa_{\chi}}}{\sqrt{|\localcode|}}\chi^{-1}(\bv).
\end{align*}
Then, for the state $\ket{\psi_{\bv}}$ for $\bv\in \localcode$ this yields  
\begin{align*}
    P_{\mathcal{W}}\ket{\psi_{\bv}}&\ = \sum_{\chi\in \wC}\ket{\chi}\braket{\chi|\psi_{\bv}}\\
    &\ = \sum_{\chi\in\wC}\frac{\sqrt{\kappa_{\chi}}}{\sqrt{|\localcode|}}\chi^{-1}(\bv)\ket{\chi}\\
    &\ = \frac{1}{|\localcode|}\sum_{\chi\in \wC}\sum_{\bc\in \localcode}\chi^{-1}(\bv)\chi(\bc)\ket{\psi_{\bc}}.
\end{align*}
Since $\sum_{\chi\in \wC}\chi^{-1}(\bv)\chi(\bc)=|\localcode|\delta_{\bv,\bc}$, it follows that 
\begin{align*}
    P_{\mathcal{W}}\ket{\psi_{\bv}}=\ket{\psi_{\bv}}.
\end{align*}
The character vectors $\characterbasis$ also satisfy
    \begin{align*}
       \braket{\chi_{1}|\chi_{2}}=\delta_{\chi_1,\chi_2}.
    \end{align*}
Since $\ket{\chi}$ is defined with respect to the linear span of the codeword states $\{\psi_{\bc}\}_{\bc\in \localcode}$, for any vector $\ket{v}\in \mathcal{W}^{\perp}$ this implies $\braket{\chi|v}=0$ for all $ \chi\in \wC$. Thus we conclude that $P_{\mathcal{W}}$ is a projector onto the subspace $\mathcal{W}$.
\end{proof}
Based on the above lemmas, we are ready to state the diagonalization lemma for operators $\{\Pi_{\subspace}\}_{\subspace\in \codesubspaces^{\localcode}}$.
  \begin{lemma}\label{lem: povm subspace diagonal character basis}
Consider the codeword states $\{\ket{\psi_{\bc}}\}_{\bc\in \localcode}$, the corresponding character vectors $\characterbasis$, and an optimal POVM $\fullpovm$ satisfying Lemma~\ref{lem :coset symmetry}. Then the operators $\{\Pi_{\subspace}\}_{\subspace\in \codesubspaces^{\localcode}}$ are diagonal with respect to
$\characterbasis$. In other words, we can write every operator $\{\Pi_{\subspace}\}_{\subspace\in \codesubspaces^{\localcode}}$ as 
\begin{align*}
    \Pi_{\subspace}=\sum_{\chi\in \wC}\gamma_{\subspace}(\chi)\ket{\chi}\bra{\chi}.
\end{align*}
  \end{lemma}
  \begin{proof}
      For the operator $\Pi_{\subspace}=\sum_{\bs\in \Image(H_{\subspace})}\Pi_{\subspace,\bs}$ and for the codewords $\bc$ and $\bc'$, using Lemma~\ref{lem: povm X symmetry relations}, the following holds 
      \begin{align}\label{eq:pi_V code invariance}
          \bra{\psi_{\bc}}\Pi_{\subspace}\ket{\psi_{\bc'}} &\ = \bra{\psi_{\mathbf{0}}}X^{-\bc}\Pi_{\subspace}\ket{\psi_{\bc'}} \\
          &\ = \bra{\psi_{\mathbf{0}}}\Pi_{\subspace}X^{-\bc}\ket{\psi_{\bc'}}\nonumber \\
          &\ = \bra{\psi_{\mathbf{0}}}\Pi_{\subspace}\ket{\psi_{\bc'- \bc}}. \nonumber
      \end{align}
    This means $\bra{\psi_{\bc}}\Pi_{\subspace}\ket{\psi_{\bc'}}$ only depends on $\bc'-\bc$ which we denote as $\mu_{\subspace}( \bc'-\bc)$.
Moreover,

\begin{align*}
    \bra{\chi_{1}}\Pi_{\subspace}\ket{\chi_{2}} &\ = \frac{1}{\sqrt{\kappa_{\chi_1}\kappa_{\chi_2}}|\localcode|}\sum_{\bc,\bc'\in \localcode}\chi_{1}^{-1}(\bc)\chi_{2}(\bc')\bra{\psi_{\bc}}\Pi_{\subspace}\ket{\psi_{\bc'}}\\
    &\ =\frac{1}{\sqrt{\kappa_{\chi_1}\kappa_{\chi_2}}|\localcode|}\sum_{\bc,\bc'\in \localcode}\chi_{1}^{-1}(\bc)\chi_{2}(\bc')\mu_{\subspace}( \bc'-\bc)\\
    &\ =\frac{1}{\sqrt{\kappa_{\chi_1}\kappa_{\chi_2}}|\localcode|}\sum_{\bc,\ba\in \localcode}\chi_{1}^{-1}(\bc)\chi_{2}(\bc+\ba)\mu_{\subspace}(\ba)\\
    &\ = \frac{1}{\sqrt{\kappa_{\chi_1}\kappa_{\chi_2}}|\localcode|}\sum_{\ba\in \localcode}\mu_{\subspace}(\ba)\chi_{2}(\ba)\left(\sum_{\bc\in \localcode}\chi_{1}^{-1}(\bc)\chi_{2}(\bc)\right)\\
    &\ =\gamma_{\subspace}(\chi_1)\delta_{\chi_1,\chi_2},
\end{align*}
where $\gamma_{\subspace}(\chi_1)=\frac{1}{\kappa_{\chi_1}}\sum_{\bc\in \localcode}\mu_{\subspace}(\bc)\chi_{1}(\bc)$.
  \end{proof}
\noindent  Since each operator $\Pi_{\subspace}$ is diagonal with respect to the basis $\{\ket{\chi}\}_{\chi\in\wC}$ for all $\subspace\in \codesubspaces^{\localcode}$, we obtain a simple characterization of the operator $\Pi_e$, i.e., the measurement operator for inconclusive outcome in the space spanned by the states $\{\ket{\psi_{\bc}}\}_{\bc\in \localcode}$. 
\begin{lemma}\label{lem: pi_e diagonal form}
     Consider the codeword states $\{\ket{\psi_{\bc}}\}_{\bc\in \localcode}$  and the corresponding character vectors $\characterbasis$, the POVM $\fullpovm$ that achieves the optimal objective such that
     \begin{align*}
         \Pi_{e}=\mI-\sum_{\subspace\in \codesubspaces^{\localcode}}\Pi_{\subspace}.
     \end{align*}
    Then the operator $P_{\mathcal{W}}\Pi_{e}P_{\mathcal{W}}$ is diagonalizable with respect to $\characterbasis$.
\end{lemma}
\begin{proof}
    Applying $P_{\cW}$ on $\Pi_e$, we obtain
    \begin{align*}
     P_{\mathcal{W}}\Pi_{e}P_{\mathcal{W}} &\ =P_{\mathcal{W}}-  P_{\mathcal{W}}\left(\sum_{\subspace\in \codesubspaces^{\localcode}}\Pi_{\subspace}\right)P_{\mathcal{W}} \\
     &\ = P_{\mathcal{W}}-\sum_{\subspace\in \codesubspaces^{\localcode}}\sum_{\chi\in \wC}\gamma_{\subspace}(\chi)\ket{\chi}\bra{\chi}\\
     &\ = \sum_{\chi\in \wC}(1-\sum_{\subspace\in \codesubspaces^{\localcode}}\gamma_{\subspace}(\chi))\ket{\chi}\bra{\chi}.
    \end{align*}
\end{proof}
\noindent Since we require $\Pi_{e}\succeq 0$, this implies $P_{\mathcal{W}}\Pi_{e}P_{\mathcal{W}}\succeq 0$. Hence, we require  the following to hold
\begin{align*}
    \sum_{\subspace\in \codesubspaces^{\localcode}}\gamma_{\subspace}(\chi)\leq 1,\quad \forall\chi\in \wC.
\end{align*}
Moreover, since $\Pi_{\subspace}\succeq 0$ $\forall \subspace\in \codesubspaces^{\localcode}$,  we must have
\begin{align*}
    \gamma_{\subspace}(\chi)\geq 0, \quad \forall \subspace\in \codesubspaces^{\localcode},\chi\in \wC.
\end{align*}
Thus, by Lemma~\ref{lem: povm subspace diagonal character basis} and~\ref{lem: pi_e diagonal form}, the problem in \eqref{SDP main} reduces to the following equivalent formulation:

  \begin{align}{\label{SDP chracter}}
        \text{Maximize}  \quad &\   \sum_{\subspace\in \codesubspaces^{\localcode}}\Tr(\Pi_{\subspace}\ket{\psi_{\mzero}}\bra{\psi_{\mzero}})\reward(\subspace)\\
 \text{Subject to:} \quad &\ \Tr\left(\Pi_{\subspace,\bs}\ket{\psi_{\bc}}\bra{\psi_{\bc}}\right)=0, \quad \text{if $\bc\notin \subspace_{\bs}$}, \forall\subspace\in \codesubspaces^{\localcode},\bs\in \Image(H_{\subspace}) \nonumber \\
    &\ \sum_{\subspace\in \codesubspaces^{\localcode}}\gamma_{\subspace}(\chi)\leq 1 \qquad \forall \chi\in \wC\\
    &\ \gamma_{\subspace}(\chi)\geq 0, \quad \forall \subspace\in \codesubspaces^{\localcode},\chi\in \wC.\nonumber
    \end{align}

\subsection[Diagonalization of the 
affine filtering operator]%
{Diagonalization of the operator $\Pi_{\subspace,\bs}$}
Although the diagonalization of the operators $\{\Pi_{\subspace}\}_{\subspace\in \codesubspaces^{\localcode}}$ reduces the SDP in \eqref{SDP main} to a much simpler optimization problem, the problem still remains an SDP due to the unambiguity constraints i.e. $\Tr\left(\Pi_{\subspace,\bs}\ket{\psi_{\bc}}\bra{\psi_{\bc}}\right)=0$ $\forall\bc\notin \subspace_{\bs}$ $\forall\subspace\in \codesubspaces^{\localcode},\bs\in \Image(H_{\subspace})$. In this section, we will show that the operators $\{\Pi_{\subspace,\bs}\}_{\subspace\in \codesubspaces^{\localcode},\bs\in \Image(H_{\subspace})}$ are diagonalizable in a carefully constructed basis. In Lemma~\ref{lem: povm dual cosets}, we establish a key property for these operators with respect to the basis $\{\ket{\chi}\}_{\chi\in \wC}$. 
\begin{lemma}\label{lem: povm dual cosets}
There exists a POVM $\fullpovm$, which achieves the optimal objective value for the SDP in \eqref{SDP main} such that the POVM elements satisfy
\begin{align*}
    \bra{\chi_{1}}\Pi_{\subspace,\bs}\ket{\chi_{2}}=0, \forall \chi_{2}\chi_{1}^{-1}\notin \subspaceperp.
\end{align*}
\end{lemma}
\begin{proof}
    
For the unitary $X^{\bv}$ and the character state $\ket{\chi_{1}}$, the following holds  
\begin{align*}
    X^{\bv}\ket{\chi_{1}} &\ =\frac{1}{\sqrt{|\localcode|\kappa_{\chi_{1}}}}\sum_{\bc\in \localcode}\chi_{1}(\bc)X^{\bv}\ket{\psi_{\bc}}\\
    &\ = \frac{1}{\sqrt{|\localcode|\kappa_{\chi_{1}}}}\sum_{\bc\in \localcode}\chi_{1}(\bc)\ket{\psi_{\bc+\bv}}\\
    &\ = \frac{1}{\sqrt{|\localcode|\kappa_{\chi_{1}}}}\sum_{\bc'\in \localcode}\chi_{1}(\bc'-\bv)\ket{\psi_{\bc'}}\\
    &\ = \frac{\chi_{1}(-\bv)}{\sqrt{|\localcode|\kappa_{\chi_{1}}}}\sum_{\bc'\in \localcode}\chi_{1}(\bc')\ket{\psi_{\bc'}}\\
    &\ =\chi_{1}^{-1}(\bv)\ket{\chi_{1}}.
\end{align*}
Then, for some $\bv\in \subspace$, we obtain

\begin{align*}
    \bra{\chi_{1}}\Pi_{\subspace,\bs}\ket{\chi_{2}} &\ = \bra{\chi_{1}}X^{\bv}\Pi_{\subspace,\bs}X^{-\bv}\ket{\chi_{2}}\\
    &\ = \chi_{1}^{-1}(\bv)\chi_{2}(\bv)(\bra{\chi_{1}}\Pi_{\subspace,\bs}\ket{\chi_{2}})\\
    &\ =(\chi_{2}\chi_{1}^{-1})(\bv) (\bra{\chi_{1}}\Pi_{\subspace,\bs}\ket{\chi_{2}}).
\end{align*} 
This gives 
\begin{align*}
    \left((\chi_{2}\chi_{1}^{-1})(\bv)-1\right)\bra{\chi_{1}}\Pi_{\subspace,\bs}\ket{\chi_{2}}=0.
\end{align*}
Since $\chi_{2}\chi_{1}^{-1}\notin \subspaceperp$, there exists a codeword $\bv\in \subspace$ for which $(\chi_{2}\chi_{1}^{-1})(\bv)\neq 1$. Then, we must have $\bra{\chi_{1}}\Pi_{\subspace,\bs}\ket{\chi_{2}}=0$ for all $\chi_{2},\chi_{1}$ such that $\chi_{2}\chi_{1}^{-1}\notin \subspaceperp$.
\end{proof}
We now use Lemma~\ref{lem: povm dual cosets} to diagonalize the operators $\{\Pi_{\subspace,\bs}\}_{\subspace\in \codesubspaces^{\localcode},\bs\in \Image(H_{\subspace})}$. This lemma will play a key role for our LP reduction proof which we discuss in more detail in Section~\ref{sec:SDP to LP}.
For the remainder of this subsection, we assume that the Gram matrix $G$ is full rank, so that
$\kappa_{\chi}>0$ for all $\chi\in\wC$ and the normalized character basis is well defined.
The rank-deficient case is handled separately in Section~\ref{sec: rank deficient}.
\begin{lemma}\label{lem: povm zeta diagonal form}
    There exists a POVM $\fullpovm$, which achieves the optimal objective value for the SDP in \eqref{SDP main} such that the POVM elements admit the following diagonal form 
    \begin{align*}
        \Pi_{\subspace,\bs}=\sum_{\subspaceperp_{m}\in \wC/\subspaceperp}\beta_{\subspace,\bs}(\subspaceperp_{m})\ket{\zeta_{\subspace,\bs}^{\subspaceperp_{m}}}\bra{\zeta_{\subspace,\bs}^{\subspaceperp_{m}}},
    \end{align*}
    where $\beta_{\subspace,\bs}(\subspaceperp_{m})\geq 0$ $\forall\subspace\in \codesubspaces^{\localcode},\bs\in\Image(H_{\subspace}), \subspaceperp_{m}\in\wC/\subspaceperp$ 
    and the orthonormal vector $\ket{\zeta_{\subspace,\bs}^{\subspaceperp_{m}}}$ associated with the coset
$\subspaceperp_m\in \dualcosetspace$ is represented
    in the character basis by
    \begin{align}
        \ket{\zeta_{\subspace,\bs}^{\subspaceperp_{m}}}=\frac{1}{\sqrt{\eta_{\subspace}(\subspaceperp_{m})}}\sum_{\chi\in \subspaceperp_m}\frac{\chi(-\ba)}{\sqrt{\kappa_{\chi}}}\ket{\chi},
    \end{align}
    for some codeword $\ba\in \subspace_{\bs}$ and $\eta_{\subspace}(\subspaceperp_{m})=\sum_{\chi\in \subspaceperp_m }\frac{1}{\kappa_{\chi}}$.
\end{lemma}
\begin{proof}
For $\subspace\in \codesubspaces^{\localcode},\bs\in\Image(H_{\subspace}), \subspaceperp_{m}\in\wC/\subspaceperp$ and a codeword $\ba\in \subspace_{\bs}$, define the vector
\begin{align}
    \ket{\zeta_{\subspace,\bs}^{\subspaceperp_m,\ba}}\coloneqq\frac{1}{\sqrt{\eta_{\subspace}(\subspaceperp_m)}}\sum_{\chi\in \subspaceperp_m}\frac{\chi(-\ba)}{\sqrt{\kappa_{\chi}}}\ket{\chi}.
\end{align}
Consider $\ba_{1},\ba_{2}\in \subspace_{\bs}$  and $\chi_{1},\chi_{2}\in \subspaceperp_m$ with $\bs\in \Image(H_{\subspace})$ and $\subspaceperp_m\in \dualcosetspace$ such that $\ba_2-\ba_1=\bv\in \subspace$ and  $\chi_{2}\chi_{1}^{-1}=\chi'\in \subspaceperp$.
Then the characters satisfy 
\begin{align*}
    \frac{\chi_{1}(-\ba_1)}{\chi_{2}(-\ba_1)} &\ = \chi_{2}(\ba_1)\chi_{1}(-\ba_1)\\
    &\ = \chi_{2}\chi_{1}^{-1}(\ba_1)= \chi'(\ba_1).
\end{align*}
Similarly, we obtain
\begin{align*}
    \frac{\chi_{1}(-\ba_2)}{\chi_{2}(-\ba_2)} = \chi'(\ba_2).
\end{align*}
Moreover, from $\ba_{2}=\bv+\ba_{1}$ this yields
\begin{align*}
    \chi'(\ba_{2})=\chi'(\ba_{1}+\bv)=\chi'(\ba_1)\chi'(\bv).
\end{align*}
Since $\bv\in \subspace$ and $\chi'\in \subspaceperp$, $\chi'(\bv)=1$, which implies $\chi'(\ba_{2})=\chi'(\ba_{1})$. Since it holds for any $\ba_{1},\ba_{2}\in \subspace_{\bs}$  and $\chi_{1},\chi_{2}\in \subspaceperp_m$ with $\bs\in \Image(H_{\subspace})$ and $\subspaceperp_m\in \dualcosetspace$, the vectors $\ket{\zeta_{\subspace,\bs}^{\subspaceperp_m,\ba}}$ for fixed $\subspace\in \codesubspaces^{\localcode},\bs\in\Image(H_{\subspace}), \subspaceperp_m\in \dualcosetspace$ and different codewords $\ba\in \subspace_{\bs}$ differ from each other by a character factor. So for fixed $\subspace\in \codesubspaces^{\localcode},\bs\in\Image(H_{\subspace}), \subspaceperp_m\in \dualcosetspace$, we define the vector $\ket{\zeta_{\subspace,\bs}^{\subspaceperp_m}}$ as 
\begin{align*}
   \ket{\zeta_{\subspace,\bs}^{\subspaceperp_m}} =\frac{1}{\sqrt{\eta_{\subspace}(\subspaceperp_{m})}}\sum_{\chi\in \subspaceperp_m}\frac{\chi(-\ba)}{\sqrt{\kappa_{\chi}}}\ket{\chi},
\end{align*}
 for some arbitrary codeword $\ba\in \subspace_{\bs}$. Consider the vectors $\ket{\zeta_{\subspace,\bs}^{\subspaceperp_{m_1}}}$ and $\ket{\zeta_{\subspace,\bs}^{\subspaceperp_{m_2}}}$, for some $\subspace\in \codesubspaces^{\localcode},\bs\in\Image(H_{\subspace}), \subspaceperp_{m_1},\subspaceperp_{m_2}\in \dualcosetspace$. Then, for a codeword $\ba\in \subspace_{\bs}$, they satisfy the following relation
\begin{align}\label{eq:zeta relation}
    \braket{\zeta_{\subspace,\bs}^{\subspaceperp_{m_1}}|\zeta_{\subspace,\bs}^{\subspaceperp_{m_2}}} &\ = \frac{1}{\sqrt{\eta_{\subspace}(\subspaceperp_{m_1})\eta_{\subspace}(\subspaceperp_{m_2})}}\sum_{\chi_{1}\in\subspaceperp_{m_1}}\sum_{\chi_{2}\in\subspaceperp_{m_2}}\frac{\chi_{1}(\ba)\chi_{2}(-\ba)}{\sqrt{\kappa_{\chi_{1}}\kappa_{\chi_{2}}}}\braket{\chi_{1}|\chi_{2}}.
\end{align}
When $m_1\neq m_{2}$, the sets $\subspaceperp_{m_1}$ and $\subspaceperp_{m_2}$ are disjoint which implies $\braket{\chi_{1}|\chi_{2}}=0$ $\forall \chi_{1}\in\subspaceperp_{m_1},\chi_{2}\in\subspaceperp_{m_2}$. Thus \eqref{eq:zeta relation} reduces to 
\begin{align}
  \braket{\zeta_{\subspace,\bs}^{\subspaceperp_{m_1}}|\zeta_{\subspace,\bs}^{\subspaceperp_{m_2}}} &\ = \frac{1}{\eta_{\subspace}(\subspaceperp_{m_1})}\sum_{\chi_{1}\in\subspaceperp_{m_1}}  \frac{1}{\kappa_{\chi_{1}}}\delta_{m_1,m_2}.
\end{align}
Setting $\eta_{\subspace}(\subspaceperp_{m})=\sum_{\chi\in \subspaceperp_{m}}\frac{1}{\kappa_{\chi}}$, we obtain
\begin{align}\label{eq:zeta orthogonality}
    \braket{\zeta_{\subspace,\bs}^{\subspaceperp_{m_1}}|\zeta_{\subspace,\bs}^{\subspaceperp_{m_2}}} = \delta_{ m_1,m_2}.
\end{align}
Thus $\{\ket{\zeta_{\subspace,\bs}^{\subspaceperp_m}}\}_{\subspaceperp_m\in \dualcosetspace}$ forms a collection of orthonormal vectors.
Observe that for a codeword $\bc\in \localcode$, the state $\ket{\psi_{\bc}}$ is represented in character basis $\characterbasis$  as follows
\begin{align}
    \ket{\psi_{\bc}}= \sum_{\chi\in \wC}\frac{\sqrt{\kappa_{\chi}}}{\sqrt{|\localcode|}}\chi^{-1}(\bc)\ket{\chi}.
\end{align}
Therefore, the inner product $\braket{\psi_{\bc}|\zeta_{\subspace,\bs}^{\subspaceperp_{m}}}$ satisfies 
\begin{align}\label{eq: psi zeta inner product}
    \braket{\psi_{\bc}|\zeta_{\subspace,\bs}^{\subspaceperp_{m}}} &\ =  \frac{1}{\sqrt{\eta_{\subspace}(\subspaceperp_m)|\localcode|}}\sum_{\chi\in \wC}\sum_{\chi'\in\subspaceperp_m}\frac{\chi(\bc)\chi'(-\ba)\sqrt{\kappa_{\chi}}}{\sqrt{\kappa_{\chi'}}}\braket{\chi|\chi'}\\
    &\ = \frac{1}{\sqrt{\eta_{\subspace}(\subspaceperp_m)|\localcode|}}\sum_{\chi'\in\subspaceperp_m}\chi'(\bc-\ba).
\end{align}
Then for some character $\tilde{\chi}\in \subspaceperp_m$, we obtain
\begin{align}\label{eq: dual coset sum}
    \braket{\psi_{\bc}|\zeta_{\subspace,\bs}^{\subspaceperp_{m}}} &\ = \frac{\tilde{\chi}(\bc-\ba)}{\sqrt{\eta_{\subspace}(\subspaceperp_m)|\localcode|}}\sum_{\chi''\in \subspaceperp}\chi''(\bc-\ba).
\end{align}
Using the definition of $\subspaceperp$, it follows that
\begin{align}
    \sum_{\chi''\in \subspaceperp}\chi''(\bc-\ba)= \begin{cases}
      &  |\subspaceperp|, \quad \text{if $\bc-\ba\in \subspace$}\\
      & 0, \quad \text{otherwise}.
    \end{cases}
\end{align}
Thus, $\braket{\psi_{\bc}|\zeta_{\subspace,\bs}^{\subspaceperp_m}}=0$ whenever $\bc-\ba\notin \subspace$ for $\ba\in \subspace_{\bs}$ which implies $\braket{\psi_{\bc}|\zeta_{\subspace,\bs}^{\subspaceperp_m}}=0$, for $\bc\notin \subspace_{\bs}$. 

\noindent Therefore, every orthonormal vector $\ket{\zeta_{\subspace,\bs}^{\subspaceperp_m}}$ lies in $\mathcal{W}_{\sim \subspace_{\bs}}^{\perp}$ $\forall\subspace\in \codesubspaces^{\localcode},\bs\in\Image(H_{\subspace}), \subspaceperp_m\in\dualcosetspace$. Since each vector $\ket{\zeta_{\subspace,\bs}^{\subspaceperp_m}}$ is spanned by character states $\characterbasis$, each orthonormal vector $\ket{\zeta_{\subspace,\bs}^{\subspaceperp_m}}\in \mathcal{W}_{\sim \subspace_{\bs}}^{\perp}\cap \mathcal{W}$ $\forall\subspace\in \codesubspaces^{\localcode},\bs\in\Image(H_{\subspace}), \subspaceperp_m\in \dualcosetspace$. 
Under the full-rank assumption on $G$, the states
$\{\ket{\psi_{\bc}}\}_{\bc\in\localcode}$ are linearly independent, and hence
the dimension of $\mathcal{W}_{\sim \subspace_{\bs}}^{\perp}\cap \mathcal{W}$ is  $\dim(\mathcal{W}_{\sim \subspace_{\bs}}^{\perp}\cap \mathcal{W})=q^{\dim(\subspace)}$. As  $\zetabasis$ forms a collection of orthonormal vectors and there are $q^{\dim(\subspace)}$ such vectors, $\zetabasis$ forms an orthonormal basis for $\mathcal{W}_{\sim \subspace_{\bs}}^{\perp}\cap \mathcal{W}$.
For the vectors $\ket{\zeta_{\subspace,\bs}^{\subspaceperp_{m_1}}}$ and $\ket{\zeta_{\subspace,\bs}^{\subspaceperp_{m_2}}}$, for some $\subspace\in \codesubspaces^{\localcode},\bs\in\Image(H_{\subspace}), \subspaceperp_{m_1},\subspaceperp_{m_2}\in\dualcosetspace$ the inner product $\bra{\zeta_{\subspace,\bs}^{\subspaceperp_{m_1}}}\Pi_{\subspace,\bs}\ket{\zeta_{\subspace,\bs}^{\subspaceperp_{m_2}}} $ satisfies 
\begin{align*}
    &\ \bra{\zeta_{\subspace,\bs}^{\subspaceperp_{m_1}}}\Pi_{\subspace,\bs}\ket{\zeta_{\subspace,\bs}^{\subspaceperp_{m_2}}} \\
    &\ = \frac{1}{\sqrt{\eta_{\subspace}(\subspaceperp_{m_1})\eta_{\subspace}(\subspaceperp_{m_2})}}\sum_{\chi_{1}\in\subspaceperp_{m_1}}\sum_{\chi_{2}\in\subspaceperp
    _{m_2}}\frac{\chi_{1}(\ba)\chi_{2}(-\ba)}{\sqrt{\kappa_{\chi_{1}}\kappa_{\chi_{2}}}}\bra{\chi_{1}}\Pi_{\subspace,\bs}\ket{\chi_{2}}.
\end{align*}
When $m_{1}\neq m_{2}$, it satisfies $\chi_{2}\chi_{1}^{-1}\notin \subspaceperp$ $\forall \chi_{1}\in\subspaceperp_{m_1},\chi_{2}\in\subspaceperp_{m_2}$. Hence, using Lemma~\ref{lem: povm dual cosets}, this gives  $\bra{\chi_{1}}\Pi_{\subspace,\bs}\ket{\chi_{2}}=0$ $\forall \chi_{1}\in\subspaceperp_{m_1},\chi_{2}\in\subspaceperp_{m_2}$ with $m_1\neq m_2$. Since $\Pi_{\subspace,\bs}\in \mathcal{L}(\mathcal{W}^{\perp}_{\sim \subspace_{\bs}}\cap \mathcal{W})$, the POVM element $\Pi_{\subspace,\bs}$ is diagonal in the basis $\zetabasis$. Thus we can write each POVM element $\Pi_{\subspace,\bs}$ for $\subspace\in \codesubspaces^{\localcode},\bs\in\Image(H_{\subspace})$ as 
\begin{align}\label{eq: povm diagonal structure}
   \Pi_{\subspace,\bs}=\sum_{\subspaceperp_{m}\in \wC/\subspaceperp}\beta_{\subspace,\bs}(\subspaceperp_{m})\ket{\zeta_{\subspace,\bs}^{\subspaceperp_{m}}}\bra{\zeta_{\subspace,\bs}^{\subspaceperp_{m}}}.
\end{align}
Choosing $\beta_{\subspace,\bs}(\subspaceperp_{m})\geq 0$ for all $\subspaceperp_m\in \dualcosetspace$ makes $\Pi_{\subspace,\bs}$ PSD and completes the proof.

\end{proof}
\subsection{Obtaining the LP from SDP}\label{sec:SDP to LP}
Using the diagonal form of the operators $\{\Pi_{\subspace,\bs}\}_{\subspace\in\codesubspaces^{\localcode},\bs\in \Image(H_{\subspace})}$, we now reduce the SDP in \eqref{SDP main} to an LP.
Consider scalars $\beta_{\subspace}(\subspaceperp_m)$ defined  for each $\subspace\in \codesubspaces^{\localcode}$ and $\subspaceperp_m\in\dualcosetspace $. Define the function \\
$\flp(\{\beta_{\subspace}(\subspaceperp_m)\}_{\subspace\in \codesubspaces^{\localcode},\subspaceperp_m\in\dualcosetspace })$ as follows
\begin{align}
   \flp(\{\beta_{\subspace}(\subspaceperp_m)\}_{\subspace\in \codesubspaces^{\localcode},\subspaceperp_m\in\dualcosetspace })\coloneqq  \sum_{\subspace\in \codesubspaces^{\localcode}}\sum_{\subspaceperp_m\in \dualcosetspace}\beta_{\subspace}(\subspaceperp_m)\frac{q^{2(\dim(\localcode)-\dim(\subspace))}}{\eta_{\subspace}(\subspaceperp_m)|\localcode|}\reward(\subspace).
\end{align}
Next, we define the set of feasible scalar collections $\{\beta_{\subspace}(\subspaceperp_m)\}_{\subspace\in \codesubspaces^{\localcode},\subspaceperp_m\in\dualcosetspace }$, denoted by $\lpfeasible$,  satisfying:
\begin{align}
 &\   \lpfeasible \\
 &\ \coloneqq  \left\{ \{\beta_{\subspace}(\subspaceperp_m)\}_{\subspace\in \codesubspaces^{\localcode},\subspaceperp_m\in\dualcosetspace } \;\middle|\; 
    \begin{aligned}
        & \sum_{\subspace\in \codesubspaces^{\localcode}}\sum_{\subspaceperp_m\in\dualcosetspace}\frac{\beta_{\subspace}(\subspaceperp_m)q^{\dim(\localcode)-\dim(\subspace)}\mone(\chi\in \subspaceperp_m)}{\eta_{\subspace}(\subspaceperp_m)\kappa_{\chi}}\leq 1, \forall \chi\in \wC \\
        & \text{ where $\eta_{\subspace}(\subspaceperp_m)=\sum_{\chi\in \subspaceperp_m}\frac{1}{\kappa_{\chi}}$, $\kappa_{\chi}=\sum_{\bc\in \localcode}g(\bc)\chi(\bc)$, $g(\bc)=\braket{\psi_{\mzero}|\psi_{\bc}}$ } \\
        & \beta_{\subspace}(\subspaceperp_m)\geq 0, \qquad \forall \subspace\in \codesubspaces^{\localcode}, \subspaceperp_m\in \dualcosetspace
    \end{aligned}
    \right\}\nonumber.
\end{align}
Finally, we define the optimal LP value, $\objlp(\lpfeasible)$, as the maximization of the objective function over the feasible set:
\begin{align}
    &\ \objlp(\lpfeasible) \\
    &\ \coloneqq \text{Max} \left\{ \flp\left(\fullbeta\right) \;\middle|\; \fullbeta\in \lpfeasible \right\}\nonumber.
\end{align}
\begin{theorem}\label{thm:affine usd sdp-lp equivalence}
    For a local code $\localcode\subset\mbF^{|I|}_{q}$ supported on the index set $I$, and a collection of quantum states $\{\ket{\psi_{\bc}}\}_{\bc\in \localcode}$ defined in \eqref{def: quantum state}, for any non-negative reward function $\reward(\subspace)$ for all $\subspace\in\codesubspaces^{\localcode}$ and $\kappa_{\chi}>0$ $\forall\chi\in\wC$, the following equality holds
    \begin{align}
        \objsdp(\sdpfeasible)=\objlp(\lpfeasible).
    \end{align}
Moreover, the optimal POVM $\fullpovmstar$, which attains the SDP objective\\
$\objsdp(\sdpfeasible)$, can be constructed from optimal scalars
$\fullbetastar$, which attain the LP objective $\objlp(\lpfeasible)$, such that
\begin{align*}
      \Pi_{\subspace,\bs}^* =\sum_{\subspaceperp_{m}\in \wC/\subspaceperp}\beta_{\subspace}^*(\subspaceperp_{m})\ket{\zeta_{\subspace,\bs}^{\subspaceperp_{m}}}\bra{\zeta_{\subspace,\bs}^{\subspaceperp_{m}}},
\end{align*}
for all $\subspace\in \codesubspaces^{\localcode}$ and $\bs\in \Image(H_{\subspace})$, where the vector $\ket{\zeta_{\subspace,\bs}^{\subspaceperp_{m}}}$ is defined in Lemma~\ref{lem: povm zeta diagonal form}. The matrix $\Pi_{e}^*$ is formed as
\begin{align*}
\Pi_{e}^*=\mI-\sum_{\subspace\in \codesubspaces^{\localcode}}\Pi_{\subspace}^*,
\end{align*}
where $\Pi_{\subspace}^*=\sum_{\bs\in \Image(H_{\subspace})}\Pi_{\subspace,\bs}^*$.
\end{theorem}
\begin{proof}
    Since $\kappa_{\chi}>0$ for all $\chi\in\wC$, the normalized character basis $\{\ket{\chi}\}_{\chi\in\wC}$ and the quantities $\eta_{\subspace}(\subspaceperp_m)$ from Lemma~\ref{lem: povm zeta diagonal form} are well defined. By Lemma~\ref{lem :coset symmetry}, we may restrict attention to an optimal 
    affine filtering POVM satisfying the coset symmetry relation
\begin{align*}
    \Pi_{\subspace,\bs_2}=X^{\bw_2-\bw_1}\Pi_{\subspace,\bs_1}X^{-(\bw_2-\bw_1)}
\end{align*}
for all $\subspace\in \codesubspaces^{\localcode}$, $\bs_1,\bs_2\in \Image(H_{\subspace})$, $\bw_1\in \subspace_{\bs_1}$, and $\bw_2\in \subspace_{\bs_2}$. By Lemma~\ref{lem: povm zeta diagonal form}, each informative POVM element admits an expansion in the basis $\{\ket{\zeta_{\subspace,\bs}^{\subspaceperp_m}}\}_{\subspaceperp_m\in \dualcosetspace}$.
For $\bw_{1},\ba_1\in \subspace_{\bs_1}$ and $\bw_{2}\in \subspace_{\bs_2}$ we have
\begin{align}
    X^{\bw_{2}-\bw_{1}}\ket{\zeta_{\subspace,\bs_1}^{\subspaceperp_m}} &\ = \frac{1}{\sqrt{\eta_{\subspace}(\subspaceperp_m)}}\sum_{\chi\in \subspaceperp_m}\frac{\chi(-\ba_1)}{\sqrt{\kappa_{\chi}}}X^{\bw_{2}-\bw_{1}}\ket{\chi}\\
    &\ = \frac{1}{\sqrt{\eta_{\subspace}(\subspaceperp_m})}\sum_{\chi\in \subspaceperp_m}\frac{\chi(-\ba_1-\bw_2+\bw_1)}{\sqrt{\kappa_{\chi}}}\ket{\chi}.
\end{align}
Therefore conjugation by $X^{\bw_2-\bw_1}$ maps each rank-one projector $\ket{\zeta_{\subspace,\bs_1}^{\subspaceperp_m}}\bra{\zeta_{\subspace,\bs_1}^{\subspaceperp_m}}$ to  $\ket{\zeta_{\subspace,\bs_2}^{\subspaceperp_m}}\bra{\zeta_{\subspace,\bs_2}^{\subspaceperp_m}}.$
For a coset-symmetrized optimal POVM, the coefficients in the expansion of $\Pi_{\subspace,\bs}$ in this basis are independent of $\bs$. Therefore, for every $\subspace\in \codesubspaces^{\localcode}$ and every $\subspaceperp_m\in \dualcosetspace$, we may write
\begin{align*}
    \beta_{\subspace,\bs}(\subspaceperp_m)=\beta_{\subspace}(\subspaceperp_m), \qquad \forall \bs\in \Image(H_{\subspace}).
\end{align*}

\noindent For a character vector $\ket{\chi}$ with $\chi\in \subspaceperp_m$ and for some $\subspace\in \codesubspaces^{\localcode},\bs\in\Image(H_{\subspace}), \subspaceperp_{m}\in \dualcosetspace$, codeword $\ba\in \subspace_{\bs}$, the inner product satisfies
\begin{align}\label{eq:character zeta relation}
    \braket{\chi|\zeta_{\subspace,\bs}^{\subspaceperp_m}}=\frac{\chi(-\ba)}{\sqrt{\eta_{\subspace}(\subspaceperp_{m})\kappa_{\chi}}}.
\end{align}
Since the POVM elements $\Pi_{\subspace,\bs}$ for some subspace $\subspace\in\codesubspaces^{\localcode}$ satisfy the coset symmetry relation $\forall\bs\in \Image(H_{\subspace})$, using Lemma~\ref{lem: povm X symmetry relations}, we immediately get $\Pi_{\subspace}=X^{\bc}\Pi_{\subspace}X^{-\bc}$ for all $\subspace\in \codesubspaces^{\localcode}$ and $\bc\in \localcode$. Using Lemma~\ref{lem: povm subspace diagonal character basis},
 the operator $\Pi_{\subspace}$ has the diagonal form $ \Pi_{\subspace}=\sum_{\chi\in \wC}\gamma_{\subspace}(\chi)\ket{\chi}\bra{\chi}$ for all $\subspace\in \codesubspaces^{\localcode}$. From \eqref{eq:character zeta relation}, we can compute the expression for $\gamma_{\subspace}(\chi)$ for all $\subspace\in \codesubspaces^{\localcode}$ and $\chi\in \wC$ as follows
\begin{align}
    \gamma_{\subspace}(\chi) &\ =\sum_{\bs\in \Image(H_{\subspace})}\bra{\chi}\Pi_{\subspace,\bs}\ket{\chi}\\
    &\ = \sum_{\bs\in \Image(H_{\subspace})}\sum_{\subspaceperp_{m}\in\dualcosetspace}\beta_{\subspace}(\subspaceperp_m)|\braket{\chi|\zeta_{\subspace,\bs}^{\subspaceperp_m}}|^{2}\\
    &\ = |\Image(H_{\subspace})|\sum_{\subspaceperp_{m}\in \dualcosetspace}\frac{\beta_{\subspace}(\subspaceperp_m)\mone(\chi\in \subspaceperp_m)}{\eta_{\subspace}(\subspaceperp_m)\kappa_{\chi}}\\
    &\ = \sum_{\subspaceperp_m\in \dualcosetspace}\frac{\beta_{\subspace}(\subspaceperp_m)|\Image(H_{\subspace})|\mone(\chi\in \subspaceperp_m)}{\eta_{\subspace}(\subspaceperp_m)\kappa_{\chi}}.
\end{align}
Since $|\Image(H_{\subspace})|=q^{\dim(\localcode)-\dim(\subspace)}$, we obtain 
\begin{align}\label{eq: gamma expression wrt dual cosets}
    \gamma_{\subspace}(\chi)= \sum_{\subspaceperp_m\in \dualcosetspace}\frac{\beta_{\subspace}(\subspaceperp_m)q^{\dim(\localcode)-\dim(\subspace)}\mone(\chi\in \subspaceperp_m)}{\eta_{\subspace}(\subspaceperp_m)\kappa_{\chi}}.
\end{align}
The diagonal coefficients of $\Pi_{\subspace}$ in the character basis are completely determined by the scalars $\{\beta_{\subspace}(\subspaceperp_m)\}$. By Lemma~\ref{lem: pi_e diagonal form}, positivity of the inconclusive operator $\Pi_e$ on the span of codeword states is equivalent to requiring
\begin{align*}
    1-\sum_{\subspace\in \codesubspaces^{\localcode}}\gamma_{\subspace}(\chi)\geq 0, \qquad \forall \chi\in \wC.
\end{align*}
Substituting \eqref{eq: gamma expression wrt dual cosets} gives the LP feasibility constraint.
This implies
\begin{align}
    \sum_{\subspace\in \codesubspaces^{\localcode}}\sum_{\subspaceperp_m\in\dualcosetspace}\frac{\beta_{\subspace}(\subspaceperp_m)q^{\dim(\localcode)-\dim(\subspace)}\mone(\chi\in \subspaceperp_m)}{\eta_{\subspace}(\subspaceperp_m)\kappa_{\chi}}\leq 1, \qquad \forall \chi\in \wC.
\end{align}
When a codeword $\bc\in \subspace_{\bs}$, using \eqref{eq: psi zeta inner product}  and \eqref{eq: dual coset sum}, this gives
\begin{align}
    \bra{\psi_{\bc}}\Pi_{\subspace,\bs}\ket{\psi_{\bc}} &\ = \sum_{\subspaceperp_m\in \dualcosetspace}\beta_{\subspace}(\subspaceperp_m)\frac{|\subspaceperp|^{2}}{\eta_{\subspace}(\subspaceperp_m)|\localcode|}.
\end{align}

\noindent For fixed $\subspace\in \codesubspaces^{\localcode}$, the above quantity is the same for every syndrome $\bs$ whenever the codeword $\bc$ lies in $\subspace_{\bs}$. In particular, when $\bc=\mzero$, the unambiguity condition implies
\begin{align*}
    \bra{\psi_{\mzero}}\Pi_{\subspace,\bs}\ket{\psi_{\mzero}}=0,
    \qquad \forall \bs\neq \mzero,
\end{align*}
because $\mzero\notin \subspace_{\bs}$ whenever $\bs\neq \mzero$. Therefore
\begin{align}
    \Tr(\Pi_{\subspace}\ket{\psi_{\mzero}}\bra{\psi_{\mzero}})
    &\ =\sum_{\bs\in \Image(H_{\subspace})}\bra{\psi_{\mzero}}\Pi_{\subspace,\bs}\ket{\psi_{\mzero}} \\
    &\ = \bra{\psi_{\mzero}}\Pi_{\subspace,\mzero}\ket{\psi_{\mzero}} \\
    &\ = \sum_{\subspaceperp_m\in \dualcosetspace}\beta_{\subspace}(\subspaceperp_m)\frac{|\subspaceperp|^{2}}{\eta_{\subspace}(\subspaceperp_m)|\localcode|}.
\end{align}
Using $|\subspaceperp|=q^{\dim(\localcode)-\dim(\subspace)}$, we finally obtain 
\begin{align}\label{eq: lp probablity success relation for subspace}
    \Tr(\Pi_{\subspace}\ket{\psi_{\mzero}}\bra{\psi_{\mzero}})= \sum_{\subspaceperp_m\in \dualcosetspace}\beta_{\subspace}(\subspaceperp_m)\frac{q^{2(\dim(\localcode)-\dim(\subspace))}}{\eta_{\subspace}(\subspaceperp_m)|\localcode|}.
\end{align}
Every feasible coset-symmetrized 
affine filtering POVM determines a feasible collection of LP variables $\{\beta_{\subspace}(\subspaceperp_m)\}_{\subspace,\subspaceperp_m}$ with the same objective value. This shows
\begin{align*}
    \objsdp(\sdpfeasible)\leq \objlp(\lpfeasible).
\end{align*}
Conversely, given any feasible LP variables $\{\beta_{\subspace}(\subspaceperp_m)\}_{\subspace,\subspaceperp_m}\in \lpfeasible$, define
\begin{align*}
    \Pi_{\subspace,\bs}:=\sum_{\subspaceperp_m\in \dualcosetspace}\beta_{\subspace}(\subspaceperp_m)\ket{\zeta_{\subspace,\bs}^{\subspaceperp_m}}\bra{\zeta_{\subspace,\bs}^{\subspaceperp_m}},
\end{align*}
for each $\subspace\in \codesubspaces^{\localcode}$ and $\bs\in \Image(H_{\subspace})$, and let
\begin{align*}
    \Pi_{\subspace}:=\sum_{\bs\in \Image(H_{\subspace})}\Pi_{\subspace,\bs},
    \qquad
    \Pi_e:=\mI-\sum_{\subspace\in \codesubspaces^{\localcode}}\Pi_{\subspace}.
\end{align*}
Since $\beta_{\subspace}(\subspaceperp_m)\geq 0$, each operator
$\Pi_{\subspace,\bs}$ is PSD. By
Lemma~\ref{lem: povm zeta diagonal form}, every vector
$\ket{\zeta_{\subspace,\bs}^{\subspaceperp_m}}$ lies in
$\mathcal{W}^{\perp}_{\sim \subspace_{\bs}}\cap \mathcal{W}$, so the 
affine filtering
unambiguity constraints hold. Moreover, the LP feasibility constraint is exactly the
condition from Lemma~\ref{lem: pi_e diagonal form} guaranteeing
\[
P_{\mathcal{W}}\Pi_e P_{\mathcal{W}}\succeq \mzero.
\]
Since all informative POVM elements act on $\mathcal{W}$, the operator $\Pi_e$
acts as the identity on $\mathcal{W}^{\perp}$. Hence $\Pi_e\succeq \mzero$.
These operators therefore form a feasible POVM attaining the LP objective, which proves
\begin{align*}
    \objlp(\lpfeasible)\leq \objsdp(\sdpfeasible).
\end{align*}
Combining the two inequalities proves the theorem.
Consequently the SDP optimization problem stated in \eqref{SDP main} can be converted into an LP which is formally stated below.

\begin{LPbox}
For a local code $\localcode\subset\mbF^{|I|}_{q}$ supported on the index set $I$, and a collection of quantum states $\{\ket{\psi_{\bc}}\}_{\bc\in \localcode}$ defined in \eqref{def: quantum state}, suppose we want to find the optimal POVM $\fullpovmstar$ for the affine filtering measurement defined in Def~\ref{def:affineUSD}. Then, this problem can be solved using the following LP with variables $\{\beta_{\subspace}(\subspaceperp_m)\}_{\subspace,\subspaceperp_m}$ 
\begin{align}\label{LP: final form}
        \text{Maximize}  \quad &\   \sum_{\subspace\in \codesubspaces^{\localcode}}\sum_{\subspaceperp_m\in \dualcosetspace}\beta_{\subspace}(\subspaceperp_m)\frac{q^{2(\dim(\localcode)-\dim(\subspace))}}{\eta_{\subspace}(\subspaceperp_m)|\localcode|}\reward(\subspace)\\
 \text{Subject to:} \quad 
    &\ \sum_{\subspace\in \codesubspaces^{\localcode}}\sum_{\subspaceperp_m\in\dualcosetspace}\frac{\beta_{\subspace}(\subspaceperp_m)q^{\dim(\localcode)-\dim(\subspace)}\mone(\chi\in \subspaceperp_m)}{\eta_{\subspace}(\subspaceperp_m)\kappa_{\chi}}\leq 1, \qquad \forall \chi\in \wC\\
    &\ \beta_{\subspace}(\subspaceperp_m)\geq 0, \qquad \forall \subspace\in \codesubspaces^{\localcode}, \subspaceperp_m\in \dualcosetspace \nonumber
    \end{align}
    where $\eta_{\subspace}(\subspaceperp_m)=\sum_{\chi\in \subspaceperp_m}\frac{1}{\kappa_{\chi}}$, $\kappa_{\chi}=\sum_{\bc\in \localcode}g(\bc)\chi(\bc)$, $g(\bc)=\braket{\psi_{\mzero}|\psi_{\bc}}$.
\end{LPbox}
\end{proof}

 \noindent  Define the quantity $\theta_{\subspace}(\subspaceperp_m)=\beta_{\subspace}(\subspaceperp_m)\frac{q^{\dim(\localcode)-\dim(\subspace)}}{\eta_{\subspace}(\subspaceperp_m)}$. Then the LP primal form reduces to 
    \begin{align}\label{eq:theta-LP form}
        \text{Maximize}  \quad &\   \sum_{\subspace\in \codesubspaces^{\localcode}}\sum_{\subspaceperp_m\in \dualcosetspace}\theta_{\subspace}(\subspaceperp_m)\frac{q^{\dim(\localcode)-\dim(\subspace)}}{|\localcode|}\reward(\subspace)\\
 \text{Subject to:} \quad 
    &\ \sum_{\subspace\in \codesubspaces^{\localcode}}\sum_{\subspaceperp_m\in\dualcosetspace:\chi\in \subspaceperp_m}\theta_{\subspace}(\subspaceperp_m)\leq \kappa_{\chi}, \qquad \forall \chi\in \wC\\
    &\ \theta_{\subspace}(\subspaceperp_m)\geq 0, \qquad \forall \subspace\in \codesubspaces^{\localcode}, \subspaceperp_m\in \dualcosetspace.
    \end{align}
We denote this LP by $ \text{LP}\Big(
    \{\theta_{\subspace}(\subspaceperp_m)\}_{
    \subspace\in \codesubspaces^{\localcode},
    \subspaceperp_m\in \dualcosetspace},
    \{\kappa_{\chi}\}_{\chi\in \wC}
    \Big)$ and its optimal value by \\
    $\objlp(\{\theta_{\subspace}(\subspaceperp_m)\}_{
    \subspace\in\codesubspaces^{\localcode},
    \subspaceperp_m\in\dualcosetspace},
    \{\kappa_{\chi}\}_{\chi\in\wC})$. Then from Theorem~\ref{thm:affine usd sdp-lp equivalence}
    \begin{align*}
        \objlp(\lpfeasible)= \objlp(\{\theta_{\subspace}(\subspaceperp_m)\}_{
    \subspace\in\codesubspaces^{\localcode},
    \subspaceperp_m\in\dualcosetspace},
    \{\kappa_{\chi}\}_{\chi\in\wC}).
    \end{align*}
\subsection{Finding the Dual LP}
By standard LP duality, the primal LP in Eq.~\eqref{LP: final form} has the following dual form.
Consider dual variables $\dualvariable_{\chi}\geq 0$ for each $\chi\in \wC$. Then dual problem is 
\begin{align}
    \text{Minimize}  \quad \sum_{\chi\in \wC}\dualvariable_{\chi}.
\end{align}
For each subspace $\subspace\in \codesubspaces^{\localcode}$ and coset $\subspaceperp_m\in \dualcosetspace$, the corresponding dual constraint is
\begin{align}
   \sum_{\chi\in \wC} \frac{q^{\dim(\localcode)-\dim(\subspace)}\mone(\chi\in \subspaceperp_m)}{\eta_{\subspace}(\subspaceperp_m)\kappa_{\chi}}\dualvariable_{\chi}\geq \frac{q^{2(\dim(\localcode)-\dim(\subspace))}}{\eta_{\subspace}(\subspaceperp_m)|\localcode|}\reward(\subspace)
\end{align}
which can be further reduced to 
\begin{align}
    \sum_{\chi\in\subspaceperp_m}\frac{\dualvariable_{\chi}}{\kappa_{\chi}}\geq \frac{q^{\dim(\localcode)-\dim(\subspace)}\reward(\subspace)}{|\localcode|}.
\end{align}
Since $|\localcode|=q^{\dim(\localcode)}$, the final dual LP form can be written as 
\begin{LPdualbox}
    \begin{align}
    \text{Minimize}  \quad &\  \sum_{\chi\in \wC}\dualvariable_{\chi}\\
     \text{Subject to:} \quad &\ \sum_{\chi\in\subspaceperp_m}\frac{\dualvariable_{\chi}}{\kappa_{\chi}}\geq \frac{\reward(\subspace)}{q^{\dim(\subspace)}}, \quad \forall\subspace\in \codesubspaces^{\localcode}, \subspaceperp_m\in \dualcosetspace.\\
     &\ \dualvariable_{\chi}\geq 0,\quad \forall\chi\in \wC
\end{align}
where $\kappa_{\chi}=\sum_{\bc\in \localcode}g(\bc)\chi(\bc)$, $g(\bc)=\braket{\psi_{\mzero}|\psi_{\bc}}$.
\end{LPdualbox}

\subsection{Optimal Objective for Two Level Gram Matrix Eigenvalue Spectrum}\label{sec:SDP-LP fgum special case}
In this section, we show that if the Gram matrix for the collection of quantum states
$\{\ket{\psi_{\bc}}\}_{\bc\in \localcode}$ has a specific two-level eigenvalue spectrum
$\{\kappa_{\chi}\}_{\chi\in \wC}$ such that
$\kappa_{\chitriv}=\kappamax$ and
$\kappa_{\chi}=\kappamin$ for all $\chi\neq \chitriv$,
and if we fix
$\reward(\subspace)=\dim(\localcode)-\dim(\subspace)$,
then the LP objective, namely the optimal expected number of equations, becomes directly computable.
Moreover, this eigenvalue spectrum is directly related to the choice of quantum states in the local quantum decoding based optimization framework of~\cite{shutty2026lqd}.
In Lemma~\ref{lem:kappa_config_for_global_usd}, we rigorously prove the closed-form expression of the affine filtering objective for this specific eigenvalue spectrum of the Gram matrix $G$.
\begin{lemma}\label{lem:kappa_config_for_global_usd}
    Consider the collection of quantum states $\{\ket{\psi_{\bc}}\}_{\bc\in \localcode}$, such that $\kappa_{\chitriv}=\kappamax$ and $\kappa_{\chi}=\kappamin$ $\forall\chi\neq \chitriv$, with $\kappamax> \kappamin$ and $\reward(\subspace)=\dim(\localcode)-\dim(\subspace)$. Then the following holds
    \begin{align*}
        \objlp(\lpfeasible)=\kappamin\dim(\localcode).
    \end{align*}
\end{lemma}
\begin{proof}
    Since $\kappa_{\chi}=\kappamin$ $\forall\chi\neq \chitriv$, summing constraints over all $\chi\neq\chitriv$ on the LP in \eqref{eq:theta-LP form}, this yields 
    \begin{align*}
        \sum_{\chi\neq\chitriv}\sum_{\subspace\in \codesubspaces^{\localcode}}\sum_{\subspaceperp_m\in\dualcosetspace:\chi\in \subspaceperp_m}\theta_{\subspace}(\subspaceperp_m)\leq \kappamin (|\localcode|-1).
    \end{align*}
    For subspace $\subspace\in \codesubspaces^{\localcode}$ and $\subspaceperp_{m}\in \dualcosetspace$, let $\xi_{\subspace}(\subspaceperp_{m})$ be the number of nontrivial characters (i.e. $\chi\neq\chitriv$) in the coset $\subspaceperp_m$. Then we obtain 
    \begin{align*}
        \sum_{\subspace\in \codesubspaces^{\localcode}}\sum_{\subspaceperp_m\in\dualcosetspace}\theta_{\subspace}(\subspaceperp_m)\xi_{\subspace}(\subspaceperp_{m})\leq \kappamin (|\localcode|-1).
    \end{align*}
    Define the quantity $\Xi_{\subspace}(\subspaceperp_{m})$ as
    \begin{align*}
        \Xi_{\subspace}(\subspaceperp_{m})=\frac{q^{\dim(\localcode)-\dim(\subspace)}\reward(\subspace)}{|\localcode|\xi_{\subspace}(\subspaceperp_m)}.
    \end{align*}
    Consider the subspace $\subspace=\{\mzero\}$, in this case $\subspaceperp=\wC$ and $\dualcosetspace=\{\subspaceperp\}$. It also satisfies 
    \begin{align*}
        \Xi_{\subspace=\{\mzero\}}(\subspaceperp) & =\frac{q^{\dim(\localcode)}\dim(\localcode)}{|\localcode|(|\localcode|-1)}\\
        & = \frac{\dim(\localcode)}{|\localcode|-1}.
    \end{align*}
    For any other subspace  $\subspace\neq \{\mzero\}$ and $\subspaceperp_{m}\in \dualcosetspace$, if $\chitriv\in \subspaceperp_{m}$, it follows that
    \begin{align*}
        \Xi_{\subspace}(\subspaceperp_m)
         &  = \frac{q^{\dim(\localcode)-\dim(\subspace)}(\dim(\localcode)-\dim(\subspace))}{|\localcode|(|\subspaceperp_{m}|-1)}\\
         & = \frac{q^{\dim(\localcode)-\dim(\subspace)}(\dim(\localcode)-\dim(\subspace))}{|\localcode|(q^{\dim(\localcode)-\dim(\subspace)}-1)}.
    \end{align*}
For subspace  $\subspace\neq \{\mzero\}$ and $\subspaceperp_{m}\in \dualcosetspace$, if $\chitriv\notin\subspaceperp_m$, then 
    \begin{align*}
        \Xi_{\subspace}(\subspaceperp_m)
         &  = \frac{q^{\dim(\localcode)-\dim(\subspace)}(\dim(\localcode)-\dim(\subspace))}{|\localcode||\subspaceperp_{m}|}\\
         & = \frac{\dim(\localcode)-\dim(\subspace)}{|\localcode|}.
    \end{align*}
    Therefore, $\Xi_{\subspace=\{\mzero\}}(\subspaceperp)\geq \Xi_{\subspace\neq\{\mzero\}}(\subspaceperp_{m})$. Writing the LP objective in terms of $\Xi_{\subspace}(\subspaceperp_{m})$ we obtain
    \begin{align*}
        & \flp(\fulltheta)\\
        & =\frac{1}{|\localcode|}\sum_{\subspace\in \codesubspaces^{\localcode}}\sum_{\subspaceperp_m\in \dualcosetspace}\theta_{\subspace}(\subspaceperp_m)q^{\dim(\localcode)-\dim(\subspace)}\reward(\subspace)\\
        & = \sum_{\subspace\in \codesubspaces^{\localcode}}\sum_{\subspaceperp_m\in \dualcosetspace}\theta_{\subspace}(\subspaceperp_m)\Xi_{\subspace}(\subspaceperp_m)\xi_{\subspace}(\subspaceperp_m)\\
        & = \theta_{\subspace=\{\mzero\}}(\subspaceperp)\Xi_{\subspace=\{\mzero\}}(\subspaceperp)\xi_{\subspace=\{\mzero\}}(\subspaceperp)+\sum_{\subspace\in \codesubspaces^{\localcode}:\subspace\neq\{\mzero\}}\sum_{\subspaceperp_m\in \dualcosetspace}\theta_{\subspace}(\subspaceperp_m)\Xi_{\subspace}(\subspaceperp_m)\xi_{\subspace}(\subspaceperp_m)\\
        & \leq \Xi_{\subspace=\{\mzero\}}(\subspaceperp)\left(\sum_{\subspace\in \codesubspaces^{\localcode}}\sum_{\subspaceperp_m\in \dualcosetspace}\theta_{\subspace}(\subspaceperp_m)\xi_{\subspace}(\subspaceperp_m)\right)\\
        & \leq \frac{\dim(\localcode)}{|\localcode|-1}\kappamin(|\localcode|-1)=\kappamin\dim(\localcode).
    \end{align*}
    This is achieved by setting $\theta_{\subspace}(\subspaceperp_{m})=0$ $\forall \subspace\in \codesubspaces^{\localcode}$ such that $\subspace\neq \{\mzero\}$ and $\subspaceperp_{m}\in \dualcosetspace$ and $\theta_{\subspace=\{\mzero\}}(\subspaceperp)=\kappamin$.
\end{proof}


\begin{corollary}
For the collection of quantum states $\{\ket{\psi_{\bc}}\}_{\bc\in \localcode}$, such that
$\kappa_{\chitriv}=\kappamax$ and $\kappa_{\chi}=\kappamin$ for all $\chi\neq\chitriv$,
with $\kappamax>\kappamin$ and
$\reward(\subspace)=\dim(\localcode)-\dim(\subspace)$, an optimal 
affine filtering measurement is a 
codeword filtering measurement from Definition~\ref{def:global USD}. Equivalently, one optimal measurement is
\begin{align*}
    \left\{\Pi_{\bc}^{\mathrm{Cw\mbox{-}Filt}}\right\}_{\bc\in\localcode}\cup
    \left\{\Pi_{e}^{\mathrm{Cw\mbox{-}Filt}}\right\},
\end{align*}
which outputs the transmitted codeword $\bc$ or an erasure.
\end{corollary}

\begin{proof}
By Lemma~\ref{lem:kappa_config_for_global_usd}, the LP optimum is achieved by setting
$\theta_{\subspace}(\subspaceperp_m)=0$ for all $\subspace\neq\{\mzero\}$ and concentrating all mass on
$\subspace=\{\mzero\}$. In this case every conclusive affine outcome is a singleton affine subspace
$\{\bc\}$, so the resulting affine filtering measurement coincides with the 
codeword filtering measurement from
Definition~\ref{def:global USD}. Hence the
codeword filtering is optimal among affine filtering measurements for this two level Gram matrix eigenvalue spectrum.
\end{proof}
Since the FGUM in~\cite{shutty2026lqd} has conclusive outcomes that identify the transmitted local SPC codeword exactly, with the remaining outcome treated as an erasure, it is a 
codeword filtering measurement for the SPC-associated state family considered above.  Therefore, for the same choice of quantum states, the above corollary shows that FGUM is an optimal 
affine filtering measurement for maximizing the expected number of recovered linear equations.
\subsection{Limiting LP for Full-Rank Perturbations of Rank-Deficient Gram Matrices}\label{sec: rank deficient}
So far, we have shown that if the Gram matrix is full rank, or equivalently if the quantum states $\{\ket{\psi_{\bc}}\}_{\bc\in \localcode}$ are linearly independent, then we can reduce the SDP for constructing optimal 
affine filtering measurement to an LP.  
In this section, we provide a limiting LP formulation associated with rank-deficient Gram matrices.
We perturb the state family so that the corresponding Gram matrix becomes full rank,
apply the LP characterization from the previous section, and then pass to the limit. 
This gives a limiting LP characterization of the values of the full-rank perturbed Gram matrices.
For $\epsilon>0$, define a perturbed seed state
\[
\ket{\psi^\epsilon}
=
\sum_{\bx\in\mathbb{F}_q^{|I|}}
\alpha_{\bx}^{\epsilon}\ket{\bx},
\qquad
\alpha_{\bx}^{\epsilon}
:=
\begin{cases}
\dfrac{\alpha_{\bx}}{|\alpha_{\bx}|}
\sqrt{\dfrac{|\alpha_{\bx}|^2+\epsilon/q^{|I|}}{1+\epsilon}},
& \alpha_{\bx}\neq 0,\\[2ex]
\sqrt{\dfrac{\epsilon/q^{|I|}}{1+\epsilon}},
& \alpha_{\bx}=0 .
\end{cases}
\]
and let
\[
    \ket{\psi_{\bc}^{\epsilon}}
    =
    X^{\bc}\QFT^{\otimes |I|}\ket{\psi^{\epsilon}},
    \qquad \bc\in\localcode .
\]
\begin{lemma}\label{lem:non-fullrank Gram}
Consider the collection of quantum states
$\{\ket{\psi_{\bc}}\}_{\bc\in\localcode}$ with rank-deficient Gram matrix
$G$ characterized by eigenvalues $\{\kappa_{\chi}\}_{\chi\in \wC}$.
For $\epsilon>0$, let $\{\ket{\psi_{\bc}^{\epsilon}}\}_{\bc\in\localcode}$
be the perturbed state family defined above. Its Gram matrix is
\[
    G_{\epsilon}=\frac{G+\epsilon\mI}{1+\epsilon},
\]
with eigenvalues
\[
    \kappa_{\chi}^{\epsilon}
    =
    \frac{\kappa_{\chi}+\epsilon}{1+\epsilon},
    \qquad \chi\in\wC .
\]
For every $\epsilon>0$, the perturbed state family has full-rank Gram
matrix. Hence, by Theorem~\ref{thm:affine usd sdp-lp equivalence}, the
full-rank 
affine filtering measurement optimization reduces to
\[
    \text{LP}\Big(
    \{\theta_{\subspace}(\subspaceperp_m)\}_{
    \subspace\in \codesubspaces^{\localcode},
    \subspaceperp_m\in \dualcosetspace},
    \{\kappa_{\chi}^{\epsilon}\}_{\chi\in \wC}
    \Big),
\]
and
\[
    \objlp(\cL(\localcode,\{\ket{\psi_{\bc}^{\epsilon}}\}_{\bc\in\localcode}))
    =
    \objlp\Big(
    \{\theta_{\subspace}(\subspaceperp_m)\}_{
    \subspace\in\codesubspaces^{\localcode},
    \subspaceperp_m\in\dualcosetspace},
    \{\kappa_{\chi}^{\epsilon}\}_{\chi\in\wC}
    \Big).
\]
As $\epsilon\to 0$, these LPs converge to
\[
    \text{LP}\Big(
    \{\theta_{\subspace}(\subspaceperp_m)\}_{
    \subspace\in \codesubspaces^{\localcode},
    \subspaceperp_m\in \dualcosetspace},
    \{\kappa_{\chi}\}_{\chi\in \wC}
    \Big),
\]
and their optimal values satisfy
\[
    \lim_{\epsilon\to 0}
    \objlp\Big(
    \{\theta_{\subspace}(\subspaceperp_m)\}_{
    \subspace\in\codesubspaces^{\localcode},
    \subspaceperp_m\in\dualcosetspace},
    \{\kappa_{\chi}^{\epsilon}\}_{\chi\in\wC}
    \Big)
    =
    \objlp\Big(
    \{\theta_{\subspace}(\subspaceperp_m)\}_{
    \subspace\in\codesubspaces^{\localcode},
    \subspaceperp_m\in\dualcosetspace},
    \{\kappa_{\chi}\}_{\chi\in\wC}
    \Big).
\]
Consequently,
\begin{align*}
    \lim_{\epsilon\to 0} \objsdp(\cM(\localcode,\{\ket{\psi_{\bc}^{\epsilon}}\}_{\bc\in \localcode}))  &=  \lim_{\epsilon\to 0}
    \objlp(\cL(\localcode,\{\ket{\psi_{\bc}^{\epsilon}}\}_{\bc\in\localcode}))\\
   & =
    \objlp\Big(
    \{\theta_{\subspace}(\subspaceperp_m)\}_{
    \subspace\in\codesubspaces^{\localcode},
    \subspaceperp_m\in\dualcosetspace},
    \{\kappa_{\chi}\}_{\chi\in\wC}
    \Big).
\end{align*}
In particular, if $\subspaceperp_m$ contains a character $\chi$ such that
$\kappa_{\chi}=0$, then every feasible solution of the limiting LP satisfies
\[
    \theta_{\subspace}(\subspaceperp_m)=0
\]
for all $\subspace\in\codesubspaces^{\localcode}$.
\end{lemma}
\begin{proof}
Consider the LP in \eqref{eq:theta-LP form} $ \text{LP}\Big(
    \{\theta_{\subspace}(\subspaceperp_m)\}_{
    \subspace\in \codesubspaces^{\localcode},
    \subspaceperp_m\in \dualcosetspace},
    \{\kappa_{\chi}\}_{\chi\in \wC}
    \Big)$ where the right hand side values $\{\kappa_{\chi}\}_{\chi\in \wC}$ corresponds to the eigenvalues of rank-deficient Gram matrix $G$. If
$\kappa_{\chi}=0$ for some $\chi\in\wC$, then the corresponding constraint
in \eqref{eq:theta-LP form} gives
\begin{align*}
    \sum_{\subspace\in \codesubspaces^{\localcode}}
    \sum_{\subspaceperp_m\in\dualcosetspace:\chi\in \subspaceperp_m}
    \theta_{\subspace}(\subspaceperp_m)\leq 0.
\end{align*}
Since $\theta_{\subspace}(\subspaceperp_m)\geq 0$ for all
$\subspace\in \codesubspaces^{\localcode}$ and
$\subspaceperp_m\in\dualcosetspace$, every term in the above sum must
vanish. Therefore,
\begin{align*}
    \theta_{\subspace}(\subspaceperp_m)=0,
\end{align*}
whenever $\subspaceperp_m$ contains a character $\chi$ with
$\kappa_{\chi}=0$. Thus every feasible solution of \eqref{eq:theta-LP form} is
supported only on dual cosets $\subspaceperp_m$ for which
$\kappa_{\chi}>0$ for all $\chi\in\subspaceperp_m$.
 For perturbed states $\{\ket{\psi_{\bc}^{\epsilon}}\}_{\bc\in\localcode}$, we obtain
\begin{align}
    \braket{\psi_{\bc}^{\epsilon}|\psi_{\bc'}^{\epsilon}}
    &=
    \sum_{\bx\in\mathbb F_q^{|I|}}
    \frac{|\alpha_{\bx}|^2+\epsilon/q^{|I|}}{1+\epsilon}
    \omega^{\langle \bx,\bc-\bc'\rangle} \\
    &=
    \frac{1}{1+\epsilon}
    \sum_{\bx\in\mathbb F_q^{|I|}}
    |\alpha_{\bx}|^2
    \omega^{\langle \bx,\bc-\bc'\rangle}
    +
    \frac{\epsilon}{1+\epsilon}
    \frac{1}{q^{|I|}}
    \sum_{\bx\in\mathbb F_q^{|I|}}
    \omega^{\langle \bx,\bc-\bc'\rangle}.
\end{align}
The last sum is equal to $q^{|I|}$ if $\bc=\bc'$ and is equal to $0$ otherwise.
Hence
\[
    \braket{\psi_{\bc}^{\epsilon}|\psi_{\bc'}^{\epsilon}}
    =
    \frac{G_{\bc,\bc'}+\epsilon\delta_{\bc,\bc'}}{1+\epsilon}.
\]
Thus the Gram matrix of
$\{\ket{\psi_{\bc}^{\epsilon}}\}_{\bc\in\localcode}$ is
\[
    G_{\epsilon}
    =
    \frac{G+\epsilon\mI}{1+\epsilon}.
\]
Since $G$ is diagonalized by the character basis, $G_{\epsilon}$ is diagonalized
by the same basis.  Therefore, if $\{\kappa_{\chi}\}_{\chi\in\wC}$ are the
eigenvalues of $G$, then the eigenvalues of $G_{\epsilon}$ are
\[
    \kappa_{\chi}^{\epsilon}
    =
    \frac{\kappa_{\chi}+\epsilon}{1+\epsilon},
    \qquad \chi\in\wC .
\]
In particular, $\kappa_{\chi}^{\epsilon}>0$ for every $\chi\in\wC$, so
$G_{\epsilon}$ is full rank.
    For subspace $\subspace$ and $\subspaceperp_{m}\in \dualcosetspace$, we consider 
    \begin{align*}
        \eta_{\subspace}^{\epsilon}(\subspaceperp_{m})=\sum_{\chi\in\subspaceperp_m}\frac{1}{\kappa_{\chi}^{\epsilon}}.
    \end{align*}
    Let $\{\theta_{\subspace}^{\epsilon,*}(\subspaceperp_m)\}_{\subspace\in \codesubspaces^{\localcode},\subspaceperp_m\in\dualcosetspace }$ denote the optimal solution for $\text{LP}\Big(
    \{\theta_{\subspace}(\subspaceperp_m)\}_{
    \subspace\in \codesubspaces^{\localcode},
    \subspaceperp_m\in \dualcosetspace},
    \{\kappa_{\chi}^{\epsilon}\}_{\chi\in \wC}
    \Big)$.
Since $G_{\epsilon}$ is full rank, using Theorem~\ref{thm:affine usd sdp-lp equivalence} we reduce the SDP in \eqref{SDP main} for perturbed quantum sates $\{\ket{\psi_{\bc}^{\epsilon}}\}_{\bc\in \localcode}$ to an LP. Therefore, from \eqref{eq:theta-LP form}, we get 
\begin{align*}
    \objsdp(\cM(\localcode,\{\ket{\psi_{\bc}^{\epsilon}}\}_{\bc\in \localcode})) & =  \objlp(\cL(\localcode,\{\ket{\psi_{\bc}^{\epsilon}}\}_{\bc\in \localcode}))\\
    & = \objlp\Big(
    \{\theta_{\subspace}(\subspaceperp_m)\}_{
    \subspace\in\codesubspaces^{\localcode},
    \subspaceperp_m\in\dualcosetspace},
    \{\kappa_{\chi}^{\epsilon}\}_{\chi\in\wC}
    \Big)\\
    & = \sum_{\subspace\in \codesubspaces^{\localcode}}\sum_{\subspaceperp_m\in \dualcosetspace}\theta_{\subspace}^{\epsilon,*}(\subspaceperp_m)\frac{q^{\dim(\localcode)-\dim(\subspace)}}{|\localcode|}\reward(\subspace).
\end{align*}
Next, define the feasible set $F(\epsilon)$ for the LP parameters by
\begin{align}
 &\   F(\epsilon) \\
 &\ \coloneqq  \left\{ \{\theta_{\subspace}^{\epsilon}(\subspaceperp_m)\}_{\subspace\in \codesubspaces^{\localcode},\subspaceperp_m\in\dualcosetspace } \;\middle|\; 
    \begin{aligned}
        & \sum_{\subspace\in \codesubspaces^{\localcode}}\sum_{\subspaceperp_m\in\dualcosetspace:\chi\in \subspaceperp_m}\theta_{\subspace}^{\epsilon}(\subspaceperp_m)\leq \kappa_{\chi}^{\epsilon}, \qquad \forall \chi\in \wC \\
        & \theta_{\subspace}^{\epsilon}(\subspaceperp_m)\geq 0, \qquad \forall \subspace\in \codesubspaces^{\localcode}, \subspaceperp_m\in \dualcosetspace
    \end{aligned}
    \right\}\nonumber.
\end{align}
Since $\theta_{\subspace}^{\epsilon}(\subspaceperp_m)\geq 0$ and $\kappa^{\epsilon}_{\chi}\geq 0$ for $\epsilon\geq 0$, $F(\epsilon)$ is non empty $\forall \epsilon\geq 0$.
Using the fact $\sum_{\chi}\kappa_{\chi}=|\localcode|$, we obtain
\begin{align*}
    \sum_{\chi}\kappa_{\chi}^{\epsilon}=\frac{\sum_{\chi}(\kappa_{\chi}+\epsilon)}{1+\epsilon}=|\localcode|,
\end{align*}
where we used the fact that $|\wC|=|\localcode|$.
This implies $\theta^{\epsilon}_{\subspace}(\subspaceperp_m)$ is upper bounded by $|\localcode|$. Thus, $F(\epsilon)$ is closed and bounded, hence compact. 
Next, consider sequences $\{\epsilon_n\}$, $\{\kappa^{\epsilon_n}_{\chi}\}$ $\forall\chi\in \wC$ and $\{\theta_{\subspace}^{\epsilon_n, *}(\subspaceperp)\}$ for all $ \subspace,\subspaceperp_m$, where the optimal LP parameters $\{\theta_{\subspace}^{\epsilon_n, *}(\subspaceperp_m)\}_{\subspace\in \codesubspaces^{\localcode},\subspaceperp_m\in\dualcosetspace }\in F(\epsilon_n)$ .  Since $0\leq\theta_{\subspace}^{\epsilon,*}(\subspaceperp_m)\leq |\localcode|$ for all $\subspace,\subspaceperp_m$, from uniform boundedness the sequence $\{\theta_{\subspace}^{\epsilon_n, *}(\subspaceperp)\}$ admits  a convergent subsequence $\{\theta_{\subspace}^{\epsilon_{n_k}, *}(\subspaceperp)\}$ for all $ \subspace,\subspaceperp_m$,  such that $\theta_{\subspace}^{\epsilon_{n_k},*}(\subspaceperp_m)\to \overline{\theta}_{\subspace}(\subspaceperp_m)$. By the continuity of linear constraints and using $\kappa_{\chi}^{\epsilon_{n_k}}\to \kappa_{\chi}$, this yields $\{\overline{\theta}_{\subspace}(\subspaceperp_m)\}_{\subspace\in \codesubspaces^{\localcode},\subspaceperp_m\in\dualcosetspace }\in F(0)$. Therefore 
\begin{align*}
    \lim_{k\to \infty}\objlp(\cL(\localcode,\{\ket{\psi_{\bc}^{\epsilon_{n_k}}}\}_{\bc\in \localcode})) & = \lim_{k\to \infty}\sum_{\subspace\in \codesubspaces^{\localcode}}\sum_{\subspaceperp_m\in \dualcosetspace}\theta_{\subspace}^{\epsilon_{n_k},*}(\subspaceperp_m)\frac{q^{\dim(\localcode)-\dim(\subspace)}}{|\localcode|}\reward(\subspace)\\
    & = \sum_{\subspace\in \codesubspaces^{\localcode}}\sum_{\subspaceperp_m\in \dualcosetspace}\overline{\theta}_{\subspace}(\subspaceperp_m)\frac{q^{\dim(\localcode)-\dim(\subspace)}}{|\localcode|}\reward(\subspace)\\
    & \leq \sum_{\subspace\in \codesubspaces^{\localcode}}\sum_{\subspaceperp_m\in \dualcosetspace}\theta_{\subspace}^{*}(\subspaceperp_m)\frac{q^{\dim(\localcode)-\dim(\subspace)}}{|\localcode|}\reward(\subspace)\\
    & = \mathrm{OPT}_{\mathrm{LP}}(\{\theta_{\subspace}(\subspaceperp_m)\}_{ \subspace\in \codesubspaces^{\localcode}, \subspaceperp_m\in \dualcosetspace},\{\kappa_{\chi}\}_{\chi\in \wC}),
\end{align*}
where $\{\theta^{*}_{\subspace}(\subspaceperp_m)\}_{\subspace\in \codesubspaces, \subspaceperp_m\in \dualcosetspace}\in F(0)$ is the optimal solution of $\text{LP}(\{\theta_{\subspace}(\subspaceperp_m)\}_{ \subspace\in \codesubspaces^{\localcode}, \subspaceperp_m\in \dualcosetspace},\{\kappa_{\chi}\}_{\chi\in \wC})$.
Next, consider the set $S_{+}=\{\chi\colon \kappa_{\chi}> 0\}$. Since $\sum_{\chi}\kappa_{\chi}=|\localcode|$, the set $S_{+}$ is nonempty. Define the quantity $\lambda(\epsilon)$ as follows 
 \begin{align*}
     \lambda(\epsilon)\coloneqq\min\{1,\min_{\chi\in S_+}\frac{\kappa_{\chi}^{\epsilon}}{\kappa_{\chi}}\}.
 \end{align*}
 When $\kappa_{\chi}> 0$, this gives
 \begin{align*}
  &   \sum_{\subspace\in \codesubspaces^{\localcode}}\sum_{\subspaceperp_m\in\dualcosetspace:\chi\in \subspaceperp_m}\lambda(\epsilon)\theta_{\subspace}^{*}(\subspaceperp_m)\\
  & =\lambda(\epsilon) \sum_{\subspace\in \codesubspaces^{\localcode}}\sum_{\subspaceperp_m\in\dualcosetspace:\chi\in \subspaceperp_m}\theta^{*}_{\subspace}(\subspaceperp_m)\\
  & \leq \lambda(\epsilon)\kappa_{\chi}\\
  & \leq \kappa_{\chi}^{\epsilon}.
 \end{align*}
 When $\kappa_{\chi}=0$, the sum satisfies
 \begin{align*}
     &   \sum_{\subspace\in \codesubspaces^{\localcode}}\sum_{\subspaceperp_m\in\dualcosetspace:\chi\in \subspaceperp_m}\lambda(\epsilon)\theta_{\subspace}^{*}(\subspaceperp_m)\\
     & = 0\\
     & \leq \kappa_{\chi}^{\epsilon}.
 \end{align*}
 This implies $\{\lambda(\epsilon)\theta^{*}_{\subspace}(\subspaceperp_m)\}_{\subspace\in \codesubspaces, \subspaceperp_m\in \dualcosetspace}\in F(\epsilon)$. Therefore
 \begin{align*}
    \objlp(\cL(\localcode,\{\ket{\psi_{\bc}^{\epsilon}}\}_{\bc\in \localcode})) & =   \sum_{\subspace\in \codesubspaces^{\localcode}}\sum_{\subspaceperp_m\in \dualcosetspace}\theta_{\subspace}^{\epsilon,*}(\subspaceperp_m)\frac{q^{\dim(\localcode)-\dim(\subspace)}}{|\localcode|}\reward(\subspace)\\
    & \geq \sum_{\subspace\in \codesubspaces^{\localcode}}\sum_{\subspaceperp_m\in\dualcosetspace}\lambda(\epsilon)\theta_{\subspace}^{*}(\subspaceperp_m)\frac{q^{\dim(\localcode)-\dim(\subspace)}}{|\localcode|}\reward(\subspace)\\
    & =\lambda(\epsilon)\mathrm{OPT}_{\mathrm{LP}}(\{\theta_{\subspace}(\subspaceperp_m)\}_{ \subspace\in \codesubspaces^{\localcode}, \subspaceperp_m\in \dualcosetspace},\{\kappa_{\chi}\}_{\chi\in \wC}).
 \end{align*}
 Combining the above facts and since $\lim_{\epsilon\to 0}\lambda(\epsilon)=1$, we obtain
 \begin{align}
  \lim_{\epsilon\to 0} \objsdp(\cM(\localcode,\{\ket{\psi_{\bc}^{\epsilon}}\}_{\bc\in \localcode}))  &=
     \lim_{\epsilon\to 0}\objlp(\cL(\localcode,\{\ket{\psi_{\bc}^{\epsilon}}\}_{\bc\in \localcode}))\\
     & = \lim_{\epsilon \to 0}\objlp\Big(
    \{\theta_{\subspace}(\subspaceperp_m)\}_{
    \subspace\in\codesubspaces^{\localcode},
    \subspaceperp_m\in\dualcosetspace},
    \{\kappa_{\chi}^{\epsilon}\}_{\chi\in\wC}
    \Big)\\
     & = \objlp(\{\theta_{\subspace}(\subspaceperp_m)\}_{ \subspace\in \codesubspaces^{\localcode}, \subspaceperp_m\in \dualcosetspace},\{\kappa_{\chi}\}_{\chi\in \wC}).
 \end{align}
 \end{proof}

 \begin{remark}
The limiting LP proof is also consistent with standard stability results for parametric linear programming. Under right-hand-side perturbations, the value function of a finite-dimensional LP with fixed objective vector and fixed
constraint matrix is piecewise linear, and hence locally Lipschitz, on regions where the problem remains feasible and bounded \cite{gisbert2019lipschitz}.
\end{remark}
\section{Affine Filtering Measurement Based Decoding}\label{sec:affine spc decoder}
The previous section characterizes 
affine filtering measurements for a fixed local code and a corresponding collection of symmetric codeword states.  We now specialize
this construction to local codes arising from SPC constraints in a regular LDPC
code and describe the resulting decoding procedure.  The role of the local
measurement is to convert each conclusive outcome into linear constraints on a
small subset of code symbols.  These locally generated constraints are then used
together with the global parity-check equations to determine the transmitted
codeword.
\subsection{Decoding on I.I.D. States for Single Parity Check Codes}

In Section~\ref{sec:affineUSD}, we described 
affine filtering measurements for an arbitrary local code $\cC_{I}$ supported on an index set $I$. In this section, we specialize that framework to single-parity-check (SPC) local codes and the corresponding 
affine filtering decoder for LDPC codes on i.i.d.\ pure-state channels. Let $\{\beta_{\subspace}^{*}(\subspaceperp_m)\}_{\subspace\in\codesubspaces^{\cC_I},\,\subspaceperp_m\in\dualcosetspace}$ denote an optimal solution of the LP in~\eqref{LP: final form}. Then the corresponding POVM
\begin{align}
\{\Pi_{\subspace,\bs}^{*}\}_{\subspace\in \codesubspaces^{\cC_{I}},\,\bs\in \Image(H_{\subspace})}\cup\{\Pi_{e}^{*}\}
\end{align}
with
\begin{align}
\Pi_{\subspace,\bs}^{*}
&=
\sum_{\subspaceperp_{m}\in \dualcosetspace}
\beta_{\subspace}^{*}(\subspaceperp_{m})
\ket{\zeta_{\subspace,\bs}^{\subspaceperp_{m}}}
\bra{\zeta_{\subspace,\bs}^{\subspaceperp_{m}}},
\\
\Pi_{e}^{*}
&=
\mI-\sum_{\subspace\in\codesubspaces^{\cC_I}}\Pi_{\subspace}^{*},
\end{align}
gives the optimal affine information for a degree-$D$ SPC local code. We now consider an i.i.d.~pure-state channel in which each symbol is transmitted through the same pure-state channel. Fix $0\leq \alpha\leq 1$ and a non-empty subset $S\subset [q]$, define
\begin{align}
\ket{\psi'}
&\coloneqq
\sqrt{\alpha}\sum_{i\in S}\frac{1}{\sqrt{|S|}}\ket{i}
+
\sqrt{1-\alpha}\sum_{j\in [q]\setminus S}\frac{1}{\sqrt{q-|S|}}\ket{j}.
\end{align}
For each input symbol $a\in\mbF_q$, the channel output is the pure state obtained from the same seed by the action of $X^a\QFT$. Since the channel is i.i.d., the output corresponding to a local codeword $\bc\in\cC_I$ is the tensor product of the single-symbol outputs, and hence is given by
\begin{align}
\ket{\psi_{\bc}}
&=
X^{\bc}\QFT^{\otimes |I|}\ket{\psi'}^{\otimes |I|}.
\end{align}
Therefore the family $\{\ket{\psi_{\bc}}\}_{\bc\in\cC_I}$ is exactly of the form considered in Section~\ref{sec:affineUSD}, and the LP in~\eqref{LP: final form} applies directly to this local SPC setting. Our goal in this section is to characterize the affine information revealed by the optimal SPC POVM for this family of states as a function of $\alpha$ and $S$. For a fixed $(k,D)$ pair and a fixed subset $S\subset[q]$, let 
$\alpha_{(k,D)}^{\mathrm{Affine\text{-}Filtering}}$
denote the parameter value for which the LP objective equals $D-k$, namely
\begin{align}
\sum_{\subspace\in \codesubspaces^{\cC_{I}}}
\sum_{\subspaceperp_m\in \dualcosetspace}
\beta_{\subspace}^{*}(\subspaceperp_m)
\frac{q^{2(\dim(\cC_{I})-\dim(\subspace))}}
{\eta_{\subspace}(\subspaceperp_m)|\cC_{I}|}
\reward(\subspace)
&=
D-k.
\end{align}
Now consider the quantum decoder which, upon receiving the state $\ket{\psi_{\bc}}$ for a codeword $\bc\in\cC$, applies the 
affine filtering POVM to disjoint SPC constraints and then runs Gaussian elimination to recover $\bc$. For a random $(k,D)$-regular LDPC code, dependencies among parity checks may reduce the decoding threshold below 
$\alpha_{(k,D)}^{\mathrm{Affine\text{-}Filtering}}$.
Moreover, it is not always possible to find disjoint SPC supports $\{I_j\}_{j\in [J]}$ such that $\bigsqcup_{j\in [J]} I_j=[N]$, where $N$ is the code length. Accordingly, the value 
$\alpha_{(k,D)}^{\mathrm{Affine\text{-}Filtering}}$ 
obtained from the LP gives an upper bound on the performance of the 
affine-filtering+GE decoder.


 Recent works on BPQM~\cite{brandsen2022belief, mandal2026belief} derive thresholds for $(k,D)$-regular LDPC codes using density-evolution analysis. A basic limitation of this approach is that it relies on the assumption that the neighborhood of a root node in the factor graph induced by a random $(k,D)$-regular LDPC code is tree-like up to sufficiently large depth. Moreover, it is not clear how to construct a practical BPQM decoder for LDPC codes with comparable complexity guarantees. Nevertheless, the threshold obtained from density evolution provides an upper bound on the performance of a BPQM-based quantum decoder.

\subsection{Greedy Affine Filtering Decoder for SPC  }
To construct an 
affine filtering measurement-based quantum decoder, we first find single parity check codes from the large code $\localcode$.
Consider a random parity check matrix $H$ of $(k,D)$-regular LDPC code $\cC$ and codeword $\bc\in\cC$ satisfying $H\bc=\mathbf{0}$. To avoid correlated outcomes when applying 
affine filtering measurement to SPCs, the SPCs must be pairwise disjoint. This can be done efficiently using a greedy search. Suppose the greedy search produces $J$ SPCs supported on index sets $\{I_{j}\}_{j\in [J]}$ with $J\leq \frac{N}{D}$. In this case $j$\textsuperscript{th} SPC corresponds to the code $\cC_{I_{j}}$. Let $I_{\text{single}}$ be the index set containing indices that do not belong to any of the index sets $\{I_{j}\}_{j\in [J]}$. In other words, $\forall k\in I_{\text{single}}$, we have $k\notin I_{j}$ for all $j \in [J]$. 
Thus the index sets satisfy $\sqcup_{j\in [J]}I_{j}\sqcup I_{\text{single}}=[N]$. Let $\{\Pi_{\subspace,\bs}^*\}_{\subspace\in \codesubspaces^{\cC_{I}},\bs\in \Image(H_{\subspace})}\cup\{\Pi_{e}^{*}\}$ be the POVM obtained from the LP solution \eqref{LP: final form} corresponding to degree $D$ SPC.  

\begin{algorithm}[!ht]
\caption{Affine Filtering Measurement on SPC based Quantum Decoding (Affine-Filtering+GE)\label{algo:affine usd decoding}}
\begin{algorithmic}[1]
      \Require Codeword state $\ket{\psi_{\bc}}$ with $\bc\in\cC$, POVM $\{\Pi_{\subspace,\bs}^*\}_{\subspace\in \codesubspaces^{\cC_{I}},\bs\in \Image(H_{\subspace})}\cup\{\Pi_{e}^{*}\}$, and single qudit measurement $\{\Pi_{u}^{\text{USD}}\}_{u\in \mathbb{F}_q}\cup\{\Pi_{e}^{\text{USD}}\}$
      \State Use greedy search to obtain disjoint SPCs supported on $\{I_{j}\}_{j\in [J]}$ from parity check matrix $H$ of code $\cC$\\
      Construct index set $I_{\text{single}}$ using indices $k$ such that $\forall k\in I_{\text{single}}$, $k\notin I_{j}$ $\forall j \in J$, and the disjoint union satisfies $\sqcup_{j\in J}I_{j}\sqcup I_{\text{single}}=[N]$\\
      Set $H_{\text{update}}=H$ and $\bb=\mathbf{0}$
      \For{each index set $I_{j}$ with $j\in [J]$}\\
\quad Apply POVM $\{\Pi_{\subspace,\bs}^*\}_{\subspace\in \codesubspaces^{\cC_{I}},\bs\in \Image(H_{\subspace})}\cup\{\Pi_{e}^{*}\}$ on state $\ket{\psi_{\bc_{I_j}}}$\\
\quad 
Obtain outcome $\subspace_{\bs}$ w.p. $p^*_{\subspace}$ for the unique syndrome $\bs\in\Image(H_{\subspace})$ such that $\bc_{I_j}\in\subspace_{\bs}$,\\
\quad or erasure w.p. $1-\sum_{\subspace\in \codesubspaces^{\cC_{I}}}p^*_{\subspace}$
\If{Outcome is an affine subspace $\subspace_{\bs}$}\\
\qquad Obtain updated linear system from $H_{\text{update}}=\begin{bmatrix}
    H_{\text{update}}\\
    (H_{\subspace})^{I_{j}}
\end{bmatrix}$ and $\bb=\begin{bmatrix}
    \bb\\
    \bs
\end{bmatrix}$
\EndIf\\

\EndFor\\

    \For{each index $k\in I_{\text{single}}$ }\\
\quad Apply POVM $\{\Pi_{u}^{\text{USD}}\}_{u\in \mathbb{F}_q}\cup\{\Pi_{e}^{\text{USD}}\}$ on single qudit state $\ket{\psi_{c_{k}}}$\\
\quad Obtain outcome $u\in \mathbb{F}_q$ w.p. $p_{u}^{\text{USD}}$ and erasure w.p. $1-\sum_{u}p_{u}^{\text{USD}}$
\If{Outcome is $u\in \mathbb{F}_q$}\\
\qquad Set $c_{k}=u$\\
\qquad Obtain updated linear system from $H_{\text{update}}=\begin{bmatrix}
    H_{\text{update}}\\
    e_{k}^{T}
\end{bmatrix}$ and $\bb=\begin{bmatrix}
    \bb\\
    c_{k}
\end{bmatrix}$
     \EndIf
    \EndFor\\
Use Gaussian elimination to obtain the codeword $\bc$ from linear system $H_{\text{update}}\bc=\bb$\\
\If{Gaussian elimination is successful}\\
\quad \Return Flag= Decoding is successful and codeword $\bc$
\Else\\
\Return Flag= Decoding Failure
\EndIf
    \end{algorithmic}
\end{algorithm}

We now explain how each step of Algorithm~\ref{algo:affine usd decoding} contributes linear information about the transmitted codeword and how this information is assembled into a global system of equations. The decoder begins with the parity-check equations of the LDPC code itself, and initializes
\begin{align}
H_{\mathrm{update}}=H,
\qquad
\bb=\mzero.
\end{align}
Throughout the decoding procedure, every additional conclusive measurement outcome contributes further linear constraints that are appended to this system.
Consider an arbitrary selected SPC support $I_j$, and let $\cC_{I_j}$ denote the corresponding SPC code. For this local code, Section~\ref{sec:affineUSD} provides the 
affine filtering POVM obtained from the LP in~\eqref{LP: final form} once the associated family of local codeword states $\{\ket{\psi_{\bc_{I_j}}}\}_{\bc_{I_j}\in\cC_{I_j}}$ is specified. In the present setting, these local states are induced by the i.i.d.\ PSC described above. The decoder may therefore apply the corresponding 
affine filtering POVM directly to the local state $\ket{\psi_{\bc_{I_j}}}$. For a fixed subspace $\subspace\in\codesubspaces^{\cC_{I_j}}$, define
\begin{align}
p_{\subspace}^{*}
\coloneqq
\sum_{\subspaceperp_m\in \dualcosetspace}
\beta_{\subspace}^{*}(\subspaceperp_m)
\frac{|\subspaceperp|^{2}}
{\eta_{\subspace}(\subspaceperp_m)|\localcode_{I}|}.
\end{align}
When the 
affine filtering measurement is applied to $\ket{\psi_{\bc_{I_j}}}$, the outcome is either an affine subspace $\subspace_{\bs}$, for the unique syndrome $\bs\in\Image(H_{\subspace})$ satisfying $\bc_{I_j}\in\subspace_{\bs}$ or an erasure. A conclusive outcome $\subspace_{\bs}$ determines a valid collection of linear equations on the coordinates indexed by $I_j$, namely
\begin{align}
H_{\subspace}\bc_{I_j}=\bs,
\end{align}
where $H_{\subspace}$ is any matrix satisfying $\ker(H_{\subspace})=\subspace$.
To incorporate this local information into the global decoding problem, we embed the matrix $H_{\subspace}$ into the full set of codeword coordinates. More precisely, let $(H_{\subspace})^{I_j}\in \mbF_q^{(\dim(C_{I_j})-\dim(\subspace))\times N}$ denote the matrix obtained by placing the columns of $H_{\subspace}$ in the coordinates indexed by $I_j$ and setting all remaining columns equal to $\mzero$. Then a conclusive 
affine filtering outcome yields the global update
\begin{align}
H_{\mathrm{update}}
=
\begin{bmatrix}
H_{\mathrm{update}}\\
(H_{\subspace})^{I_j}
\end{bmatrix},
\quad \bb
=
\begin{bmatrix}
\bb\\
\bs
\end{bmatrix}.
\end{align}
Repeating this procedure over all selected SPC supports $\{I_j\}_{j=0}^{J-1}$ accumulates the linear information obtained from the 
affine filtering measurements.
We next consider the remaining coordinates in $I_{\mathrm{single}}$, namely those indices that do not belong to any selected SPC support. For each $k\in I_{\mathrm{single}}$, we apply the optimal single-qudit USD measurement
\begin{align}
\{\Pi_{u}^{\mathrm{USD}}\}_{u\in\mbF_q}\cup\{\Pi_{e}^{\mathrm{USD}}\}
\end{align}
to the state $\ket{\psi_{c_k}}$. If the outcome is a symbol $u\in\mbF_q$, then the corresponding linear constraint is $c_k=u$,
which is appended to the global system through the update
\begin{align}
H_{\mathrm{update}}
=
\begin{bmatrix}
H_{\mathrm{update}}\\
e_k^T
\end{bmatrix},
\quad
\bb
=
\begin{bmatrix}
\bb\\
u
\end{bmatrix}.
\end{align}
If the single-qudit measurement returns an erasure, then no additional equation is added. After processing all SPC supports and all indices in $I_{\mathrm{single}}$, the decoder applies Gaussian elimination to the linear system
\begin{align}
H_{\mathrm{update}}\bc=\bb.
\end{align}
If this system determines the transmitted codeword uniquely, then the decoder outputs $\bc$; otherwise, it declares decoding failure. In this way, the decoder converts every conclusive 
affine filtering or single-qudit USD outcome into explicit linear constraints and then combines all such constraints through Gaussian elimination.
\section{Numerical Results }\label{sec:numerical results}
\subsection[Affine Filtering Decoder for Binary Code]%
{Affine Filtering Decoder Performance for LDPC Codes with $q=2$}
In Table \ref{tab:quantum decoding method comparison}, we provide a detailed comparison between several decoding algorithms  for $(k,D)$ LDPC codes from Gallager ensembles. 
Since for $q=2$, the problem is symmetric with respect to $\alpha=\frac{1}{2}$, we report the threshold for $\alpha\geq \frac{1}{2}$. Moreover, by symmetry, the choice of elements in the subset $S\subset [q]$ does not matter; we therefore use only the cardinality $|S|=r$ as a parameter for the channel $W$. For $q=2$, the only choice is $r=1$, so we will not specify $r$ to describe the channel. 

The BPQM threshold refers to the  threshold for $\alpha$ obtained using density evolution~\cite{PMBPQM_BSCQ} for the BPQM based quantum decoder. 
Similarly, the affine filtering threshold refers to the threshold
$\alpha_{(k,D)}^{\mathrm{Affine-Filtering}}$.
We observe that for $(k,D)$ pairs $(5,6),(6,7)$ and $(7,8)$, the
affine filtering upper bound is strictly larger than the corresponding
thresholds for other methods. Since there is no known classical method yet to
simulate the performance of BPQM based decoder for LDPC codes other than
density evolution, the threshold from affine filtering suggests it can become a
valid candidate for designing quantum decoders for LDPC codes, as it
outperforms qudit PGM+BP and qudit USD+GE performance for most pairs $(k,D)$.

  \begin{table}[h]
       \centering
     \begin{tabular}{|c|c|c|c|c|}
        \hline
     	(k,D)  & USD+GE & PGM+BP  & BPQM  & Affine Filtering+GE  \\
     	\hline
     	(3,4) & 0.875 &  0.87065 & \bf{0.9339} & 0.88268 \\
     	\hline
     	(3,5)& 0.8  &  0.81648 & \bf{0.8937}  & 0.84375 \\
     	\hline
     	(3,6)& 0.75 &  0.77562 & \bf{0.86159} & 0.79375  \\
     	\hline
     	(3,7) & 0.71428 & 0.74727 & \bf{0.83521} & 0.72926  \\
     	\hline
     	(3,8)& 0.6875  &  0.72351 & \bf{0.81541} & 0.70079 \\
     	\hline
     	(4,5)&  0.9 &  0.8597 & \bf{0.92607} & 0.90824  \\
     	\hline   
        (4,6) & 0.8333 &  0.82062 & \bf{0.89745} & 0.84376  \\
     	\hline
     	(4,7)&  0.78571&  0.78862 & \bf{0.87405} & 0.79946 \\
     	\hline
     	(4,8) & 0.75  &  0.76539 & \bf{0.85192} & 0.76225  \\
     	\hline
     	(5,6) & 0.91667 &  0.84305 & 0.91616 & \bf{0.91964}  \\
     	\hline
     	(5,7) & 0.85714& 0.81422 & \bf{0.8937} & 0.86445 \\
     	\hline
     	(5,8) & 0.8125  &  0.78752 & \bf{0.873} & 0.81973  \\
     	\hline
     	(6,7) & 0.92857 &  0.82759 & 0.90567 & \bf{0.92996}  \\
     	\hline
     	(6,8)&  0.875 &  0.80327 & \bf{0.88504} & 0.87725  \\
     	\hline
     	(7,8) &  0.9375&  0.813 & 0.8937 & \bf{0.9377}  \\
     	\hline
     \end{tabular}

   \caption{Comparison of Quantum Decoding Methods for $(k,D)$ regular LDPC codes}
       \label{tab:quantum decoding method comparison}
 \end{table}

While the theoretical bounds provide insights about the performance of decoding methods, to validate the performance we simulate 
affine filtering measurement along with GE based decoder for $(k,D)$ regular LDPC codes. Since, for LDPC codes, we do not have a method to simulate BPQM based decoder classically yet, we compare the performance of 
affine filtering+GE-based decoder with qubit USD+GE and qubit PGM+BP based decoders.

 In Figures~\ref{plot:3,4_F2_r1_plot} and~\ref{plot:4,5_F2_r1_plot} we compare the performance of 
 affine filtering+GE-decoder with qubit PGM+BP and qubit USD+GE decoders for the $(3,4)$ and $(4,5)$ regular LDPC code with standard Gallager construction respectively for code blocklength $N=1600$. We observe that decoder performance almost matches with the threshold and outperforms both qubit-wise decoders. In Figure~\ref{plot:6,7_F2_r1_plot}, we compare 
 affine filtering+GE-decoder performance with qubit wise USD+GE decoder for $(6,7)$ regular LDPC codes and different codeword lengths $N$. Since, for $(6,7)$ LDPC code, the BP performance is much worse compared to qubit USD+GE decoder we omit the curves for qubit PGM+BP decoder. More specifically, for $(6,7)$ LDPC code, the 
 affine filtering+GE decoder also achieves much higher threshold than the BPQM threshold suggesting it is the best performing decoder for this ensemble on an i.i.d PSC. 
 In Figure~\ref{6,7 threshold}, we provide a closer view of 
 affine filtering+GE decoder curves for different lengths, to show that the decoder performance achieves nearly the same threshold as the theoretical upper bound.  
 \begin{figure}[!htbp]
    \centering
    \includegraphics[width=0.6\linewidth]{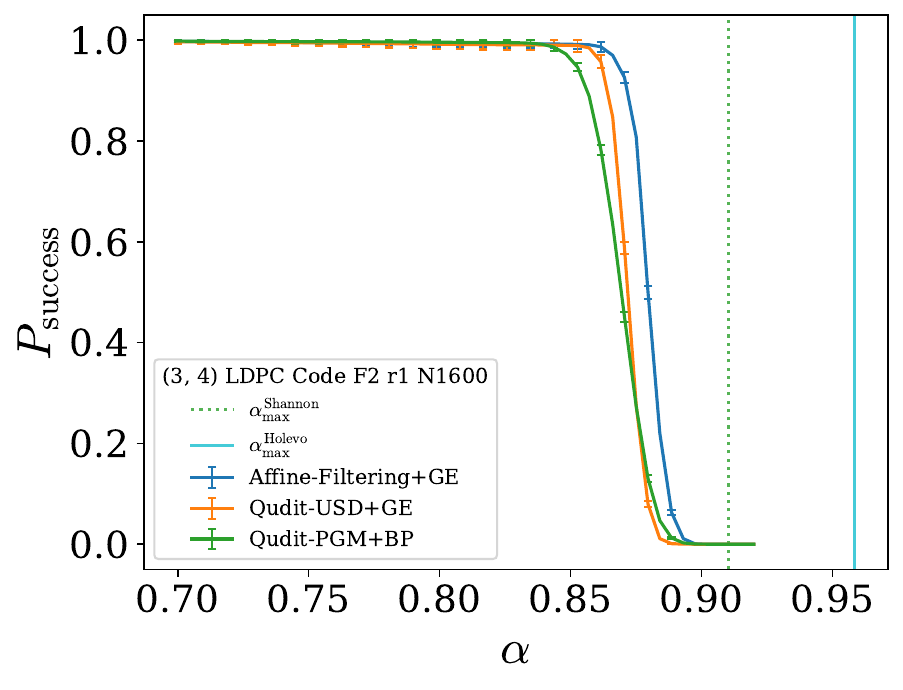}
    \caption{Performance Comparison for (3,4) Regular LDPC Codes on $\mathbb{F}_2$ and $r=1$ with $N=1600$}
    \label{plot:3,4_F2_r1_plot}
\end{figure}

\begin{figure}[!htbp]

    \centering
    \includegraphics[width=0.6\linewidth]{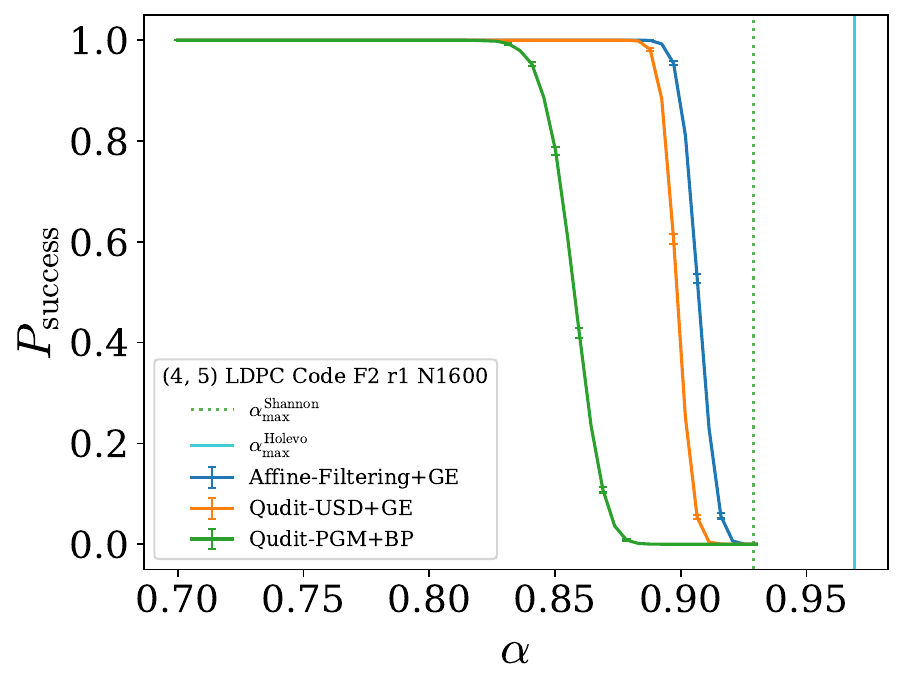}
    \caption{Performance Comparison for (4,5) Regular LDPC Codes on $\mathbb{F}_2$ and $r=1$ with $N=1600$}
    \label{plot:4,5_F2_r1_plot}
\end{figure}

\begin{figure}[!htbp]
    \centering
    \includegraphics[width=0.6\linewidth]{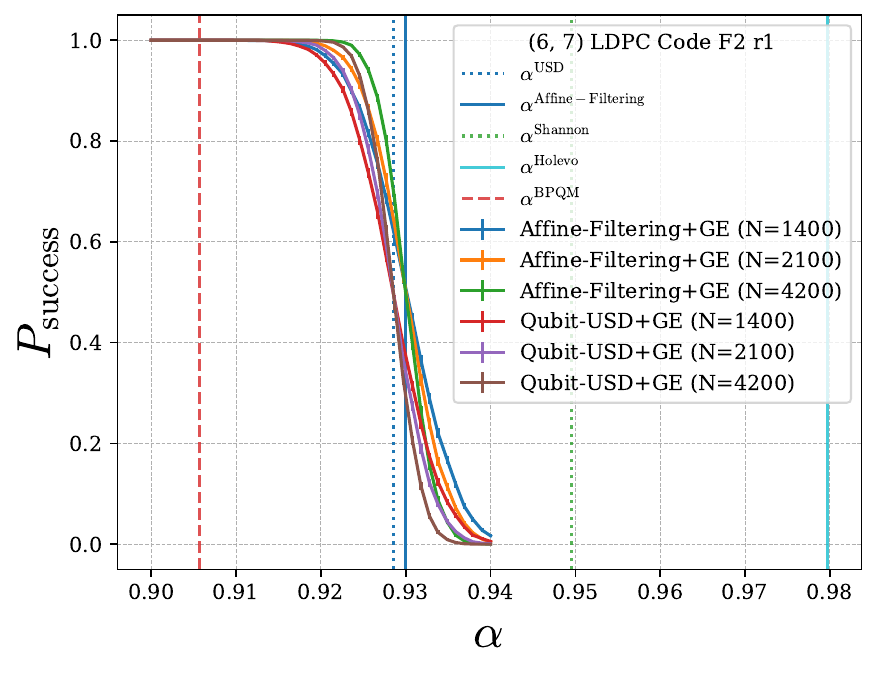}
    \caption{Comparison between affine-filtering+GE and qubit USD+GE decoding for (6,7) regular LDPC codes}
    \label{plot:6,7_F2_r1_plot}
\end{figure}

\begin{figure}[!htbp]
    \centering
    \includegraphics[width=0.6\linewidth]{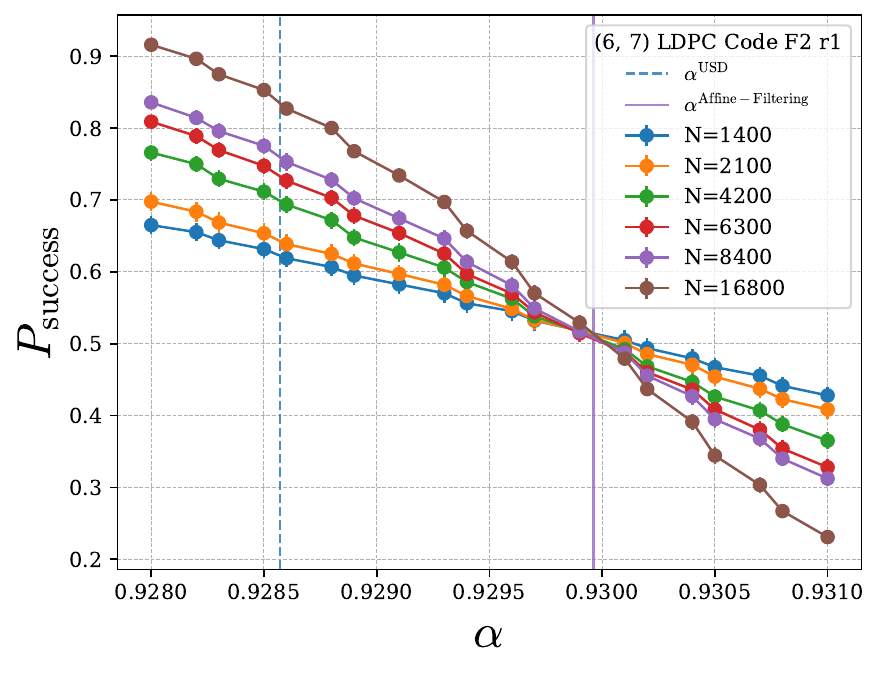}
    \caption{Threshold plot for affine-filtering+GE decoding for (6,7) regular LDPC codes}
    \label{6,7 threshold}
\end{figure}
\subsection[Affine Filtering Decoder]%
{Affine Filtering+GE Decoder Performance for LDPC Codes with $q>2$}
Similarly to the binary case, we simulate the 
affine filtering+GE decoder for $q>2$. In this case, we fix $r$ and compare the performance of 
affine filtering+GE decoder with qudit PGM+BP and qudit USD+GE decoders. In Figure~\ref{plot:3,4_F3_r2_plot}, we show the performance comparison for (3,4) LDPC code with $r=2$ and blocklength $N=1600$. We observe that 
affine filtering+GE decoder outperforms both decoders.

\begin{figure}[!htbp]
    \centering
    \includegraphics[width=0.6\linewidth]{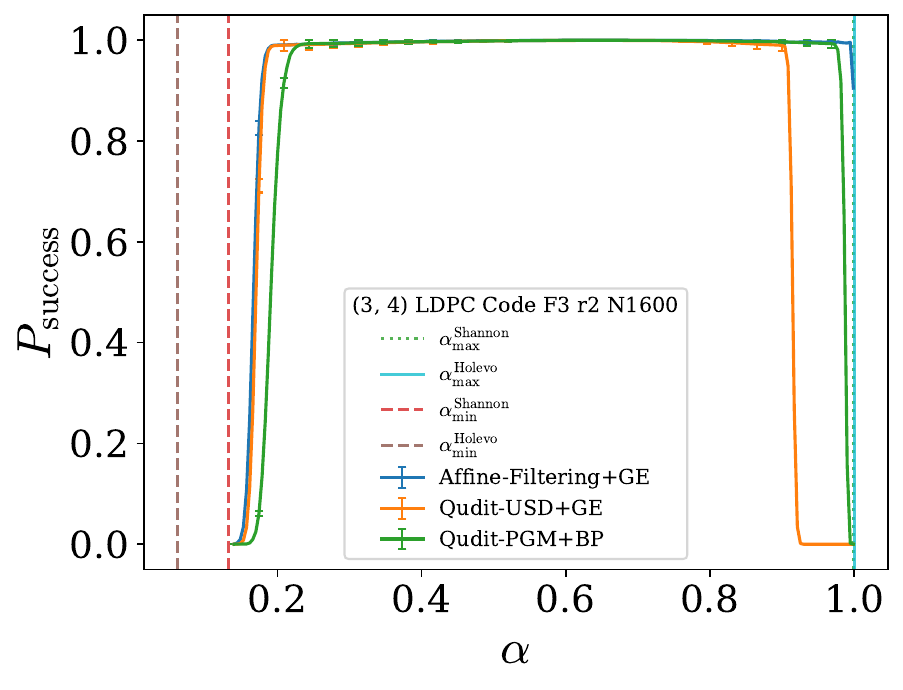}
    \caption{Performance Comparison for (3,4) Regular LDPC Codes on $\mathbb{F}_3$ and $r=2$ with $N=1600$}
    \label{plot:3,4_F3_r2_plot}
\end{figure}

\begin{figure}[!htbp]
    \centering
    \includegraphics[width=0.6\linewidth]{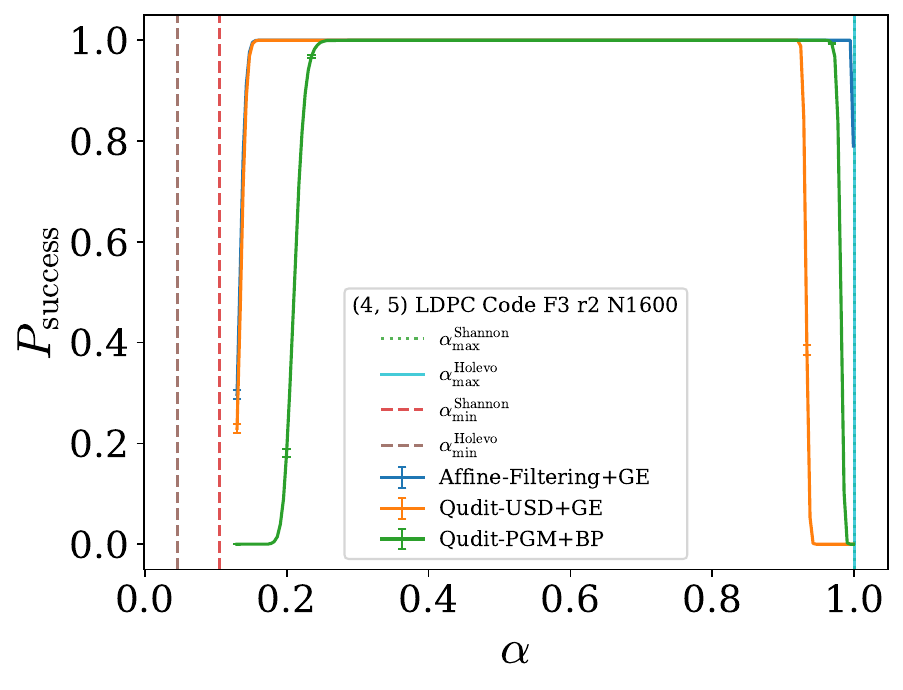}
    \caption{Performance Comparison for (4,5) Regular LDPC Codes on $\mathbb{F}_3$ and $r=2$ with $N=1600$}
    \label{plot:4,5_F3_r2_plot}
\end{figure}
\FloatBarrier
\section{Conclusion}

We studied affine filtering measurements, a structured class of unambiguous
measurements whose conclusive outcomes filter affine subspaces containing the transmitted codeword of a linear code. For symmetric codeword indexed states, we showed that the optimal affine filtering design
problem reduces from an SDP to an LP, enabling tractable computation of optimal
measurements for small local codes. Using these measurements, we proposed a
quantum decoding framework for pure state channels and demonstrated that
applying affine filtering to qudit states associated with SPC neighborhoods of
regular LDPC codes can improve decoding thresholds relative to qudit-wise USD
and qudit-wise PGM based decoders.

\section{Code and data availability}

The simulation code used to generate the numerical results in this paper is publicly available at \url{https://github.com/Aviemathelec1995/Affine_filtering_decoder}. The repository contains implementations of the LP computation for affine-filtering measurements over prime fields, the affine-filtering+GE decoder for LDPC codes with local single-parity-check constraints, the symbol-wise USD decoder, the PGM-based BP decoder, and scripts for reproducing the numerical comparisons reported in the paper. The code assumes that the alphabet size is prime. The LDPC simulations use regular parity-check matrices from Gallager ensembles, with the default affine-filtering experiments using binary-support single-parity-check constraints, as in the local measurement construction analyzed in this work.



\section{Author contributions}

Noah Shutty conceived the affine-filtering measurement formulation and the associated decoder. Avijit Mandal carried out the analysis of affine-filtering measurements and the SDP to LP reduction in Section~\ref{sec:affineUSD}. Henry D. Pfister contributed to the limiting LP formulation in Section~\ref{sec: rank deficient}. Avijit Mandal implemented the affine-filtering+GE decoder and produced the numerical results. Avijit Mandal wrote the first draft of the manuscript. Stephen P.~Jordan supervised the project. All authors reviewed, edited, and approved the final manuscript.

The authors used generative AI tools, including ChatGPT and Gemini, to brainstorm possible proof strategies and to assist with the presentation of the manuscript. The authors also used Codex to assist with code refactoring, debugging, documentation, and polishing of the public code repository and scripts used to generate the numerical results. All mathematical arguments, code, simulations, numerical outputs, and interpretations were checked and validated by the authors.

\bibliographystyle{abbrvnat}
\bibliography{ref}

@article{chailloux2024quantum,
  title={Quantum advantage from soft decoders},
  author={Chailloux, Andr{\'e} and Tillich, Jean-Pierre},
  journal={arXiv preprint arXiv:2411.12553},
  year={2024}
}

@article{chefles1998optimum,
  title={Optimum unambiguous discrimination between linearly independent symmetric states},
  author={Chefles, Anthony and Barnett, Stephen M},
  journal={Physics letters A},
  volume={250},
  number={4-6},
  pages={223--229},
  year={1998},
  publisher={Elsevier}
}

@article{chailloux2023quantum,
  title={The quantum decoding problem},
  author={Chailloux, Andr{\'e} and Tillich, Jean-Pierre},
  journal={arXiv preprint arXiv:2310.20651},
  year={2023}
}

@article{yamakawa2024verifiable,
  title={Verifiable quantum advantage without structure},
  author={Yamakawa, Takashi and Zhandry, Mark},
  journal={Journal of the ACM},
  volume={71},
  number={3},
  pages={1--50},
  year={2024},
  publisher={ACM New York, NY}
}

@article{buzet2025fine,
  title={Fine-Grained Unambiguous Measurements},
  author={Buzet, Quentin and Chailloux, Andr{\'e}},
  journal={arXiv preprint arXiv:2510.07298},
  year={2025}
}

@article{piveteau2025efficient,
  title={Efficient and optimal quantum state discrimination via quantum belief propagation},
  author={Piveteau, Christophe and Renes, Joseph M},
  journal={arXiv preprint arXiv:2509.19441},
  year={2025}
}

@article{mandal2024polar,
  title={Polar codes for cq channels: Decoding via belief-propagation with quantum messages},
  author={Mandal, Avijit and Brandsen, Sarah and Pfister, Henry D},
  journal={arXiv preprint arXiv:2401.07167},
  year={2024}
}

@inproceedings{brandsen2022belief,
  title={Belief propagation with quantum messages for symmetric classical-quantum channels},
  author={Brandsen, Sarah and Mandal, Avijit and Pfister, Henry D},
  booktitle={2022 IEEE Information Theory Workshop (ITW)},
  pages={494--499},
  year={2022},
  organization={IEEE}
}

@article{eldar2002quantum,
  title={On quantum detection and the square-root measurement},
  author={Eldar, Yonina C and Forney, G David},
  journal={IEEE Transactions on Information Theory},
  volume={47},
  number={3},
  pages={858--872},
  year={2002},
  publisher={IEEE}
}

@misc{PMBPQM_BSCQ,
  author       = {Mandal, Avijit and Pfister, Henry D},
  title        = {{PMBPQM\_BSCQ: Code repository for belief propagation with quantum messages
                  to find thresholds for regular LDPC codes over classical--quantum channels}},
  howpublished = {\url{https://github.com/Aviemathelec1995/PMBPQM_BSCQ}},
}

@book{luong2009fourier,
  title={Fourier analysis on finite Abelian groups},
  author={Luong, Bao},
  year={2009},
  publisher={Springer Science \& Business Media}
}

@Article{RU01,
  author = {Thomas J. Richardson and R{\"u}diger L. Urbanke},
  title = {The capacity of low-density parity-check codes under message-passing decoding},
  journal = {IEEE Transactions on Information Theory},
  volume = {47},
  number = {2},
  pages = {599--618},
  year = {2001}
}

@inproceedings{wilde2013towards,
  title={Towards Efficient Decoding of Classical-Quantum Polar Codes},
  author={Wilde, Mark M and Landon-Cardinal, Olivier and Hayden, Patrick},
  booktitle={8th Conference on the Theory of Quantum Computation, Communication and Cryptography},
  pages={157},
  year={2013}
}

@article{wilde2012polar,
  title={Polar codes for classical-quantum channels},
  author={Wilde, Mark M and Guha, Saikat},
  journal={IEEE Transactions on Information Theory},
  volume={59},
  number={2},
  pages={1175--1187},
  year={2012},
  publisher={IEEE}
}

@article{Renes-njp17,
  title={Belief propagation decoding of quantum channels by passing quantum messages},
  author={Renes, Joseph M},
  journal={New Journal of Physics},
  volume={19},
  number={7},
  pages={072001},
  year={2017},
  publisher={IOP Publishing},
  url = {http://arxiv.org/abs/1607.04833}
}

@article{holevo1998capacity,
  title={The capacity of the quantum channel with general signal states},
  author={Holevo, Alexander S},
  journal={IEEE Transactions on Information Theory},
  volume={44},
  number={1},
  pages={269--273},
  year={1998},
  publisher={IEEE}
}

@book{richardson2008modern,
  title={Modern coding theory},
  author={Richardson, Tom and Urbanke, Ruediger},
  year={2008},
  publisher={Cambridge university press}
}

@article{schumacher1997sending,
  title={Sending classical information via noisy quantum channels},
  author={Schumacher, Benjamin and Westmoreland, Michael D},
  journal={Physical Review A},
  volume={56},
  number={1},
  pages={131},
  year={1997},
  publisher={APS}
}

@article{krovi2015optimal,
  title={Optimal measurements for symmetric quantum states with applications to optical communication},
  author={Krovi, Hari and Guha, Saikat and Dutton, Zachary and da Silva, Marcus P},
  journal={Physical Review A},
  volume={92},
  number={6},
  pages={062333},
  year={2015},
  publisher={APS}
}

@article{da2013achieving,
  title={Achieving minimum-error discrimination of an arbitrary set of laser-light pulses},
  author={da Silva, Marcus P and Guha, Saikat and Dutton, Zachary},
  journal={Physical Review A},
  volume={87},
  number={5},
  pages={052320},
  year={2013},
  publisher={APS}
}

@article{gallager1962low,
  title={Low-density parity-check codes},
  author={Gallager, Robert},
  journal={IRE Transactions on information theory},
  volume={8},
  number={1},
  pages={21--28},
  year={1962},
  publisher={IEEE}
}

@article{hausladen1994pretty,
  title={A ‘pretty good’measurement for distinguishing quantum states},
  author={Hausladen, Paul and Wootters, William K},
  journal={Journal of Modern Optics},
  volume={41},
  number={12},
  pages={2385--2390},
  year={1994},
  publisher={Taylor \& Francis}
}

@article{holevo1978asymptotically,
  title={On asymptotically optimal hypotheses testing in quantum statistics},
  author={Holevo, Alexander Semenovich},
  journal={Teoriya Veroyatnostei i ee Primeneniya},
  volume={23},
  number={2},
  pages={429--432},
  year={1978},
  publisher={Russian Academy of Sciences, Steklov Mathematical Institute of Russian~…}
}

@article{eldar2003geometrically,
  title={Geometrically uniform frames},
  author={Eldar, Yonina C and Bolcskei, Helmut},
  journal={IEEE Transactions on Information Theory},
  volume={49},
  number={4},
  pages={993--1006},
  year={2003},
  publisher={IEEE}
}

@inproceedings{mandal2023belief,
  title={Belief-propagation with quantum messages for polar codes on classical-quantum channels},
  author={Mandal, Avijit and Brandsen, Sarah and Pfister, Henry D},
  booktitle={2023 IEEE International Symposium on Information Theory (ISIT)},
  pages={613--618},
  year={2023},
  organization={IEEE}
}

@article{jordan2025optimization,
  title={Optimization by decoded quantum interferometry},
  author={Jordan, Stephen P and Shutty, Noah and Wootters, Mary and Zalcman, Adam and Schmidhuber, Alexander and King, Robbie and Isakov, Sergei V and Khattar, Tanuj and Babbush, Ryan},
  journal={Nature},
  volume={646},
  number={8086},
  pages={831--836},
  year={2025},
  publisher={Nature Publishing Group UK London}
}

@article{mandal2026belief,
  title={Belief Propagation with Quantum Messages for Symmetric Q-ary Pure-State Channels},
  author={Mandal, Avijit and Pfister, Henry D},
  journal={arXiv preprint arXiv:2601.21330},
  year={2026}
}

@inproceedings{mandal2025reed,
  title={Reed--Muller Codes on CQ Channels via a New Correlation Bound for Quantum Observables},
  author={Mandal, Avijit and Pfister, Henry D},
  booktitle={2025 IEEE International Symposium on Information Theory (ISIT)},
  pages={1--6},
  year={2025},
  organization={IEEE}
}

@article{peres1988differentiate,
  title={How to differentiate between non-orthogonal states},
  author={Peres, Asher},
  journal={Physics Letters A},
  volume={128},
  number={1-2},
  pages={19},
  year={1988},
  publisher={Elsevier}
}

@article{ivanovic1987differentiate,
  title={How to differentiate between non-orthogonal states},
  author={Ivanovic, Igor D},
  journal={Physics Letters A},
  volume={123},
  number={6},
  pages={257--259},
  year={1987},
  publisher={Elsevier}
}

@article{dieks1988overlap,
  title={Overlap and distinguishability of quantum states},
  author={Dieks, Dennis},
  journal={Physics Letters A},
  volume={126},
  number={5-6},
  pages={303--306},
  year={1988},
  publisher={Elsevier}
}

@article{mandal2026qmp,
      title={Quantum Message Passing for Factor Graphs over Finite Abelian Groups}, 
      author={Avijit Mandal and Henry D. Pfister},
      journal={arXiv preprint arXiv:2604.12186},
      year={2026},
      url={https://arxiv.org/abs/2604.12186}
}

@article{shutty2026lqd,
      title={Optimization Using Locally-Quantum Decoders}, 
      author={Noah Shutty and Avijit Mandal and Seyoon Ragavan and Quentin Buzet and André Chailloux and Nicholas C. Rubin and Abid Khan and Sami Boulebnane and Ruslan Shaydulin and John Azariah and Stephen P. Jordan},
      year={2026},
      journal={arXiv preprint arXiv:2604.24633},
      url={https://arxiv.org/abs/2604.24633}, 
}

@inproceedings{chen2022quantum,
  title={Quantum algorithms for variants of average-case lattice problems via filtering},
  author={Chen, Yilei and Liu, Qipeng and Zhandry, Mark},
  booktitle={Annual international conference on the theory and applications of cryptographic techniques},
  pages={372--401},
  year={2022},
  organization={Springer}
}

@article{gisbert2019lipschitz,
  title={Lipschitz modulus of the optimal value in linear programming},
  author={Gisbert, Mar{\'\i}a Jes{\'u}s and C{\'a}novas, Mar{\'\i}a Josefa and Parra, Juan and Toledo, Fco Javier},
  journal={Journal of Optimization Theory and Applications},
  volume={182},
  number={1},
  pages={133--152},
  year={2019},
  publisher={Springer}
}
\appendix
\renewcommand{\thesection}{{\color{black}Appendix} \Alph{section}}
\section{Computing Expected Linear Equations using LP for Non-SPC Codes}

In Table~\ref{tab:LP results for Non SPC Codes}, we provide a comparison of the expected number of linear equations between qudit USD and  
affine filtering for several codes whose parity check matrix rank is more than 1, i.e. non-SPC codes. This suggests that we can apply 
affine filtering measurements on quantum states associated with nontrivial local codes of a large code and apply GE afterwards to achieve better decoding performance. This unlocks new possibilities to explore these measurements and come up with methods to design quantum decoders on PSCs.

\begingroup
\small
\setlength{\tabcolsep}{4pt}
\renewcommand{\arraystretch}{1.15}

\begin{longtable}{|>{\centering\arraybackslash}p{0.23\linewidth}|
                  >{\centering\arraybackslash}p{0.36\linewidth}|
                  >{\centering\arraybackslash}p{0.14\linewidth}|
                  >{\centering\arraybackslash}p{0.14\linewidth}|}

\hline
Parameter choice & Parity-check matrix $H$ for code $\localcode$ & Qudit-USD & Affine Filtering \\
\hline
\endfirsthead

\hline
Parameter choice & Parity-check matrix $H$ for code $\localcode$ & Qudit-USD & Affine Filtering \\
\hline
\endhead

\hline
\multicolumn{4}{r}{}
\endfoot

\hline
\endlastfoot

$q=2$, $\alpha=0.8$, $S=\{0\}$ &
$\begin{bmatrix}
1 & 1 & 1 & 0 \\
0 & 1 & 1 & 1
\end{bmatrix}$ &
$1.3376$ & $1.424$ \\
\hline

$q=2$, $\alpha=0.85$, $S=\{0\}$ &
$\begin{bmatrix}
1 & 1 & 1 & 0 \\
0 & 1 & 1 & 1
\end{bmatrix}$ &
$1.064$ & $1.167$ \\
\hline

$q=2$, $\alpha=0.9$, $S=\{0\}$ &
$\begin{bmatrix}
1 & 1 & 1 & 0 \\
0 & 1 & 1 & 1
\end{bmatrix}$ &
$0.7456$ & $0.848$ \\
\hline

$q=2$, $\alpha=0.8$, $S=\{0\}$ &
$\begin{bmatrix}
1 & 1 & 1 & 0 & 0 \\
0 & 1 & 1 & 1 & 0 \\
0 & 0 & 1 & 1 & 1
\end{bmatrix}$ &
$1.516$ & $1.568$ \\
\hline

$q=2$, $\alpha=0.85$, $S=\{0\}$ &
$\begin{bmatrix}
1 & 1 & 1 & 0 & 0 \\
0 & 1 & 1 & 1 & 0 \\
0 & 0 & 1 & 1 & 1
\end{bmatrix}$ &
$1.2419$ & $1.314$ \\
\hline

$q=2$, $\alpha=0.9$, $S=\{0\}$ &
$\begin{bmatrix}
1 & 1 & 1 & 0 & 0 \\
0 & 1 & 1 & 1 & 0 \\
0 & 0 & 1 & 1 & 1
\end{bmatrix}$ &
$0.894$ & $0.976$ \\
\hline

$q=2$, $\alpha=0.8$, $S=\{0\}$ &
$\begin{bmatrix}
1 & 1 & 1 & 1 & 0 \\
0 & 1 & 1 & 1 & 1
\end{bmatrix}$ &
$1.799$ & $1.92576$ \\
\hline

$q=2$, $\alpha=0.85$, $S=\{0\}$ &
$\begin{bmatrix}
1 & 1 & 1 & 1 & 0 \\
0 & 1 & 1 & 1 & 1
\end{bmatrix}$ &
$1.3962$ & $1.5020$ \\
\hline

$q=2$, $\alpha=0.9$, $S=\{0\}$ &
$\begin{bmatrix}
1 & 1 & 1 & 1 & 0 \\
0 & 1 & 1 & 1 & 1
\end{bmatrix}$ &
$0.957$ & $1.0236$ \\
\hline

$q=3$, $\alpha=0.8$, $S=\{0,1\}$ &
$\begin{bmatrix}
1 & 2 & 2 & 1 & 0 \\
0 & 1 & 2 & 2 & 1
\end{bmatrix}$ &
$2.55$ & $2.877$ \\
\hline

$q=3$, $\alpha=0.85$, $S=\{0,1\}$ &
$\begin{bmatrix}
1 & 2 & 2 & 1 & 0 \\
0 & 1 & 2 & 2 & 1
\end{bmatrix}$ &
$2.072$ & $2.7065$ \\
\hline

$q=3$, $\alpha=0.9$, $S=\{0,1\}$ &
$\begin{bmatrix}
1 & 2 & 2 & 1 & 0 \\
0 & 1 & 2 & 2 & 1
\end{bmatrix}$ &
$1.453$ & $2.4235$ \\
\hline

$q=3$, $\alpha=0.95$, $S=\{0,1\}$ &
$\begin{bmatrix}
1 & 2 & 2 & 1 & 0 \\
0 & 1 & 2 & 2 & 1
\end{bmatrix}$ &
$0.74525$ & $2$ \\
\hline

$q=5$, $\alpha=0.8$, $S=\{0,1,2\}$ &
$\begin{bmatrix}
1 & 2 & 3 & 4 & 0 \\
0 & 4 & 3 & 2 & 1
\end{bmatrix}$ &
$2.281$ & $2.8783$ \\
\hline

$q=5$, $\alpha=0.85$, $S=\{0,1,2\}$ &
$\begin{bmatrix}
1 & 2 & 3 & 4 & 0 \\
0 & 4 & 3 & 2 & 1
\end{bmatrix}$ &
$1.798$ & $2.749$ \\
\hline

$q=5$, $\alpha=0.9$, $S=\{0,1,2\}$ &
$\begin{bmatrix}
1 & 2 & 3 & 4 & 0 \\
0 & 4 & 3 & 2 & 1
\end{bmatrix}$ &
$1.233$ & $2.5449$ \\
\hline

$q=5$, $\alpha=0.95$, $S=\{0,1,2\}$ &
$\begin{bmatrix}
1 & 2 & 3 & 4 & 0 \\
0 & 4 & 3 & 2 & 1
\end{bmatrix}$ &
$0.6238$ & $2.3438$ \\
\hline

$q=5$, $\alpha=0.8$, $S=\{0,1,2,3\}$ &
$\begin{bmatrix}
1 & 2 & 3 & 0 & 0 \\
0 & 4 & 3 & 2 & 0 \\
0 & 0 & 1 & 1 & 1
\end{bmatrix}$ &
$2$ & $2$ \\
\hline

$q=5$, $\alpha=0.95$, $S=\{0,1,2,3\}$ &
$\begin{bmatrix}
1 & 2 & 3 & 0 & 0 \\
0 & 4 & 3 & 2 & 0 \\
0 & 0 & 1 & 1 & 1
\end{bmatrix}$ &
$1.1035$ & $1.9694$ \\
\hline
\caption{Optimal expected number of linear equations comparison}
       \label{tab:LP results for Non SPC Codes}
   
\end{longtable}

\endgroup
\end{document}